\documentclass[final, a4paper, UKenglish, cleveref, autoref, thm-restate, pdfa,numberwithinsect]{lipics-v2021}

\pdfoutput=1 %uncomment to ensure pdflatex processing (mandatory, e.g., to submit to arXiv)
\title{Positive Data Languages}

\author{Florian Frank}{Friedrich-Alexander-Universit\"at Erlangen-N\"urnberg, Germany}%
{florian.ff.frank@fau.de}%
{https://orcid.org/0000-0002-9458-3408}%
{Supported by Deutsche Forschungsgemeinschaft (DFG, German Research Foundation)
as part of the Research and Training Group 2475 ``Cybercrime and Forensic Computing'' (grant
number 393541319/GRK2475/1-2019).}
\author{Stefan Milius}%
{Friedrich-Alexander-Universit\"at Erlangen-N\"urnberg, Germany}%
{mail@stefan-milius.eu}%
{https://orcid.org/0000-0002-2021-1644}%
{Supported by the Deutsche Forschungsgemeinschaft (DFG, German Research Foundation) -- project number 419850228.}
\author{Henning Urbat}{Friedrich-Alexander-Universit\"at Erlangen-N\"urnberg, Germany}%
{henning.urbat@fau.de}%
{https://orcid.org/0000-0002-3265-7168}%
{Supported by the Deutsche Forschungsgemeinschaft (DFG, German Research Foundation) -- project number 470467389.}

\authorrunning{F.~Frank, S.~Milius, and H.~Urbat} %TODO mandatory. First: Use abbreviated first/middle names. Second (only in severe cases): Use the first author plus 'et al.'

\Copyright{Florian Frank, Stefan Milius and Henning Urbat} %TODO mandatory, please use full first names. LIPIcs license is "CC-BY";  http://creativecommons.org/licenses/by/3.0/

 %TODO mandatory; please add comma-separated list of keywords

\category{} %optional, e.g. invited paper

\nolinenumbers %uncomment to disable line numbering

\usepackage{amssymb,amsxtra,mathrsfs,mathtools,bm,stmaryrd}
\usepackage{thmtools}
\usepackage{xspace}
\usepackage{dsfont}
\usepackage{tikz}
\usepackage{tikz}
\usetikzlibrary{arrows,decorations.markings,decorations.pathreplacing,arrows.meta,fit,calc,positioning}
\usepackage{ifthen}
\usepackage[notref,notcite]{showkeys}
\usepackage[noadjust]{cite}
\usepackage[babel]{csquotes}
\usepackage{rotating}
\usepackage{tikz-cd}
\usepackage{blkarray}
\usepackage{adjustbox}
\usepackage{calc}

\tikzcdset{scale cd/.style={every label/.append style={scale=#1},
    cells={nodes={scale=#1}}}}

\usepackage{accents}

\declaretheorem[name=Definition,style=definition,numberwithin=section, sibling=definition]{defn}

\declaretheorem[name=Remark,style=definition,sibling=defn]{rem}
\declaretheorem[name=Assumptions,style=definition,sibling=defn]{assumptions}

\declaretheorem[name=Notation,style=definition,sibling=defn]{notation}

\declaretheorem[name=Construction,style=definition,sibling=defn]{construction}

\newcommand{\resetCurThmBraces}{%
\gdef\curThmBraceOpen{(}%
\gdef\curThmBraceClose{)}}
\resetCurThmBraces
\newcommand{\removeThmBraces}{%
\gdef\curThmBraceOpen{}%
\gdef\curThmBraceClose{}}
\resetCurThmBraces

\usepackage{etoolbox}
\patchcmd{\thmhead}{(#3)}{\curThmBraceOpen #3\curThmBraceClose }{}{}

\usepackage{hyperref}
\hypersetup{hidelinks,final,bookmarks}

\addto\extrasUKenglish{% only necessary if babel is used
  
}

\usepackage{seqsplit}
\usepackage{xstring}
\newcommand{\defaultshowkeysformat}[1]{%
\StrSubstitute{#1}{ }{\textvisiblespace}[\TEMP]%
\parbox[t]{\marginparwidth}{\raggedright\normalfont\small\ttfamily\(\{\){\color{red!50!black}\expandafter\seqsplit\expandafter{\TEMP}}\(\}\)}%
}

\renewcommand*\showkeyslabelformat[1]{%
\noexpandarg%
\defaultshowkeysformat{#1}%
}

\usepackage[footnote,marginclue,nomargin,draft]{fixme}

\newcommand{\mypar}[1]{%
  \vspace*{-0.75\baselineskip}
  \subparagraph*{#1}%
  }

\newcommand{\seq}{\subseteq}
\newcommand{\xra}[1]{\xrightarrow{~#1~}}

\newcommand{\xto}[1]{\xra{#1}}

\newcommand{\mar}{\ar[rightarrowtail]}
\newcommand{\ear}{\ar[two heads]}

\newcommand{\At}{\mathds{A}}
\newcommand{\Sigmas}{\Sigma^{\star}}
\newcommand{\Ats}{\At^{\!\raisebox{1pt}{\scriptsize$\star$}}}
\newcommand{\ol}{\overline}
\def\ul{\underline}

\newcommand{\op}[1]{\operatorname{\mathsf{#1}}}
\newcommand{\id}{\op{id}}

\newcommand{\inl}{\op{inl}}
\newcommand{\inr}{\op{inr}}

\newcommand{\abs}{\mathsf{abs}}

\newcommand{\states}{\mathsf{s}}
\newcommand{\al}{\mathsf{a}}
\newcommand{\trans}{\mathsf{t}}
\newcommand{\init}{\mathsf{i}}
\newcommand{\final}{\mathsf{f}}

\newcommand{\canmap}{\mathsf{can}}

\newcommand{\cat}[1]{\mathscr{#1}}

\def\C{\cat C}
\def\D{\cat D}

\newcommand{\II}{\mathbb{I}}
\newcommand{\FF}{\mathbb{F}}

\newcommand{\Set}{\mathbf{Set}}

\newcommand{\Pow}{\mathcal{P}}
\newcommand{\NAut}{\mathbf{NAut}}
\newcommand{\NAutfp}{\mathbf{NAut}_{\mathsf{fp}}}

\renewcommand{\epsilon}{\varepsilon}

\newcommand{\colim}{\mathop{\mathsf{colim}}}

\def\Eq{\mathbb{E}}

\newcommand{\reg}{{\mathrm{reg}}}

\newcommand{\Nom}{\mathbf{Nom}}
\newcommand{\RnNom}{\mathbf{RnNom}}
\DeclareMathOperator{\supp}{\mathsf{supp}}

\newcommand{\Perm}{\mathsf{Perm}}
\newcommand{\names}{\At}

\newcommand{\pow}{\mathcal{P}}

\newcommand{\fresh}{\mathbin{\#}}

\newcommand{\MSOep}{\text{MSO}^{\sim,+}}
\newcommand{\MSOe}{\text{MSO}^{\sim}}

\renewcommand{\Eq}{\mathsf{Eq}}
\renewcommand{\phi}{\varphi}

\newcommand{\hookto}{\hookrightarrow}

\newcommand{\epito}{\twoheadrightarrow}

\newcommand{\monoto}{\rightarrowtail}

\newcommand{\f}{\mathsf{f}}

\newcommand{\ini}{\iota}

\newcommand{\set}[1]{\{#1\}}
\newcommand{\setw}[2]{\{#1\,:\,#2\}}

\DeclareMathOperator{\im}{\textsf{im}}
\DeclareMathOperator{\coim}{\textsf{coim}}
\newcommand{\takeout}[1]{\empty}

\tikzset{shiftarr/.style={
        rounded corners,%
        to path={--([#1]\tikztostart.center)
                     -- ([#1]\tikztotarget.center) \tikztonodes
                     -- (\tikztotarget)},
}}

\tikzset{shiftarrr/.style={
        rounded corners,%
        to path={-- ([#1]\tikztostart.center)
                    |- (\tikztotarget)  \tikztonodes},
}}

\tikzset{roundcornerarr/.style={
        rounded corners,%
        to path={--([#1]\tikztostart.south)
                     |- (\tikztotarget) \tikztonodes},
}}

\tikzset{roundcornerarrr/.style={
        rounded corners,%
        to path={ -| (\tikztotarget) \tikztonodes},
}}

\tikzset{roundcornerarrrr/.style={
        rounded corners,%
        to path={ |- (\tikztotarget) \tikztonodes},
}}

\newcommand{\pullbackangle}[2][]{\arrow[phantom,to path={
                     -- ($ (\tikztostart)!1cm!#2:([xshift=8cm]\tikztostart) $)
                        node[anchor=west,pos=0.0,rotate=#2,
                        inner xsep = 0]
                        {\begin{tikzpicture}[minimum
                        height=1mm,baseline=0,#1]
    \draw[-] (0,0) -- (.5em,.5em) -- (0,1em);
                        \end{tikzpicture}}}]{}}

\newcommand{\overbar}[1]{\mkern 1.5mu\overline{\mkern-1.5mu#1\mkern-1.5mu}\mkern 1.5mu}

\newcommand{\mybar}[3]{%
  \mathrlap{\hspace{#2}\overline{\scalebox{#1}[1]{\phantom{\ensuremath{#3}}}}}\ensuremath{#3}
}

\newcommand{\barA}{{\mybar{0.6}{2.5pt}{A}}} % don't delete the extra brackets!
\newcommand{\barF}{\mybar{0.6}{2.5pt}{F}}
\newcommand{\barG}{\mybar{0.6}{2pt}{G}}
\newcommand{\barI}{\mybar{0.6}{2pt}{I}}
\newcommand{\barJ}{\mybar{0.6}{2pt}{J}}
\newcommand{\barE}{\mybar{0.6}{2.5pt}{E}}

\newcommand{\bark}{\mybar{0.7}{1.5pt}{k}}

\newcommand{\barL}{\mybar{0.8}{1.5pt}{L}}
\newcommand{\barR}{\mybar{0.8}{2pt}{R}}

\newcommand{\barQ}{\mybar{0.7}{2pt}{Q}}
\newcommand{\barDelta}{\mybar{0.6}{2pt}{\delta}}

\newcommand{\barGF}{\mybar{0.85}{2pt}{GF}}

\newcommand{\ngt}{{\mathsf{ngt}}}
\newcommand{\rg}{{\mathsf{rg}}}
\renewcommand{\ng}{{\mathsf{ng}}}
\newcommand{\fsuba}{{\mathsf{fsuba}}}

\newcommand{\Lan}{\textsf{Lan}}
\newcommand{\Run}{\textsf{AccRun}}

\newcommand{\Fin}{\mathsf{Fin}}

\newcommand{\Nat}{\mathds{N}}

\newcommand{\wordlang}[1]{\ensuremath{\mathsf{W}(#1)}}
\newcommand{\prelang}[1]{\ensuremath{\mathsf{P}(#1)}}

\ExplSyntaxOn
\NewDocumentCommand{\makecycle}{om}{
	\ensuremath{ \left(~\guest_print_list:nn { #2 } {~{\ }~}~\right)\IfValueT{#1}{^{#1}}}
}
\NewDocumentCommand{\maketuple}{om}{
	\ensuremath{ \left(~\guest_print_list:nn { #2 } {~{,\ }~}~\right)\IfValueT{#1}{^{#1}}}
}
\NewDocumentCommand{\makemonad}{om}{
	\ensuremath{ \left\langle~\guest_print_list:nn { #2 } {~{,\ }~}~\right\rangle\IfValueT{#1}{^{#1}}}
}

\seq_new:N \l_guest_list_seq
\cs_new_protected:Nn \guest_print_list:nn
{
	\seq_set_from_clist:Nn \l_guest_list_seq { #1 }
	\seq_use:Nn \l_guest_list_seq { #2 }
}
\ExplSyntaxOff

\newcommand{\functor}[7]{%
    {\setlength{\arraycolsep}{2pt}
    \begin{blockarray}{r@{\ \phantom{\{}}ccl}
        #1\colon & #2 & \longrightarrow & #3\\
		\begin{block}{r@{\ }\{ccl}
            & #4 & \longmapsto & \displaystyle#5 \\
            & #6 & \longmapsto & \displaystyle#7 \\
		\end{block}
    \end{blockarray}
    }
}

\numberwithin{equation}{section}

\begin{document}

\FXRegisterAuthor{sm}{asm}{SM}%Stefan
\FXRegisterAuthor{hu}{ahu}{HU}%Hennig
\FXRegisterAuthor{ff}{aff}{FF}%Flo

\maketitle

\begin{abstract}
Positive data languages are languages over an infinite alphabet closed under possibly non-injective renamings of data values.
Informally, they model properties of data words expressible by assertions about equality, but not inequality, of data values occurring in the word. 
We investigate the class of positive data languages recognizable by nondeterministic orbit-finite nominal automata, an abstract form of register
automata introduced by Boja\'nczyk, Klin, and Lasota. As our main contribution we provide a number of equivalent characterizations of that class
in terms of positive register automata, monadic second-order logic with positive equality tests, and finitely presentable nondeterministic automata
in the categories of nominal renaming sets and of presheaves over finite sets.
\end{abstract}

\section{Introduction}\label{sec:intro}
Automata over infinite alphabets provide a simple computational model for reasoning about structures involving \emph{data} such as nonces~\cite{KurtzEA07}, URLs~\cite{BieleckiEA02}, or
values in XML documents~\cite{nsv04}. Consider, for instance, the (infinite) set  $\At$ of admissible user IDs for a server. The sequence of all user logins within a given time period then forms a finite word $a_1\cdots a_n\in \Ats$ over the infinite alphabet $\At$, and behaviour patterns may be modelled as data languages over $\At$, e.g.
\begin{align*}
 L_0 &= \{\,a_1\cdots a_n\in \Ats\mid \text{$a_i\neq a_n$ for all $i<n$}\,\} && \text{(``last user has not logged in before'')}, \\
 L_1 &= \{\,a_1\cdots a_n\in \Ats \mid \text{$a_i=a_j$ for some $i\neq j$}\,\} &&  \text{(``some user has logged in twice'')}. 
\end{align*}
Both $L_0$ and $L_1$ involve assertions about equality, or inequality,
of data values (here, user IDs). However, asserting \emph{inequality}
is sometimes considered problematic and thus undesired. For example,
since users may have multiple IDs, a logfile $a_1\ldots a_n\in L_0$
does not actually guarantee that the last user has not logged in
before. In contrast, if $a_1\ldots a_n\in L_1$, then it is guaranteed
that some user has indeed logged in twice. The structural difference
between the two languages is that $L_1$ is closed under arbitrary
renamings $\rho\colon \At\to \At$ (i.e.~$a_1\cdots a_n\in L_1$
implies $\rho(a_1)\cdots \rho(a_n)\in L_1$), taking into account
possible identification of data values, while~$L_0$ is only closed
under {injective} (equivalently bijective) renamings. We refer to
languages with the former, stronger closure property as \emph{positive
  data languages}. Intuitively, such languages model properties of
data words expressible by positive statements about equality of data
values. It is one of the goals of our paper to turn this into a theorem.

For that purpose, we build on the abstract account of data languages and their automata based on the theory of nominal sets~\cite{gp99,Pitts2013}, initiated by the work of Boja\'nczyk, Klin, and Lasota~\cite{BojanczykEA14}. Specifically, we investigate \emph{nondeterministic orbit-finite nominal automata} (\emph{NOFA}), the nominal version of classical nondeterministic finite automata. We approach the class of \emph{positive} NOFA-recognizable data languages from several different perspectives, ranging from concrete to more abstract and conceptual, and establish the equivalent characterizations summarized in \autoref{fig:equivalences}. In more detail, our main contributions are as follows.

\begin{figure}
\begin{tikzpicture}[label distance=-1.5mm]
\node[line width=2mm,draw=lipicsLightGray,align=center,fill=lipicsLightGray,rounded
corners] (center) at (0,1.5) (a) {Positive RA}; 
\node[line width=2mm,draw=lipicsLightGray,align=center,fill=lipicsLightGray,rounded
corners] (center) at (5,0) (b) {$\MSOep$ $\cap$ NOFA};
\node[line width=2mm,draw=lipicsLightGray,align=center,fill=lipicsLightGray,rounded
corners] (center) at (0,0) (c) {FSUBA};  
\node[line width=2mm,draw=lipicsYellow,align=center,fill=lipicsYellow,rounded
corners] (center) at (5,1.5) (d) {\textbf{Positive $\boldsymbol{\cap}$ NOFA}};
\node[line width=2mm,draw=lipicsLightGray,align=center,fill=lipicsLightGray,rounded
corners] (center) at (10,1.5) (e) {NOFRA};
\node[line width=2mm,draw=lipicsLightGray,align=center,fill=lipicsLightGray,rounded
corners] (center) at (10,0) (f) {Super-finitary $\Set^\FF$-aut.};
\draw[implies-implies,double equal sign distance] (a)--(d) node [midway, label=above:{Thm.~\ref{thm:pos-nofa-vs-pos-reg}}] {};
\draw[implies-implies,double equal sign distance] (a)--(c) node [midway, label=right:{Rem.~\ref{rem:pos-reg-aut-vs-fsuba}}] {};
\draw[implies-implies,double equal sign distance] (d)--(b) node [midway, label=left:{Thm.~\ref{thm:mso-vs-nofa}}] {};
\draw[implies-implies,double equal sign distance] (d)--(e) node [midway, label=above:{Thm.~\ref{thm:nofa-vs-nofra}}] {};
\draw[implies-implies,double equal sign distance] (d.south east) ++(-0.75,0) -- ++(0,-0.35) -| (f) node [pos=0.25, label={[label distance=-1.5mm]above:{Thm.~\ref{thm:presheaf-automata-vs-nofa}}}] {};
\end{tikzpicture}
\caption{Equivalent characterizations of positive NOFA-recognizable languages}\label{fig:equivalences}
\end{figure}

\enlargethispage{\baselineskip}
%\vspace*{-0.75\baselineskip}
%\subparagraph*{Register automata.}
\mypar{Register automata.}
NOFAs are known to be expressively equivalent to register
automata \mbox{\cite{KaminskiFrancez94,kz10}}, i.e.~finite automata
that can memorize data values using a fixed number of registers and
test the input for (in)equality with previously stored
values. Restricting transitions to positive equality tests leads to
\emph{positive register automata}, which correspond to \emph{finite-state
  unification-based automata} (\emph{FSUBA})~\cite{tal99,kt06} and are shown to capture precisely positive 
  NOFA-recognizable languages (\autoref{thm:pos-nofa-vs-pos-reg} and \autoref{rem:pos-reg-aut-vs-fsuba}).
  On the way, we isolate a remarkable property of this language class: while NOFAs generally require the
  ability to guess data values during the computation to reach their full expressive strength, guessing
  and non-guessing NOFA are equivalent for positive data languages (\autoref{thm:positive-non-guessing}). 
%\vspace*{-0.75\baselineskip}
%\subparagraph*{Monadic second-order logic.}
\mypar{Monadic second-order logic.}
As illustrated above, positive data languages model (only) positive assertions about
the equality of data values. To substantiate this intuition, we employ monadic second-order logic ($\MSOe$)
over data words~\cite{nsv04, boj-13,col-15}, an extension of classical MSO with equality tests for data values,
and consider its restriction $\MSOep$ to positive equality tests. While this logic is more expressive than NOFA,
we show that within the class of NOFA-recognizable languages it models exactly the positive languages (\autoref{thm:mso-vs-nofa}).
%\vspace*{-0.75\baselineskip}
%\subparagraph*{Categorical perspective.}
\mypar{Categorical perspective.}
The classical
notion of nondeterministic finite automata can be categorified by replacing
the finite set of states with a finitely presentable object of a
category $\C$. For example, NOFAs are precisely nondeterministic
$\C$-automata for $\C =$ nominal sets. Apart from the latter category,
several other toposes have been proposed as abstract foundations for
reasoning about {names} (data values), most prominently the category
of \emph{nominal renaming sets}~\cite{GH08}, the category $\Set^\II$
of presheaves over finite sets and injective maps~\cite{stark96}, and
the category $\Set^\FF$ of presheaves over finite sets and all maps
(equivalently, finitary set functors)~\cite{ftp99}. It is thus natural
to study nondeterministic automata in the latter three categories,
viz.\ \emph{nondeterministic orbit-finite renaming automata}
(\emph{NOFRA}), \emph{nondeterministic super-finitary
  $\Set^\II$-automata} and \emph{nondeterministic super-finitary
  $\Set^\FF$-automata}. Our final contribution is a classification of
their expressive power: we show that $\Set^\II$-automata are
equivalent to NOFAs, while both NOFRAs and $\Set^\FF$-automata capture
positive NOFA-recognizable languages %
\renewcommand{\theoremautorefname}{Theorems}%
(\autoref{thm:nofa-vs-nofra} and
\ref{thm:presheaf-automata-vs-nofa}).
\renewcommand{\theoremautorefname}{Theorem}%
%This proves in particular that the choice of categorical foundation
%(nominal sets vs.\ presheaves) is inessential for automata theory.
Hence, both nominal and presheaf-based automata are able to recognize
positive and all NOFA-recognizable languages, respectively.

\mypar{Acknowledgement.} The authors wish to thank Bartek Klin for pointing out the example in \autoref{rem:dofa}.

\section{Nominal Automata and Positive Data
  Languages}\label{sec:nofa-nofra}

For the remainder of the article, we fix a countably infinite set $\At$ of \emph{data
  values}, a.k.a.~\emph{names} or \emph{atoms}. The goal is to study positive data
languages, that is, languages of finite words over $\At$ closed under
arbitrary renamings. This is achieved via the framework of nominal
(renaming) sets~\cite{gp99,Pitts2013,GH08}.

\subsection{Nominal Sets and Nominal Renaming Sets}\label{sec:nominal-sets}
A \emph{renaming} is a finite map $\rho\colon \At\to\At$, that is, $\rho(a)=a$ for all
but finitely many $a\in \At$. We let $\Fin(\At)$ denote the monoid of renamings, with
multiplication given by composition, and $\Perm(\names)$ its subgroup given by \emph{finite
permutations}, i.e.~bijective renamings. For $M\in \{\, \Perm(\At), \Fin(\At)\,\}$ an
\emph{$M$-set} is a set $X$ equipped with a monoid action $M \times X \to X$,
denoted $(\rho,x)\mapsto \rho\cdot x$. A subset $S\seq \At$ is a \emph{support} of $x\in X$ if
for every $\rho,\sigma\in M$ such that $\rho|_S=\sigma|_S$ one has $\rho\cdot x = \sigma\cdot x$.
Informally, consider $X$ as a set of syntactic objects (e.g.~words, trees, $\lambda$-terms)
whose description may involve free names from $S$. A \emph{nominal $M$-set} is an $M$-set where
every element $x$ has a finite support. This implies that $x$ has a least finite support
$\supp x\seq \At$. A name $a \in \names$ is \emph{fresh} for $x$, denoted $a \fresh x$, if $a \notin \supp x$.

Nominal $\Perm(\At)$-sets are called \emph{nominal sets}, and nominal $\Fin(\At)$-sets are called
\emph{nominal renaming sets}. A nominal renaming set $X$ can be regarded as a nominal set by
restricting its $\Fin(\At)$-action to a $\Perm(\At)$-action. The least supports of an element $x\in X$
w.r.t.\ both actions coincide~\cite[Thm.~4.8]{gabbay07}, so the notation $\supp x$ is unambiguous.  

A subset $X$ of a nominal $M$-set $Y$ is \emph{$M$-equivariant} if $\rho\cdot x\in X$ for all
$x\in X$ and $\rho\in M$. More generally, a map $f\colon X \to Y$ between nominal $M$-sets is
\emph{$M$-equivariant} if $f(\rho \cdot x) = \rho \cdot f(x)$ for all $x \in X$ and $\rho \in M$. 
This implies $\supp f(x) \subseteq \supp x$ for all $x \in X$. 

We write $X\times Y$ for the cartesian product of nominal $M$-sets with componentwise action, and
$\coprod_{i\in I} X_i$ for the coproduct (disjoint union) with action inherited from the summands.

Given a nominal set $X$, the \emph{orbit} of an element $x\in X$ is the set  $\setw{\pi\cdot x}{\pi \in \Perm(\names)}$.
The orbits form a partition of $X$. A nominal set is \emph{orbit-finite}
if it has only finitely many orbits. A nominal renaming set is \emph{orbit-finite} if it is orbit-finite as a nominal set.

\begin{example}
The set $\At$ with the $\Fin(\At)$-action $\rho\cdot a = \rho(a)$ is a nominal renaming set, as is
the set $\Ats$ of finite words over $\At$ with $\rho\cdot w = \rho^\star(w)=\rho(a_1)\cdots \rho(a_n)$
for $w=a_1\cdots a_n$. The least support of $a_1\cdots a_n\in \Ats$ is the set $\{a_1,\ldots,a_n\}$.
The set $\Ats$ has infinitely many orbits; its equivariant subsets $\At^n$ (words of a fixed length $n$)
are orbit-finite. For instance, $\At^2$ has the two orbits $\{aa: a\in \At\}$ and $\{ab: a\neq b\in \At\}$.
An example of a nominal set that is not a renaming set is $\At^{\#n} = \{\,a_1\ldots a_n : a_i\neq a_j \text{ for $i\neq j$} \,\}$
with pointwise $\Perm(\At)$-action.
\end{example}

A nominal set $X$ is \emph{strong} if, for every $x\in X$ and $\pi\in \Perm(\At)$, one has $\pi\cdot x = x$
if and only if $\pi$ fixes every element of $\supp(x)$. (The `if' statement holds in every nominal set.)
For instance, the nominal sets $\At^{\#n}$, $\At^n$ and $\Ats$ are strong. Up to isomorphism,
(orbit-finite) strong nominal sets are precisely (finite) coproducts $\coprod_{i\in I} \At^{\# n_i}$ where
$n_i\in \Nat$. For every orbit-finite nominal set $X$, there exists a surjective $\Perm(\At)$-equivariant
map $e\colon Y\epito X$ for some orbit-finite strong nominal set $Y$ (see e.g.\ \cite[Cor.~B.27]{mil-urb-19}).
In fact, if $o$ is the number of orbits of $X$, one may take $Y=J\times \At^{\#n}$ where $J=\{1,\ldots,o\}$
and $n=\max_{x\in X} |{\supp x}|$. We refer the reader to \cite[Sec.~4.1]{gmm06} and \cite[Thm.~10.9]{BojanczykEA14} for more details on representing orbit-finite nominal sets.

\subsection{Nominal Automata and Nominal Renaming Automata}
The object of interest in this paper is data languages $L\seq \Ats$ closed under renamings:

\begin{defn}
\begin{enumerate}
\item A {data language} $L\seq \Ats$ is \emph{positive} if it is $\Fin(\At)$-equivariant.
\item The \emph{positive closure} of $L\seq \Ats$ is given by
$\barL \;=\; \{\, \rho^\star(w): w\in L,\, \rho\in \Fin(\At) \,\}$.
\end{enumerate}
\end{defn}

A natural automata model for data languages is given by nondeterministic orbit-finite automata~\cite{BojanczykEA14}
over nominal sets and their restriction to nominal renaming sets:
\begin{defn}\label{def:nofa} Let $M\in \{\, \Perm(\At), \Fin(\At) \,\}$.
\begin{enumerate}
\item A \emph{nondeterministic orbit-finite $M$-automaton} $A=(Q,\delta,I,F)$ consists of an orbit-finite
nominal $M$-set $Q$ of states, an $M$-equivariant transition relation $\delta\seq Q\times \At\times Q$,
and $M$-equivariant subsets $I,F\seq Q$ of initial and final states. Nominal orbit-finite $M$-automata
are called \emph{nondeterministic orbit-finite automata} (\emph{NOFA}) for $M=\Perm(\At)$ and
\emph{nondeterministic orbit-finite renaming automata} (\emph{NOFRA}) for $M=\Fin(\At)$.
\item Given a nominal orbit-finite $M$-automaton $A$, we write $q\xto{a} q'$ if $q'\in \delta(q,a)$.
A \emph{run} of $A$ on input $w=a_1\cdots a_n\in \Ats$ is a sequence $(q_0,a_1,q_1,a_2,\ldots,a_n,q_n)$
such that $q_0\in I$ and $q_r\xto{a_{r+1}} q_{r+1}$ for $0\leq r<n$. The run is \emph{accepting} if $q_n\in F$. The automaton
$A$ \emph{accepts} the word $w$ if $A$ admits an accepting run on input $w$. The \emph{accepted language} $L(A)\seq \Ats$ is
the set of all accepted words. A data language is \emph{NOF(R)A-recognizable} if some NOF(R)A accepts it.
\end{enumerate}
\end{defn}

For example, the languages $L_0$ and $L_1$ from the Introduction are NOFA-recognizable.
\begin{rem}
\begin{enumerate}
\item The restriction to the input alphabet $\At$ is for simplicity: all our results extend to alphabets
$\Sigma=\Sigma_0\times \At$ for a finite set $\Sigma_0$, i.e.~finite coproducts of copies of $\At$. 
\item Another use of nominal renaming sets in automata theory appears in the work by Moerman and Rot~\cite{mr20}
on deterministic nominal automata with outputs. The restrictions of their model make it unsuitable for language
recognition~\cite[Rem.~4.1]{mr20} but allow for a succinct representation of computed maps via \emph{separating automata}. 
\end{enumerate}
\end{rem}

To relate the expressive power of NOFA and NOFRA, we start with a simple observation:

\begin{proposition}\label{prop:nofra-acc-pos-lang}
Every NOFRA accepts a positive language.
\end{proposition}

The converse (\autoref{thm:nofa-vs-nofra}) needs an automata-theoretic construction of the closure of a language.
To this end, we first turn the states of a NOFA into a sort of normal form.

\begin{rem}[cf.~\cite{BojanczykEA14}]\label{rem:strong}
Every NOFA $A=(Q,\delta,I,F)$ is equivalent to one whose nominal set of states is of the form $J\times \At^{\#m}$
for some finite set $J$ and $m\in \Nat$. Indeed, choose a nominal set $Q'=J\times \At^{\#m}$ and an equivariant
surjection $e\colon Q'\epito Q$ (see \autoref{sec:nominal-sets}), and consider the NOFA $A'=(Q',\delta',I',F')$
whose structure is given by the preimages 
\[\delta'=(e\times \id_\At\times e)^{-1}[\delta],\qquad I'=e^{-1}[I],\qquad F'=e^{-1}[F].\]
It is not difficult to verify that $L(A')=L(A)$; see also \autoref{prop:epi-pres-lang}. Note that in a NOFA with states $J\times \At^{\#m}$, the equivariant sets of initial and final states are of the form $I=J_I\times \At^{\#m}$ and $F=J_F\times \At^{\#m}$ for some $J_I,J_F\seq J$.
\end{rem}

\begin{construction}[Positive Closure of a NOFA]\label{cons:bar-A}
  Let $A=(Q,\delta,I,F)$ be a NOFA with states $Q=J\times \At^{\#m}$
  (cf.\ \autoref{rem:strong}). The NOFRA
  $\barA=(\barQ,\barDelta,\barI, \barF)$ is given by the states
  $\barQ=J\times \At^{m}$, initial states $\barI=J_I\times \At^m$,
  final states $\barF=J_F\times \At^m$, and transitions
  \[ 
    \barDelta \;=\; \{\, (j,\rho^\star p)\xto{\rho a} (j',\rho^\star p') : \text{$(j,p)\xto{a}(j',p')$ in $A$ and $\rho\in\Fin(\At)$}  \,\}.
  \]
\end{construction}

\begin{proposition}\label{prop:barA-acc-barL}
The NOFRA $\barA$ accepts the positive closure of the language of $A$.
\end{proposition}
The proof of $L(\barA)\seq \ol{L(A)}$ is slightly subtle since the
transitions of a run in $\barA$ may be induced by different $\rho$'s;
some bookkeeping and sensible choice of fresh names ensures
com\-pa\-ti\-bi\-lity.

Now we come to our first characterization of positive NOFA-recognizable languages:
\begin{theorem}\label{thm:nofa-vs-nofra}
A language is positive and NOFA-recognizable iff it is NOFRA-recognizable.
\end{theorem}
\noindent
Indeed, the ``if'' direction holds due to \autoref{prop:nofra-acc-pos-lang} and 
because every NOFRA is a NOFA. The ``only if'' direction follows from \autoref{prop:barA-acc-barL}, using that $\barL =L$ for positive $L$.

\begin{rem}\label{rem:dofa}
A NOF(R)A is \emph{deterministic}, and hence called a \emph{DOF(R)A}, if it admits a single initial state and its transition relation
is a function $\delta\colon Q\times \At\to Q$. In contrast to classical finite automata, DOFAs are less expressive that NOFAs~\cite{BojanczykEA14}.
We leave it as an open problem whether \autoref{thm:nofa-vs-nofra} restricts to DOF(R)As. In this regard, observe that \autoref{cons:bar-A}
produces a \emph{non}deterministic automaton $\barA$ even if the given automaton $A$ is deterministic. Computing the positive closure of a
DOFA-recognizable language necessarily requires the introduction of nondeterminism, as illustrated by the following example due to Bartek Klin
(personal communication). Consider the language $L$ consisting of all words whose last letter appears immediately before the last occurrence of
a repeated letter; that is, words of the form $vabbwa$ where (i) $v,w\in \Ats$ and $a,b\in \At$, (ii) any two consecutive letters in $w$ are
distinct, (iii) the first letter of $w$ is distinct from $b$ and (iv) the last letter of $w$ is distinct from $a$. This language is recognizable
by a DOFA, in fact by an orbit-finite nominal monoid~\cite{boj-13}. Its positive closure $\barL$ consists of all words whose last letter appears
immediately before \emph{some} occurrence of a repeated letter, which is not DOFA-recognizable.
\end{rem}

\subsection{Abstract Transitions and Runs}
\renewcommand{\sectionautorefname}{Sections}%
\autoref{sec:reg-aut} and \ref{sec:mso}
\renewcommand{\sectionautorefname}{Section}%
will relate positive NOFA-recognizable
languages to register automata and monadic second-order logic. This
relies on a presentation of transitions of $\barA$ in terms of
abstract transitions, given by equations involving
register entries and input values.

\begin{defn}\label{def:nofa-abstract}
Let $A=(Q,\delta,I,F)$ and $\barA=(\barQ,\barDelta,\barI,\barF)$ be as in \autoref{cons:bar-A}.
\begin{enumerate}
\item\label{def:nofa-abstract:1} An \emph{equation} is an expression of the form $k=\bullet$, $\bullet=k$ or $k=\bark$, where $k,\bark\in \{1,\ldots,m\}$. 
\item An \emph{abstract transition} is a triple $(j,E,j')$ where $j,j'\in J$ and $E$ is a set of equations.
\item Every triple $((j,p),a,(j',p'))\in Q\times \At\times Q$ induces an abstract transition $(j,E,j')$ defined as follows for $k,\bark\in \{1,\ldots,m\}$ (we write $(-)_i$ for the $i$-th letter of a word):
\[ k=\bullet\in E \iff p_k=a,\qquad \bullet=k\in E \iff a=p'_k,\qquad k=\bark\in E \iff p_k=p_{\overbar{k}}'. \]
We let $\abs(\delta)$ denote the set of abstract transitions induced by transitions in $\delta$, and we write $j\xto{E} j'$ for $(j,E,j')\in \abs(\delta)$.
\item A triple $((j,q),b,(j',q'))\in \barQ\times \At\times \barQ$ is \emph{consistent} with the abstract transition $(j,E,j')$ if for every $k,\bark\in \{1,\ldots,m\}$ the following conditions hold:
\[
k=\bullet \in E \implies q_k=b,\qquad
\bullet = k \in E\implies b=q'_k,\qquad
k=\bark \in E \implies q_k=q'_{\overbar{k}}.
\]
\end{enumerate}
\end{defn}

\begin{proposition}\label{prop:trans-bar-A}
For every triple $((j,q),b,(j',q'))\in \barQ\times \At\times \barQ$, we have
\[  (j,q)\xto{b} (j',q') \text{ in $\barA$} \qquad\text{iff}\qquad \text{$((j,q),b,(j',q'))$ is consistent with some $(j,E,j')\in \abs(\delta)$}.  \] 
\end{proposition}

\begin{defn} An \emph{abstract run} in $\barA$ is a sequence $(j_0,E_1,j_1,E_2,j_2,\ldots, E_n,j_n)$ such that $j_0\in J_I$ and $j_{r-1}\xto{E_r} j_r$ for $r=1,\ldots,n$. It is \emph{accepting} if $j_n\in J_F$.
\end{defn}

\begin{notation}\label{not:xik-pred}
Given an abstract run $(j_0,E_1,j_1,E_2,j_2,\ldots, E_n,j_n)$, we inductively define the predicates
$\Eq^{(i)}_k$ ($i\in \{1,\ldots,n\}$, $k\in \{1,\ldots,m\}$) on the set $\{1,\ldots,n\}$:
\begin{enumerate}
\item if $\bullet = k$ in $E_i$ then $\Eq^{(i)}_k(i)$;
\item if $r<n$ and $k=\bark$ in $E_{r+1}$ and $\Eq^{(i)}_k(r)$ then $\Eq^{(i)}_{\overbar{k}}(r+1)$. 
\end{enumerate} 
\end{notation}
Informally, $\Eq^{(i)}_k(r)$ asserts that $1\leq i\leq r\leq n$ and that in every run in $\barA$ of length $r$
whose transitions are consistent with $E_1,\ldots,E_r$, the $i$-th input letter equals the content of register
$k$ after $r$ steps. The accepted language may be characterized using these predicates:

\begin{proposition}\label{prop:abstract-lang-char}
  The NOFRA $\barA$ accepts the word $b_1\cdots b_n\in \Ats$ iff there
  exists an accepting abstract run of length $n$ (with induced
  predicates $\Eq_k^{(i)}$) such that for $i,r\in \{1,\ldots,n\}$,
  \begin{align}
    \text{$r<n$ and $k=\bullet$ in $E_{r+1}$ and $\Eq^{(i)}_k(r)$ for some $k$}
    \quad&\implies\quad
    b_i=b_{r+1}. \label{eq:eq-cond-1}
    % \text{$\bullet=k$ in $E_r$ and $\Eq^{(i)}_k(r)$ for some $k$} \quad&\implies\quad b_i=b_r. \label{eq:eq-cond-2}
  \end{align}
\end{proposition}
As a first application of this result, we identify an important
difference between NOFA and NOFRA concerning the power of guessing
data values during the computation:

\begin{defn}\label{def:non-guessing}
  A NOFA/NOFRA is \emph{non-guessing} if each initial state has empty
  support and for each transition $q\xto{a} q'$ one has
  $\supp q' \seq \supp q \cup \{a\}$.
\end{defn}
The NOFA-recognizable language $L_0$ from the Introduction is not
recognizable by any non-guessing
NOFA~\cite[Ex.~1]{KaminskiFrancez94}. Note that $L_0$ is not positive;
in fact, it is necessarily so, since for positive languages guessing
does not add to the expressive power of automata:

\begin{theorem}\label{thm:positive-non-guessing}
  Every positive NOFA-recognizable language is accepted by some
  non-guessing NOFRA, hence by some non-guessing NOFA.
\end{theorem}
To make a NOFRA non-guessing, one keeps track (via the state) of those registers containing data
values forced by abstract transitions. The other registers then may be
modified arbitrarily, which allows the elimination of guessing transitions.

\section{Positive Register Automata}\label{sec:reg-aut}
We now relate positive NOFA-recognizable languages to register automata, a.k.a.~finite-memory automata, originally introduced by
Kaminski and Francez~\cite{KaminskiFrancez94}; we follow the equivalent presentation by Boja\'nczyk et al.~\cite{BojanczykEA14}.
A \emph{register automaton} is a quintuple $A=(C,m,\delta,I,F)$ where $C$ is a finite set of control states, $m\in \Nat$ is the
number of registers (numbered from $1$ to $m$), $I,F\seq C$ are sets of initial and final states, and
$\delta\seq C\times \mathrm{Bool}(\Phi) \times C$ is the set of transitions. Here, $\mathrm{Bool}(\Phi)$ denotes the set of boolean formulas over the atoms
$\Phi\;=\;(\,\{1,\ldots,m\}\times \{\mathrm{before}\} \cup \{\bullet\} \cup \{1,\ldots,m\}\times \{\mathrm{after}\}\,)^2$.
Elements of $\Phi$ are called \emph{equa\-tions}; we write $x=y$ for $(x,y)\in \Phi$. Moreover, we denote $(c,\phi,c')\in \delta$ by $c\xto{\phi} c'$.
A \emph{configuration} of $A$ is a pair $(c,r)$ of a state $c\in C$ and a word $r\in (\At\cup \{\bot\})^m$ corresponding to a partial assignment of data values to the registers. The initial configurations are $(c,\bot^m)$ for $c\in I$. Given an input $a\in \At$ and configurations $(c,r),(c',r')$ we write $(c,r)\xto{a}(c',r')$ if this move is consistent with some transition $c\xto{\phi} c'$, that is, the formula
$\phi$ is true under the assignment making an atom $x=y\in \Phi$ true iff the corresponding data values are defined and equal. For instance, $(k,\mathrm{before})=\bullet$ is true iff $r_k\neq \bot$ and $r_{k}=a$, and
$(k,\mathrm{before})=(\bark,\mathrm{after})$ is true iff $r_k,r'_{\overbar{k}}\neq \bot$ and $r_k=r'_{\overbar{k}}$.
A word $w=a_1\ldots a_n\in \Ats$ is \emph{accepted} by $A$ if it admits an accepting run, viz.\ a sequence of moves $(c_0,r_0)\xto{a_1} (c_1,r_1)\xto{a_2} \cdots \xto{a_{n}} (c_n,r_n)$ where $(c_0,r_0)$ is initial and $c_n\in F$. The \emph{accepted language}  $L(A)\seq \Ats$ is the set of accepted words. 

As shown by Boja\'nczyk et al.~\cite{BojanczykEA14}, register automata accept the same languages as NOFAs.
To capture positive languages, we restrict to register automata with positive transitions:

\begin{defn}
  A register automaton is \emph{positive} if for each transition
  $c\xto{\phi}c'$ the formula~$\phi$ is positive: $\phi=\mathrm{true}$
  or $\phi$ uses the boolean operations $\vee$ and $\wedge$ only.
\end{defn}

\begin{theorem}\label{thm:pos-nofa-vs-pos-reg}
A data language is positive and NOFA-recognizable iff it is accepted by some positive register automaton.
\end{theorem}
Here, the approach is to regard a configuration of a positive register automaton as a state of a NOFRA.  
Conversely, an abstract transition $j\xto{E} j'$ of a NOFA can be transformed into a transition $j\xto{\phi} j'$ of a register automaton
for the conjunction $\phi$ of all equations in $E$, identifying $k=\bullet$, $\bullet=k$, $k=\bark$ with 
$(k,\mathrm{before})=\bullet$, $\bullet=(k,\mathrm{after})$, $(k,\mathrm{before})=(\bark,\mathrm{after})$. A tweak of the initial
states accounts for the requirement that registers are initially empty.

\begin{rem}\label{rem:pos-reg-aut-vs-fsuba}
  Just like register automata are equivalent to finite-memory
  automata, positive register automata correspond to a restricted
  version of finite-memory automata called \emph{finite-state
    unification-based automata} (\emph{FSUBA})~\cite{tal99,kt06}. The
  original definition of the latter involves a fixed initial register
  assignment, which enables acceptance of non-positive
  languages. However, FSUBA with empty initial registers are
  equivalent to positive register automata; see Appendix for
  details. This implies in particular that
  positive register automata admit a decidable inclusion problem, in contrast to the case of
  unrestricted register automata~\cite{nsv04}. Indeed, FSUBA translate
  into a more general model called
  \emph{RNNA}~\cite[Sec.~6]{SchroderEA17}, for which inclusion is
  decidable. Tal~\cite{tal99} has given a direct decidability proof for FSUBA.
\end{rem}

\section{Monadic Second-Order Logic with Positive Equality Tests}\label{sec:mso}
\label{S:mso}
As motivated in the Introduction, positive data languages are considered
as expressing properties of data words involving positive
statements about equality of data values. In the following we make
this idea precise. For this purpose, we employ monadic second-order
logic with equality tests, abbreviated $\MSOe$~\cite{nsv04,
  boj-13,col-15}. Its formulae are given by the grammar
\[
  \phi, \psi
  \;\;:=\;\;
  x<y \mid x\sim y \mid X(x) \mid \neg\phi \mid \phi\vee \psi \mid
  \phi\wedge\psi \mid \exists x.\,\phi \mid \exists X.\,\phi \mid
  \forall x.\,\phi \mid \forall X.\,\phi,
\]
where $x,y$ range over first-order variables and $X$ over monadic
second-order variables. A formula is interpreted over a fixed data
word $w=a_1\ldots a_n\in \Ats$. First-order variables represent
positions, i.e.~elements of the set $\{1,\ldots,n\}$, and second-order
variables represent subsets of $\{1,\ldots,n\}$. The atomic formula
$x<y$ means ``position $x$ comes before position $y$'', and $x\sim y$
means ``the same data value occurs at positions $x$ and $y$''. The
interpretation of the remaining constructs is standard. A
\emph{sentence} is a formula without free variables. We write
$L(\phi)\seq \Ats$ for the set of data words satisfying the sentence
$\phi$. For example, the languages $L_0$ and $L_1$ from the
Introduction are defined by
$\phi_0=\forall y.\, \mathrm{last}(y) \Rightarrow (\forall x.\, x<y
\Rightarrow \neg (x\sim y))$, where
$\mathrm{last}(y)=\neg\exists z.\, y<z$ and
$\psi\Rightarrow \xi = \neg\psi \vee \xi$, and by
$\phi_1=\exists x.\,\exists y.\,x<y \wedge x\sim y$.

Recall that by standard rules of negation, every formula is equivalent
to one in \emph{negation normal form} (\emph{NNF}), where for each
subformula~$\neg\phi$ the formula $\phi$ is atomic.

\begin{defn}
  An $\MSOe$ formula lies in $\MSOep$ (\emph{monadic second-order
    logic with positive equality tests}) if it admits an NNF
  containing no subformula of the form $\neg(x\sim y)$.
  A data language is   \emph{$\MSOep$-definable} if it is 
  of the form $L(\phi)$ for an $\MSOep$ sentence $\phi$.
\end{defn}

The above sentence $\phi_1$ lies in $\MSOep$ but $\phi_0$ does not. The following is immediate:

\begin{proposition}\label{prop:mso-to-pos}
Every $\MSOep$-definable language is positive.
\end{proposition}

\begin{rem}\label{rem:mso-vs-nofa}
The logic $\MSOe$ is more expressive than NOFAs~\cite{nsv04}, and the same holds for $\MSOep$:
the language defined by the $\MSOep$ sentence $\phi=\forall x.\,\exists y.\, (x<y \,\vee\, y<x )\, \wedge\, x\sim y$ (``no data value occurs only once'')
is not NOFA-recognizable. However, within the class of NOFA-recognizable languages, positive and $\MSOep$-definable languages coincide:
\end{rem}

\begin{theorem}\label{thm:mso-vs-nofa}
A NOFA-recognizable language is positive iff it is $\MSOep$-definable.
\end{theorem}
Indeed, one can express the abstract acceptance condition of \autoref{prop:abstract-lang-char} in $\MSOep$.

\section{Toposes for Names}\label{sec:toposes}
In the remainder, we investigate positive data languages and their
automata from a more conceptual perspective. Some familiarity with
basic category theory (functors, natural transformations, (co-)limits,
adjunctions) is required; see Mac Lane~\cite{mac-71} for a gentle introduction.

Nominal sets and nominal renamings sets (\autoref{sec:nominal-sets})
were initially introduced as a convenient abstract framework for
reasoning about names, and related issues such as freshness, binding,
and substitution. An alternative, and more general, approach uses
the presheaf categories $\Set^\II$~\cite{stark96} and
$\Set^\FF$~\cite{ftp99}. The intuition behind each of these categories
$\C$ is very similar: one thinks of $X\in \C$ as a collection of
finitely supported objects, equipped with a renaming operation that
extends renamings $\rho\colon \At\to\At$ to the level of elements of
$X$. The difference between the four categories lies in whether
elements admit a \emph{least} support, or just some finite support,
and in whether renamings $\rho$ are injective or arbitrary maps; see
\autoref{fig:categories}. The last column classifies the respective finitely
presentable objects, which underly automata. We now recall the latter
concept and describe the categories in more detail.

%\vspace*{-0.5\baselineskip}%
\mypar{Finitely presentable objects.}
A diagram
\( D\colon I \rightarrow \cat C \) in a category $\C$ is
\emph{directed} if its scheme~$I$ is a directed poset: every finite
subset of $I$ has an upper bound. A \emph{directed colimit} is a
colimit of a directed diagram. An object $X$ of $\C$ is \emph{finitely
  presentable} if its hom-functor $\C(X,-)\colon \C\to \Set$ to the
category of sets and functions preserves directed colimits. In many
categories, finitely presentable objects correspond to the
objects with a finite description. For example, the finitely
presentable objects of $\Set$ are precisely finite sets, and if $\C$
is a variety of algebras (e.g.~monoids, groups, rings), an algebra is
a finitely presentable object of $\C$ iff it is presentable by
finitely many generators and
relations~\cite[Thm.~3.12]{adamek_rosicky_1994}.

%\vspace*{-0.5\baselineskip}
\mypar{Nominal (renaming) sets.}
We let $\Nom$ denote the category of nominal sets and
$\Perm(\At)$-equivariant maps, and $\RnNom$ the category of nominal
renaming sets and $\Fin(\At)$-equi\-vari\-ant maps. Both categories
are toposes, that is, they are finitely complete (with limits formed
as in $\Set$), cartesian closed, and admit a subobject
classifier. Note that $\Nom$ is a boolean topos (its subobject
classifier is $2=\{0,1\}$ with the trivial group action), which is not
true for $\RnNom$~\cite[Sec.~5]{GH08}. The next proposition provides a
categorical description of orbit-finite nominal (renaming) sets; for
nominal sets this result is well-known,
see~\cite[Prop.~2.3.7]{Petrisan2012} or~\cite[Thm.~5.16]{Pitts2013},
and the statement for nominal renaming sets may be deduced from it.
\begin{proposition}\label{prop:orbit-finite-vs-fp}
A nominal (renaming) set is orbit-finite iff it is a finitely presentable object of $\Nom$ or $\RnNom$, respectively.
\end{proposition}
The forgetful functor $U\colon \RnNom\to \Nom$ given by restricting
the $\Fin(\At)$- to a $\Perm(\At)$-action has a left adjoint
$F\colon \Nom\to\RnNom$~\cite[Thm.~2.6]{mr20}. We refer to
\emph{op.~cit.} for its explicit description, but remark that
$F(\At^{\#n})=\At^n$ for every $n\in \Nat$~\cite[Thm.~3.7]{mr20}.

%\vspace*{-0.5\baselineskip}
\mypar{Presheaves.} A
\emph{(covariant) presheaf} over a small category $\C$ is a functor
$P\colon \C\to \Set$. We write $\Set^\C$ for the category of
presheaves and natural transformations. We specifically consider
presheaves over $\FF$ and $\II$, the categories whose objects are
finite subsets $S\seq_\f\At$ and whose morphisms $\rho\colon S\to T$
are functions or injective functions, respectively. The categories
$\Nom$ and $\RnNom$ form full reflective subcategories of $\Set^\II$
and $\Set^\FF$ via embeddings
\[
  I_\star\colon \Nom \monoto \Set^\II
  \qquad\text{and}\qquad
  J_\star\colon \RnNom\monoto \Set^\FF.
\]
Here, $I_\star$ is given for $X\in \Nom$, $S\seq_\f \At$, $\rho\colon S
\to T$ in $\II$ and $f\colon X\to Y$ in $\Nom$ by
\[ (I_\star X)S = \{\, x\in X : \supp x \seq S \,\},\qquad (I_\star
  X)\rho(x) = \ol{\rho}\cdot x,\qquad (I_\star f)_S(x) = f(x), \]
where $\ol{\rho}\in \Perm(\At)$ is any permutation extending the
injective map $\rho$. The embedding $J_\star$ is defined
analogously. In both cases, the essential image of the embedding
consists precisely of the presheaves preserving pullbacks of injective
maps, see~\cite[Thm. 6.8]{Pitts2013} and
\cite[Thm. 38]{GH08}. Informally, a presheaf $P\in\Set^\C$, where
$\C\in \{\II,\FF\}$, specifies a set $PS$ of $S$-supported objects for
every $S\seq_\f \At$, and the pullback preservation property asserts
precisely that every object admits a least support. A presheaf
$P\in \Set^\C$ is \emph{super-finitary} if there exists $S\seq_\f \At$
such that (i)~$PS'$ is a finite set for all $S'\seq S$, and (ii)~for
every $T\seq_\f \At$ and $x\in PT$, there exists $S'\seq S$ and
$\rho\in\C(S',T)$ such that $x\in P\rho[PS']$. (This implies that
$PT$ is finite.) Such an~$S$ is called a \emph{generating set} for
$P$. The next proposition shows that super-finitary presheaves are the
analogue of orbit-finite sets; see \cite[Cor.~3.34]{amsw19_1} for the
case $\C=\FF$:
\begin{figure}
\begin{tabular}{lllll}
\textbf{Category} & \textbf{Objects} & \textbf{Least supp.} & \textbf{Renamings} & \textbf{Finitely pres.\ objects}\\
\hline
$\Nom$ & nominal sets & yes & injective & orbit-finite sets\\
$\RnNom$ & nominal renaming sets & yes & arbitrary & orbit-finite sets\\
$\Set^\II$ & presheaves over $\II$ & no & injective & super-finitary presheaves\\
$\Set^\FF$ & presheaves over $\FF$ & no & arbitrary & super-finitary presheaves
\end{tabular}
\caption{Toposes that model sets of finitely supported objects}\label{fig:categories}
\end{figure}
\begin{proposition}\label{prop:super-finitary}
  For $\C\in \{\II,\FF\}$ and $P\in \Set^\C$, the following are
  equivalent: (i)~$P$ is super-finitary; (ii)~$P$ is finitely
  presentable; (iii)~there exists an epimorphism (a componentwise
  surjective natural transformation)
  $\coprod_{i\in I} \C(S_i,-)\epito P$ with $I$ finite and
  $S_i\seq_\f \At$. Moreover, super-finitary presheaves are closed
  under sub-presheaves and finite products.
\end{proposition}
To relate the two presheaf categories $\Set^\II$ and $\Set^\FF$,
recall that every functor $E\colon \C\to \D$ between small categories
induces an adjunction \eqref{eq:adjuction-e-star}, where the right
adjoint $E^\star$ is given by $E^\star(P)=P\circ E$, and the left
adjoint sends a presheaf $P\in \Set^\C$ to its \emph{left Kan
  extension} $\Lan_E P$. For the inclusion functor
$E\colon \II\hookto \FF$, we obtain the commutative diagram
\eqref{eq:adjunctions} of adjunctions. Here, $I^\star$ and $J^\star$
are the reflectors, i.e.~the left adjoints of $I_\star$ and $J_\star$.

\begin{minipage}{.35\textwidth}
\begin{equation}\label{eq:adjuction-e-star}
\begin{tikzcd}
\Set^\C \ar[yshift=-5]{r}{\top}[swap]{\Lan_E} & \Set^\D \ar[yshift=5]{l}[swap]{E^\star} 
\end{tikzcd}
\end{equation}
\end{minipage}
\qquad\qquad
\begin{minipage}{.4\textwidth}
\begin{equation}\label{eq:adjunctions}
\begin{tikzcd}
\Set^\II \ar[yshift=-5]{r}{\top}[swap]{\Lan_E} \ar[xshift=5]{d}{I^\star}[swap]{\vdash} & \Set^\FF \ar[yshift=5]{l}[swap]{E^\star} \ar[xshift=-5]{d}[swap]{J^\star} \\
\Nom \ar[xshift=-5]{u}{I_\star} \ar[yshift=5]{r}{F} & \RnNom \ar[yshift=-5]{l}{U}[swap]{\bot} \ar[xshift=5]{u}{\dashv}[swap]{J_\star}
\end{tikzcd}
\end{equation}
\end{minipage}

\begin{proposition}\label{prop:fp-presheaves}
All functors in \eqref{eq:adjunctions} preserve finitely presentable objects.
\end{proposition}
Hence, the adjunctions \eqref{eq:adjunctions} restrict to the full subcategories of finitely presentable objects.

\section{Nondeterministic Automata in a Category}\label{sec:naut-cat}
Our aim is to investigate nondeterministic automata and their languages in the toposes of \autoref{fig:categories},
and to compare their expressive power. To this end, we first introduce the required automata-theoretic concepts
uniformly at the level of abstract categories.

\begin{assumptions}\label{ass:cats}
  Fix a category $\C$ with finite limits and (strong epi,
  mono)-factorizations.  We assume that strong epimorphisms are stable
  under finite products (that is, $e\times e'$ is a strong epimorphism
  whenever $e$ and $e'$ are) and pullbacks (that is, in every pullback
  square $e \circ \ol{f} = f\circ \ol{e}$, the morphism $\ol{e}$ is a
  strong epimorphism whenever $e$ is).
\end{assumptions} 

The (strong epi, mono)-factorization
$f= (
\begin{tikzcd}[cramped]
  A \ar[->>]{r}{e~~} & I \ar[>->]{r}{m} & B
\end{tikzcd})
$ of a morphism $f\colon A \to B$ in
$\C$ is its \emph{image factorization}, and the subobject represented
by $m$ is the \emph{image} of $f$.

\begin{example}\label{ex:categories}
  Every topos satisfies \autoref{ass:cats}, including $\Set$, $\Nom$,
  $\RnNom$, $\Set^\II$ and $\Set^\FF$. Note that in a topos all
  epimorphisms are strong. In the five categories above, epi- and
  monomorphisms are the (componentwise) surjective and injective
  morphisms, resp. Pullbacks and finite products are formed
  (componentwise) at the level of underlying sets.
\end{example}

\begin{defn}\label{def:language-cat}
A \emph{language} over $\Sigma\in \C$ is a family of subobjects of $\Sigma^n$ for each $n\in \Nat$:
\[L\;=\;(\,m_n^{(L)}\colon L^{(n)}\monoto \Sigma^n\,)_{n\in\Nat}.\]
We write $L\leq L'$ iff $L^{(n)}\leq L'^{(n)}$ for all $n$, using the partial order $\leq$ on subobjects of $\Sigma^n$.
\end{defn}

\begin{rem}\label{rem:extensive}
If $\C$ is countably extensive (e.g.~a topos with countable coproducts), languages correspond bijectively to subobjects of $\Sigmas = \coprod_{n\in\Nat} \Sigma^n$. Indeed, every language $L$ yields the subobject $\coprod_n m_n^{(L)}\colon \coprod_n L^{(n)}\monoto \Sigmas$, and conversely every subobject of $\Sigmas$ is of this form. In particular, this holds in the categories of \autoref{ex:categories}.
\end{rem}

\begin{defn}
A \emph{nondeterministic $\C$-automaton} is a quintuple $A=(Q,\Sigma,\delta,I,F)$ consisting of an object $Q\in \C$ of \emph{states}, an \emph{input alphabet} $\Sigma\in \C$, and subobjects
\[ m_\delta\colon \delta\monoto Q\times\Sigma\times Q,\qquad m_I\colon I\monoto Q,\qquad m_F\colon F\monoto Q, \]
representing \emph{transitions}, \emph{initial states}, and \emph{final states}, respectively. A \emph{morphism} $h\colon A'\to A$ of nondeterministic $\C$-automata is given by a pair of morphisms $h_\states\colon Q'\to Q$ and $h_\al\colon \Sigma'\to \Sigma$ of $\C$ that restrict as shown below (note that $h_\trans$, $h_\init$ and $h_\final$ are uniquely determined):
\begin{equation}\label{eq:naut-morphism}
\begin{tikzcd}[column sep=4.5em]
\delta' \ar[dashed]{r}{h_\trans} \ar[rightarrowtail]{d}[swap]{m_{\delta'}} & \delta \ar[rightarrowtail]{d}{m_\delta} \\
Q'\times \Sigma'\times Q' \ar{r}{h_\states\times h_\al \times h_\states} & Q\times \Sigma\times Q
\end{tikzcd}
\qquad
\begin{tikzcd}[column sep=3em]
I' \ar[dashed]{r}{h_\init} \ar[rightarrowtail]{d}[swap]{m_{I'}} & I \ar[rightarrowtail]{d}{m_I} \\
Q' \ar{r}{h_\states} & Q
\end{tikzcd}
\qquad
\begin{tikzcd}[column sep=3em]
F' \ar[dashed]{r}{h_\final} \ar[rightarrowtail]{d}[swap]{m_{F'}} & F \ar[rightarrowtail]{d}{m_F} \\
Q' \ar{r}{h_\states} & Q
\end{tikzcd}
\end{equation}
We write $\NAut(\C)$ for the category of nondeterministic automata in $\C$ and their morphisms, and $\NAutfp(\C)$ for its full subcategory given by \emph{nondeterministic fp-automata}, viz.\ automata where $Q$, $\Sigma$, $\delta$, $I$, $F$ are finitely presentable objects of $\C$. 
\end{defn}

\begin{defn}\label{def:naut-acc-lang}
For every nondeterministic $\C$-automaton $A=(Q,\Sigma,\delta,I,F)$, its \emph{accepted language} is the language $L(A)$ over $\Sigma$ given as follows:
\begin{enumerate}
\item $m_{L(A)}^{(0)}\colon L^{(0)}(A)\monoto 1=\Sigma^0$ is the image
  of the unique morphism $I\cap F\xto{!} 1$, where $1$ is the terminal
  object of $\C$ and $I\cap F$ is the intersection (pullback) of $m_I$
  and $m_F$.
\item For $n>0$, the subobject
  $m_{L(A)}^{(n)}\colon L^{(n)}(A)\monoto \Sigma^n$ is defined via the
  commutative diagram
\[
\begin{tikzcd}
L^{(n)}(A) \ar[rightarrowtail]{d}[swap]{m_{L(A)}^{(n)}} & \Run^{(n)}_A  \pullbackangle{-45} \ar[two heads]{l}[swap]{e_{n,A}} \ar[rightarrowtail]{r}{\ol{d}_{n,A}} \ar[rightarrowtail]{d}[pos=.7]{\ol{m}_\delta^{(n)}} & \delta^n \ar[rightarrowtail]{d}{m_\delta^n} &\\
\Sigma^n & I\times (\Sigma \times Q)^{n-1}\times \Sigma \times F \ar{l}[swap]{p_{n,A}} \mar{r}{d_{n,A}} & (Q\times \Sigma\times Q)^n 
\end{tikzcd}
\]
Here, letting $\Delta\colon Q\monoto Q\times Q$ denote the diagonal, $d_{n,A}$ is the monomorphism
\[ I\times (\Sigma \times Q)^{n-1}\times \Sigma \times F \xto{m_I\times (\id\times \Delta)^{n-1}\times \id\times m_F} Q\times (\Sigma \times Q\times Q)^{n-1}\times \Sigma\times Q \cong (Q\times \Sigma \times Q)^n,  \]
the morphisms $\ol{d}_{n,A}$ and $\ol{m}_{\delta}^{(n)}$ form the pullback of $d_{n,A}$ and $m_\delta^n$, the morphism $p_{n,A}$ is the projection, and $e_{n,A}$ and $m_{L(A)}^{(n)}$ form the image factorization of $p_{n,A}\circ \ol{m}_\delta^{(n)}$.
\end{enumerate}
\end{defn}

\begin{example}\label{ex:automata}
\begin{enumerate}
\item\label{ex:automata-1} A nondeterministic fp-automaton in $\Set$
  is a classical nondeterministic finite automaton. The pullback
  $\Run_A^{(n)}$ is the set of accepting runs of length $n$,
  hence~$L(A)$ is the usual accepted language: the set of words with
  an accepting run.
  
\item\label{ex:automata-2} A nondeterministic fp-automaton in $\Nom$
  or $\RnNom$ with alphabet $\Sigma=\At$ is a NOFA or NOFRA,
  respectively. The two notions of accepted language in \autoref{def:nofa} and
  \autoref{def:naut-acc-lang} match, that is,
  $L(A)$ is the set of words with an accepting run.
  
\item In the next section, we will also look into nondeterministic
  $\Set^\II$- and $\Set^\FF$-automata.
\end{enumerate}
\end{example}

\begin{rem}
  Readers familiar with coalgebras~\cite{rutten00} may note that if
  $\C$ is a topos, the final states and transitions of a
  nondeterministic $\C$-automaton correspond to a coalgebra
  $\gamma \colon Q\to \Omega\times (\Pow Q)^\Sigma$ where $\Omega$ is
  the subobject classifier and $\Pow\colon \C\to\C$ is the covariant
  power object functor~\cite[Sec.~A.2.3]{johnstone02}.
  We expect our above definition of accepted language to match the one
  given by coalgebraic trace semantics~\cite{HasuoEA07,SilvaEA13},
  with the required arguments relying on the internal logic of the
  topos $\C$. Details are left for future work; we have found that the
  present relational approach to automata leads to shorter and more
  direct proofs.
\end{rem}

\begin{proposition}\label{prop:epi-pres-lang}
  Let $h\colon A'\to A$ be an $\NAut(\C)$-morphism where
  $\Sigma'=\Sigma$ and $h_\al=\id_{\Sigma}$.
  \begin{enumerate}
  \item The accepted language of $A'$ is contained in that of $A$,
    that is, $L(A')\leq L(A)$.
    
  \item If $h_\states$ is strongly epic in $\C$ and the squares
    \eqref{eq:naut-morphism} are pullbacks, then $L(A')=L(A)$.
\end{enumerate}
\end{proposition}

Hence, the construction $A\mapsto A'$ of \autoref{rem:strong} indeed yields an equivalent NOFA.

\begin{proposition}\label{prop:liftadjaut}
  Let $\C$ and $\D$ be categories satisfying the \autoref{ass:cats}.
  \begin{enumerate}
  \item Every functor $G\colon \C \to \D$ lifts to a functor
    $\barG\colon \NAut(\C) \to \NAut(\D)$ defined by
    \[
      \barG(Q,\Sigma,\delta,I,F) = (GQ,G\Sigma, \ol{G\delta},
      \ol{GI}, \barGF)
      \qquad \text{and}\qquad
      \barG f = Gf.
    \]
    Here, $\ol{G\delta}$, $\ol{GI}$, $\barGF$ are the images of the
    morphisms shown below, with $\canmap$ denoting the canonical
    morphism induced by the product projections:
    \[
      G\delta \xto{Gm_\delta} G(Q\times \Sigma\times Q) \xto{\canmap}
      GQ\times G\Sigma\times GQ,
      \qquad GI\xto{Gm_I} GQ, \qquad GF \xto{Gm_F} GQ.
    \]
    
  \item Every adjunction $L \dashv R\colon \C \to \D$ lifts to an
    adjunction $\barL \dashv \barR\colon \NAut(\C) \to \NAut(\D)$.
  \end{enumerate}
\end{proposition}

In particular, the adjunctions \eqref{eq:adjunctions} lift to adjunctions between the respective categories of nondeterministic automata, which in turn restrict to fp-automata by \autoref{prop:fp-presheaves}:
\begin{equation}\label{eq:adjunctions-naut}
\begin{tikzcd}[column sep=2em]
\NAutfp(\Set^\II) \ar[yshift=-5]{r}{\top}[swap]{\ol{\Lan}_E} \ar[xshift=5]{d}{{\barI}^\star}[swap]{\vdash} & \NAutfp(\Set^\FF) \ar[yshift=5]{l}[swap]{{\barE}^\star} \ar[xshift=-5]{d}[swap]{{\barJ}^\star} \\
\NAutfp(\Nom) \ar[xshift=-5]{u}{\barI_\star} \ar[yshift=5]{r}{\barF} & \NAutfp(\RnNom) \ar[yshift=-5]{l}{\ol{U}}[swap]{\bot} \ar[xshift=5]{u}{\dashv}[swap]{\barJ_\star}
\end{tikzcd}
\end{equation}

The positive closure $A\mapsto \barA$ of \autoref{cons:bar-A}, which
is key to our results in
\renewcommand{\sectionautorefname}{Sections}%
\autoref{sec:nofa-nofra} through~\ref{sec:mso},
\renewcommand{\sectionautorefname}{Section}%
is an instance of the proposition since $\barA=\barF A$ for the left adjoint $F\colon \Nom\to \RnNom$. 

\section{Nondeterministic Presheaf Automata}\label{sec:preaut}
We proceed to relate the expressive power of the four automata models in \eqref{eq:adjunctions-naut}. Specifically, for $\C\in \{\II,\FF\}$ we consider
nondeterministic $\Set^\C$-automata $A=(Q,\Sigma,\delta,I,F)$ with a super-finitary (= finitely presentable) presheaf $Q$ of states and input
alphabet $\Sigma=V_\C\in \Set^\C$, for the inclusion functor $V_\C(S)=S$. (This implies that $\delta$, $I$ and $F$ are super-finitary by~\Cref{prop:super-finitary}.)
Note that $V_\C$ corresponds to the input alphabet $\At$ used for NOF(R)As:
\[ V_\II = I_\star(\At)\qquad\text{and}\qquad V_\FF = J_\star(\At) = \Lan_E(V_\II). \]
A \emph{language} in $\Set^\C$ is a sub-presheaf $L\seq V_\C^\star$, or equivalently a family of sub-presheaves $L^{(n)}\seq V_\C^n$ for $n\in \Nat$
(\autoref{def:language-cat} and \autoref{rem:extensive}). Here, $V_\C^\star (S)=S^\star$, the set of words over the finite alphabet $S\seq_\f\At$,
and $V_\C^n(S)=S^n$, the subset of words of length $n$. 

\begin{rem}\label{rem:presheaf-lang-vs-word-lang}
For the sake of distinction, we refer to languages in $\Set^\C$ as \emph{presheaf languages}, and to subsets of $\Ats$ as \emph{word languages}.
Both concepts are closely related: Every presheaf language $L\seq V_\II^\star$ in $\Set^\II$ induces a $\Perm(\At)$-equivariant word language
$\wordlang{L}\seq \Ats$ given by $\wordlang{L} = \bigcup_{S\seq_\f \At} L(S)$, and, conversely, every $\Perm(\At)$-equivariant word language $K\seq \Ats$
induces a presheaf language $\prelang{K}\seq V_\II^\star$ given by $[\prelang{K}]S= K\cap S^\star$ for $S\seq_\f \At$. Analogously for presheaf languages
in $\Set^\FF$ and $\Fin(\At)$-equivariant word languages. In both cases, these translations almost yield a bijective correspondence: one has $K=\wordlang{\prelang{K}}$,
but generally only $L\seq \prelang{\wordlang{L}}$. For instance, for $L\seq V_\FF^\star$ given by $L(\emptyset)=\emptyset$ and $L(S)=\{\varepsilon\}$ for
$S\neq\emptyset$ one has $[\prelang{\wordlang{L}}]\emptyset=\{\varepsilon\}$, so $L\subsetneq \prelang{\wordlang{L}}$. The equality $L= \prelang{\wordlang{L}}$
holds iff $L$ is \emph{downwards closed}, that is, $L(S')=L(S)\cap (S')^\star$ for all $S'\seq S\seq_\f \At$.
\end{rem}

The presheaf version of positive word languages and positive closures is as follows:

\begin{defn}\label{def:pos-closure-presheaf-lang} Let $L\seq V_\II^\star$ be a presheaf language in $\Set^\II$.
\begin{enumerate}
\item The language $L$ is \emph{positive} if $L=KE$ for some (unique) language $K\seq V_\FF^\star$ in $\Set^\FF$.
\item A \emph{positive closure} of $L$ is a language $\barL$ in $\Set^\FF$ such that $L\seq \barL E$ and $\barL$ is minimal with that property, that is, $\barL\seq K$ for every language $K\seq V_\FF^\star$ in $\Set^\FF$ such that $L\seq KE$.
\end{enumerate}
\end{defn}

A positive closure is clearly unique; its existence is ensured by the next proposition, which is proved using the universal property of left Kan extensions.

\begin{proposition}\label{prop:pos-closure-presheaf-lang}
The positive closure of $L\seq {V_{\II}^*}$ is given by the image of the
  morphism
  \[ \varphi\colon \Lan_E(L) \xra{\Lan_E(\iota)} \Lan_E(V_{\II}^*) \cong \coprod_{k}\Lan_E(V_{\II}^k) \xto{\coprod_k \canmap_k} \coprod_{k}\Lan_E(V_{\II})^k = \coprod_{k} V_\FF^k = V_{\FF}^* \]
where $\iota\colon L \hookrightarrow {V_{\II}^*}$ is the inclusion, the isomorphism witnesses preservation of coproducts by the left adjoint $\Lan_E$, and $\canmap_k$ is the canonical map induced by the product projections. 
\end{proposition} 

\begin{rem}\label{rem:strong-presheaf}
  A presheaf $P\in \Set^\II$ is \emph{strong} if $P=I_\star(X)$ for a
  strong nominal set~$X$. Since~$I_\star$ preserves coproducts,
  (super-finitary) strong presheaves are exactly (finite) coproducts
  $\coprod_{j\in J} \II(S_j,-)$ of representable presheaves.  By
  \autoref{prop:super-finitary} and \autoref{prop:epi-pres-lang},
  every super-finitary $\Set^\II$-automaton is equivalent to one whose
  presheaf of states is strong.  Given such an automaton $A$ with
  states $Q=\coprod_{j\in J} \II(S_j,-)$, applying the lifted left
  adjoint $\ol{\Lan}_E$ yields a super-finitary $\Set^\FF$-automaton
  $\barA$ with states $\Lan_E(Q)=\coprod_{j\in J} \FF(S_j,-)$, using
  that $\Lan_E$ preserves coproducts and representables (see
  e.g.~\cite[Ex.~X.3.2]{mac-71}). This is the analogue of
  \autoref{cons:bar-A} for presheaf automata. Similar to
  \autoref{prop:barA-acc-barL}, we have
\end{rem}

\begin{proposition}\label{prop:presheaf-aut-closure}
  For every super-finitary nondeterministic $\Set^\II$-automaton $A$
  with a strong presheaf of states, the $\Set^\FF$-automaton
  $\barA=\ol{\Lan}_E(A)$ accepts the language $\ol{L(A)}$.
\end{proposition}

While by definition nondeterministic presheaf automata accept presheaf
languages, using \autoref{rem:presheaf-lang-vs-word-lang} we can also
naturally associate a word language semantics to them:
\begin{defn}
\begin{enumerate}
\item The word language \emph{accepted} by a nondeterministic
  $\Set^\C$-automaton~$A$ is $\wordlang{L(A)}\seq\Ats$, the word
  language induced by the presheaf language of $A$.
  
\item A word language $L\seq\Ats$ is \emph{$\Set^\C$-recognizable} if
  there exists a super-finitary non\-de\-termi\-nis\-tic $\Set^\C$-automaton
  accepting it.
\end{enumerate}
\end{defn}

This enables a classification of the expressive power of nondeterministic $\Set^\C$-automata:

\begin{theorem}\label{thm:presheaf-automata-vs-nofa}
\begin{enumerate}
\item\label{thm:pavs:1} A word language is NOFA-recognizable iff it is $\Set^\II$-recognizable.
\item\label{thm:pavs:2} A word language is positive and NOFA-recognizable iff it is $\Set^\FF$-recognizable. 
\end{enumerate}
\end{theorem}

For \autoref{thm:pavs:1} one shows that the functors $\barI_\star$ and ${\barI}^\star$
of \eqref{eq:adjunctions-naut} preserve the accepted word languages of
automata. For \autoref{thm:pavs:2} one uses \autoref{prop:presheaf-aut-closure} and
the observation that every nondeterministic $\Set^\FF$-automaton
accepts a positive word language.

This shows that the theory of data languages can be based on presheaves rather than nominal sets~\cite{BojanczykEA14}. In particular, the conceptual difference between the two approaches~(viz.~ex\-istence of least supports) is largely inessential from the perspective of automata theory.

\section{Conclusions and Future Work}
We have characterized positive data languages recognizable by
NOFAs in terms of register
automata, logic, and category theory; see \autoref{fig:equivalences}
for a summary of our contributions. Our results underline the
phenomenon that weak classes of data languages tend to have a rich
theory and admit many equivalent perspectives, paralleling classical
regular languages over finite alphabets. For example, a similar observation has
been made for data languages recognizable by orbit-finite
nominal monoids~\cite{boj-13,col-15,boj-20}.

The logic $\MSOep$ defines positive data languages, but is more
expressive than NOFAs. Identifying a suitable syntactic fragment of
$\MSOep$ that captures precisely the positive NOFA-recognizable
languages remains an open problem. The same holds for the decidability
of the satisfiability problem for $\MSOep$, which is known to be undecidable for
$\MSOe$~\cite{klt21}. On a related note, it might be interesting to characterize the expressive power of full $\MSOep$. Specifically, does it capture precisely the $\MSOe$-definable positive languages?

Finally, besides register automata, a number of further
automata models for data languages have been proposed, most notably
pebble automata~\cite{nsv04} and data
automata~\cite{bdmss06,bmssd06}. In general, these models differ in
their expressive power.  However, it is conceivable that some or all
of them may become equivalent when restricted to positive data
languages.

%%
%% Bibliography
%%
%\clearpage
\bibliography{refs}

\clearpage
\appendix

\section{Appendix}
This Appendix provides proof details and additional explanations omitted for lack of space.

\subsection*{Proof of \autoref{prop:nofra-acc-pos-lang}}
Let $A=(Q,\delta,I,F)$ be a NOFRA. Given a word $w\in L({A})$ with accepting run \[(j_0,q_0)\xto{a_1} (j_1,q_1)\xto{a_2}\cdots \xto{a_n} (j_n,q_n)\] and a renaming $\rho\colon \At\to \At$, we have the accepting run 
\[(j_0,\rho^\star q_0)\xto{\rho a_1} (j_1, \rho^\star q_1) \xto{\rho a_2}\cdots \xto{\rho a_n} (j_n,\rho^\star q_n)\] 
by $\Fin(\At)$-equivariance of $\delta$, $I$, $F$. Hence $\rho^\star(w)\in L(A)$, so $L(A)$ is $\Fin(\At)$-equivariant.

\subsection*{Proof of \autoref{prop:barA-acc-barL}}
Our task is to prove $L(\barA)=\ol{L(A)}$.

\medskip\noindent ($\supseteq$) We have $L(A)\seq L(\barA)$ because the NOFA $A$ is a sub-NOFA of $\barA$. Moreover, the language $L(\barA)$ is positive by \autoref{prop:nofra-acc-pos-lang}, so $\ol{L(A)}\seq L(\barA)$.

\medskip\noindent ($\seq$) We prove that for every run
\begin{equation}\label{eq:run-barA} 
(j_0,q_0)\xto{b_1} (j_1,q_1)\xto{b_2}\cdots \xto{b_n} (j_n,q_n)
\end{equation}
in $\barA$, there exists a renaming $\rho\colon \At\to \At$ and a run
\begin{equation}\label{eq:run-A} 
(j_0,p_0)\xto{a_1} (j_1,p_1)\xto{a_2}\cdots \xto{a_n} (j_n,p_n)
\end{equation}
in $A$ such that $\rho^\star p_i=q_i$ for $i=0,\ldots,n$ and $\rho a_i =b_i$ for $i=1,\ldots,n$. Note that if \eqref{eq:run-barA} is accepting then so is \eqref{eq:run-A}; therefore $L(\barA)\seq \ol{L(A)}$. We construct \eqref{eq:run-A} by induction on $n$.

\medskip\noindent \emph{Induction base ($n=0$)}. Choose $p_0\in \At^{\#m}$ arbitrary and a renaming $\rho\colon \At\to\At$ mapping each letter of $p_0$ to the corresponding letter of $q_0\in \At^m$. Then $\rho^\star q_0=p_0$, as required.  

\medskip\noindent \emph{Induction step ($n\to n+1$)}.  Suppose that
\[(j_0,q_0)\xto{b_1} (j_1,q_1)\xto{b_2}\cdots \xto{b_n} (j_n,q_n) \xto{b_{n+1}} (j_{n+1},q_{n+1}) \]
is a run in $\barA$. By induction, we know that there exists a renaming $\rho\colon \At\to \At$ and a run
\[
(j_0,p_0)\xto{a_1} (j_1,p_1)\xto{a_2}\cdots \xto{a_n} (j_n,p_n)
\]
in $A$ such that $\rho^\star p_i = q_i$ for $i=0,\ldots,n$ and $\rho a_i=b_i$ for $i=1,\ldots,n$. Furthermore, since $(j_n,q_n)\xto{b_{n+1}} (j_{n+1}, q_{n+1})$ in $\barA$, there exists a renaming $\sigma\colon \At\to \At$ and a transition $(j_{n},p_n')\xto{a_{n+1}} (j_{n+1},p_{n+1})$
in $A$ such that $\sigma^\star p_{n}'=q_n$, $\sigma^\star p_{n+1}=q_{n+1}$ and $\sigma a_{n+1}=b_{n+1}$. We show below that we can choose this transition in such a way that (1) $p_n'=p_n$, (2) all names in $\{a_{n+1}\} \cup \supp(p_{n+1})$ that are fresh for $p_n$ are fresh for $p_0,\ldots, p_{n}$ and $a_1,\ldots,a_n$, and (3) $\rho=\sigma$. Then, by (1) and (3), we obtain the run
\[
(j_0,p_0)\xto{a_1} (j_1,p_1)\xto{a_2}\cdots \xto{a_n} (j_n,p_n)\xto{a_{n+1}} (j_{n+1},p_{n+1})
\]
in $A$ with the required properties. It remains to show how to enforce (1), (2), (3).

\medskip\noindent\emph{Ad (1).} Since $p_n,p_n'\in \At^{\# m}$, there exists a permutation $\pi\in \Perm(\At)$ such that $\pi \cdot p_n'=p_n$. Then, by equivariance, we have the transition 
\[p_n=\pi\cdot p_n' \xto{\pi\cdot a_{n+1}} \pi\cdot p_{n+1}\] in $A$, and
\[(\sigma\circ \pi^{-1})^\star(p_n)=q_n,\quad (\sigma\circ \pi^{-1})^\star(\pi\cdot p_{n+1})=q_{n+1},\quad (\sigma\circ \pi^{-1})\cdot \pi\cdot  a_{n+1}=b_{n+1}.\] 
Thus (1) holds after replacing $p_n',a_{n+1},p_{n+1},\sigma$ with $p_n,\pi\cdot a_{n+1}, \pi\cdot p_{n+1}, \sigma\circ \pi^{-1}$.

\medskip\noindent\emph{Ad (2).} Suppose that (1) holds. Let $c_1,\ldots,c_k$ be the names in $\{a_{n+1}\} \cup \supp(p_{n+1})$ that are fresh for $p_n$. Choose names $d_1,\ldots,d_k$ fresh for  $p_0,\ldots, p_{n}$, $a_1,\ldots,a_n,c_1,\ldots,c_k$. Then, by equivariance, the permutation $\pi=(c_1\, d_1)\cdots (c_k\, d_k)$ yields the transition  
\[ (j_n,p_n) = (j_n,\pi\cdot p_n) \xto{\pi\cdot a_{n+1}} (j_{n+1},\pi\cdot p_{n+1}) \]
By definition of $\pi$, the names in $\{\pi\cdot a_{n+1}\}\, \cup\, \supp(\pi\cdot p_{n+1})= \{\pi\cdot a_{n+1}\}\, \cup\, \pi\cdot \supp(p_{n+1})$ that are fresh for $p_n=\pi\cdot p_n$ are precisely $d_1,\ldots, d_k$, and thus are fresh for $p_0,\ldots, p_{n}$ and $a_1,\ldots,a_n$. Thus (1) and (2) hold after replacing $p_n,a_{n+1},p_{n+1},\sigma$ with $p_n, \pi\cdot a_{n+1}, \pi\cdot p_{n+1}, \sigma\circ \pi^{-1}$.
 
\medskip\noindent\emph{Ad (3).} Finally, suppose that (1) and (2) hold. Choose a renaming $\tau\colon \At\to\At$ that agrees with $\rho$ on $\{a_1,\ldots,a_n\}\,\cup\,\supp(p_0)\,\cup\, \cdots \,\cup\, \supp(p_n)$ and with $\sigma$ on $\{a_{n+1}\}\,\cup\, \supp(p_n)\,\cup\, \supp(p_{n+1})$. Such $\tau$ exists by (2) and because $\rho^\star(p_n) = q_n= \sigma^\star(p_n)$ implies that $\rho$ and $\sigma$ agree on $\supp(p_n)$. Thus, after replacing $\rho$ and $\sigma$ with $\tau$, all three conditions (1), (2), (3) hold.

\subsection*{Proof of \autoref{prop:trans-bar-A}}

($\Rightarrow$) Suppose that $(j,q)\xto{b} (j',q')$ in $\barA$. By definition of $\barDelta$, this means that there exists a transition $(j,p)\xto{a}(j',p')$ in $A$ and a renaming $\rho\colon \At\to\At$ such that $\rho^\star p=q$, $\rho^\star p'=q'$, $\rho a=b$. Then the induced abstract transition $(j,E,j')$ lies in $\abs(\delta)$, i.e. $j\xto{E} j'$, and the triple $((j,q),b,(j',q'))$ is consistent with it. Indeed, if $k=\bullet$ in $E$ then $p_k=a$, hence $q_k=\rho p_k=\rho a=b$. Similarly for equations $\bullet = k$ and $k=\bark$ in $E$.

\medskip\noindent ($\Leftarrow$) Suppose that the triple $((j,q),b,(j',q'))$ is consistent with some $j\xto{E} j'$. Choose a transition $(j,p)\xto{a} (j,p')$ in $A$ inducing the abstract transition $j\xto{E} j'$, and a renaming $\rho\colon \At\to \At$ mapping $p_k$ to $q_k$, $p_k'$ to $q_k'$ and $a$ to $b$. (Note that a well-defined choice of $\rho$ is possible: If $p_k=a$ then $k=\bullet$ in $E$ and hence $q_k=b$ by consistency. Similarly, $a=p_k'$ implies $b=q_k'$ and $p_k=p_{\overbar{k}}'$ implies $q_k=q_{\overbar{k}}'$.) Since $\rho^\star p = q$, $\rho^\star p'=q'$ and $\rho a=b$, we conclude that $(j,q)\xto{b} (j',q')$ in $\barA$.

\subsection*{Proof of \autoref{prop:abstract-lang-char}}
We start with a remark and a technical lemma:

\begin{rem}\label{rem:nofa-abstr-transition-elements}
The abstract transition $(j,E,j')$ induced by $((j,p),a,(j,p'))\in Q\times \At\times Q$ contains (i) at most one equation $k=\bullet$, (ii) at most one equation $\bullet = k$, (iii) for each $k$ at most one equation $k=\bark$, (iv) for each $\bark$ at most one equation $k=\bark$. Indeed, since $p,p'\in \At^{\#m}$ every data value occurs at most once in $p$ or $p'$, respectively. Moreover if $E$ contains any two of the equations $k=\bullet$, $\bullet = \bark$, $k=\bark$, then it contains the third one.
\end{rem}

\begin{lemma}\label{lem:cons-trans-choice}
For every abstract transition $j\xto{E} j'$ in $\abs(\delta)$ and $q\in \At^m$, there exists a transition $(j,q)\xto{b} (j',q')$ in $\barA$ consistent with it.
\end{lemma}

\begin{proof}
Choose a transition $(j,p)\xto{a} (j',p')$ in $A$ inducing the abstract transition $j\xto{E} j'$, and let $\rho\colon \At\to\At$ be a renaming sending $p_k$ to $q_k$ for each $k\in \{1,\ldots,m\}$; hence $\rho^\star p=q$. Then, putting $b=\rho a$ and $q'=\rho^\star p'$, we obtain the transition $(j,q)\xto{b} (j',q')$ in $\barA$ consistent with $j\xto{E} j'$.
\end{proof}
Now we prove \autoref{prop:abstract-lang-char}.

\medskip\noindent ($\Rightarrow$) Suppose that the word $b_1\cdots b_n\in \Ats$ is accepted by $\barA$ via the accepting run
\[
(j_0,q_0)\xto{b_1} (j_1,q_1)\xto{b_2}\cdots \xto{b_n} (j_n,q_n).
\]
By \autoref{prop:trans-bar-A}, each transition $(j_{r-1},q_{r-1})\xto{b_r} (j_r,q_n)$ is consistent with some abstract transition $j_{r-1}\xto{E_r} j_r$ in $\abs(\delta)$. Thus each transition of the above run is consistent with the transitions of the abstract run
\[ j_0\xto{E_1} j_1 \xto{E_2} \cdots  \xto{E_n}  j_n. \]
Then by definition of the predicates $\Eq^{(i)}_k$, the condition
\eqref{eq:eq-cond-1} holds.

\medskip\noindent ($\Leftarrow$) Suppose that there exists an accepting abstract run 
\[ j_0\xto{E_1} j_1 \xto{E_2} \cdots  \xto{E_n}  j_n \]
such that \eqref{eq:eq-cond-1} holds. We show that for each $r=0,\ldots,n$ there exists a run
\begin{equation}\label{eq:run-cons}
(j_0,q_0)\xto{b_1} (j_1,q_1)\xto{b_2}\cdots \xto{b_r} (j_r,q_r)
\end{equation}
whose transitions are consistent with the first $r$ abstract transitions of the abstract run; in particular, putting $r=n$ this proves that $\barA$ accepts $b_1\cdots b_n$.

The run is constructed by induction on $r$. For $r=0$, any choice of $q_0\in \At^m$ will do. For $r=1$ choose the transition $(j_0,b_1^m)\xto{b_1} (j_1,b_1^m)$, which is trivially consistent with $j_0\xto{E_1} j_1$. 
Thus suppose that $0<r<n$ and that a consistent run \eqref{eq:run-cons} of length $r$ has been constructed. By \autoref{lem:cons-trans-choice} there exists a transition $(j_r,q_r) \xto{b} (j_{r+1},q_{r+1})$ in $\barA$ consistent with $j_r \xto{E_{r+1}} j_{r+1}$. 
We show how to turn the run 
\[(j_0,q_0)\xto{b_1} (j_1,q_1)\xto{b_2}\cdots \xto{b_r} (j_r,q_r)\xto{b} (j_{r+1},q_{r+1})\]
into a run for the word $b_1\ldots b_rb_{r+1}$ satisfying the required consistency property. This requires a case distinction depending on the equations occurring in $E_{r+1}$:

\medskip\noindent\underline{Case 1:} $k=\bullet$ in $E_{r+1}$ for some $k$.

\medskip\noindent\underline{Subcase 1.1:} $\Eq^{(i)}_k(r)$ for some $i$.\\
Note that necessarily $i\leq r$ by definition of $\Eq^{(i)}_k$. Then $b=q_{r,k}=b_i=b_{r+1}$: The first equality holds because the transition $(j_r,q_r)\xto{b} (j_{r+1},q_{r+1})$ is consistent with $k=\bullet$, the second one because $\Eq^{(i)}_k(r)$, and the third one by \eqref{eq:eq-cond-1}. We thus obtain the following consistent run for $b_1\ldots b_rb_{r+1}$: 
\[ (j_0,q_0)\xto{b_1} (j_1,q_1)\xto{b_2}\cdots \xto{b_r} (j_r,q_r)\xto{b_{r+1}} (j_{r+1},q_{r+1}). \]
\noindent\underline{Subcase 1.2:} $\Eq^{(i)}_k(r)$ does not hold for any $i$.\\
Consider the unique $s\in \{0,\ldots,r\}$ and the unique $k_{s}, k_{s+1}, \ldots, k_r=k$ such that $k_{t-1}=k_{t}$ in $E_{t}$ for $t\in \{ s+1,\ldots,r\}$ and no equation $\bark=k$ is contained in $E_{s}$ (putting $E_{0}=\emptyset$). Then 
\begin{itemize}
\item for $t\in \{s,\ldots,r\}$ one has $q_{k_t}=b$;
\item for $t\in \{s+1,\ldots, r\}$ one does not have $\Eq^{(i)}_{k_t}(t)$ for any $i$ (for otherwise $\Eq^{(i)}_k(r)$ by definition of the predicates);
\item for $t\in \{s,\ldots,r\}$ the equation $\bullet=k_t$ is not contained in $E_t$ (for otherwise $\Eq^{(t)}_{k_t}(t)$).
\item for $t\in \{s+1,\ldots,r\}$ the equation $k_{t-1}=\bullet$ is not contained in $E_{t}$ (for otherwise $\bullet=k_{t}$ in $E_{t}$ since $k_{t-1}=k_{t}$ in $E_{t}$).  
\end{itemize}

For $t\in \{s,\ldots,r\}$ let $q_t'$ emerge from $q_t$ by replacing the letter $b$ at position $k_t$ by the letter $b_{r+1}$. Then the triples $((j_{t-1},q_{t-1}'),b_{t},(j_{t},q_{t}'))$ for $t\in \{s+1,\ldots, r\}$ are consistent with $(j_{t-1},E_t,j_t)$, as is the triple $((j_{s-1},q_{s-1}),b_s,(j_s,q_s'))$ if $s>0$. It follows by \autoref{prop:trans-bar-A} that  $(j_t,q_t')\xto{b_{t+1}} (j_{t+1},q_{t+1}')$ and $(j_{s-1},q_{s-1})\xto{b_s} (j_s,q_s')$ (if $s>0$) are transitions in $\barA$. Thus we obtain the consistent run
\[ (j_0,q_0)\xto{b_1} (j_1,q_1) \xto{b_2}\cdots\xto{b_{s-1}} (j_{s-1},q_{s-1}) \xto{b_s} (j_s,q_s') \xto{b_{s+1}} \cdots \xto{b_r} (j_r,q_r'). \]
If $\bullet=\bark$ in $E_{r+1}$ for some $\bark$, then let $q_{r+1}'$ emerge from $q_{r+1}$ by replacing the $\bark$-th letter of $q_{r+1}$ with $b_{r+1}$; otherwise put $q_{r+1}'=q_{r+1}$. Then the triple $((j_r,q_{r}'),b_{r+1},(j_{r+1},q_{r+1}'))$ is consistent with $(j_r,E_{r+1},j_{r+1})$, so $(j_r,q_r')\xto{b_{r+1}} (j_{r+1},q_{r+1}')$ in $\barA$ and thus 
\[(j_0,q_0)\xto{b_1} (j_1,q_1) \cdots (j_{s-1},q_{s-1}) \xto{b_s} (j_s,q_s') \xto{b_{s+1}} \cdots \xto{b_r} (j_r,q_r') \xto{b_{r+1}} (j_{r+1},q_{r+1}') \]
is a consistent run for $b_1\ldots b_r b_{r+1}$.

\medskip\noindent\underline{Case 2:} No $k=\bullet$ in $E_{r+1}$.

\medskip\noindent\underline{Subcase 2.1:} $\bullet=\bark$ in $E_{r+1}$ for some $\bark$.\\
Let $q_{r+1}'$ emerge from $q_{r+1}$ by replacing the $\bark$-th letter (viz.\ $b$) with $b_{r+1}$. It then follows that the triple $((j_r,q_r),b_{r+1},(j_{r+1},q_{r+1}'))$ is consistent with $(j_r,E_{r+1},j_{r+1})$. (To see this, note that $E_{r+1}$ does not contain an equation $k=\bark$, for otherwise $k=\bullet$ in $E_{r+1}$.)
 Hence $((j_r,q_r)\xto{b_{r+1}} (j_{r+1},q_{r+1}'))$ in $\barA$ and we obtain the following consistent run for $b_1\cdots b_r b_{r+1}$:
\[ (j_0,q_0)\xto{b_1} (j_1,q_1)\xto{b_2}\cdots \xto{b_r} (j_r,q_r)\xto{b_{r+1}} (j_{r+1},q_{r+1}'). \]

\noindent\underline{Subcase 2.2:} No $\bullet=\bark$ in $E_{r+1}$.\\
Since also no $k=\bullet$ in $E_{r+1}$, the triple $((j_r,q_r),b_{r+1},(j_{r+1},q_{r+1}))$ is consistent with $(j_r,E_{r+1},j_{r+1})$.
It follows that $((j_r,q_r)\xto{b_{r+1}} (j_{r+1},q_{r+1}))$ in $\barA$, which yields the following consistent run for $b_1\cdots b_r b_{r+1}$:
\[ (j_0,q_0)\xto{b_1} (j_1,q_1)\xto{b_2}\cdots \xto{b_r} (j_r,q_r)\xto{b_{r+1}} (j_{r+1},q_{r+1}). \]
This concludes the proof.

\subsection*{Proof of~\autoref{thm:positive-non-guessing}}

By \autoref{prop:barA-acc-barL} every positive NOFA-recognizable language is accepted by some NOFRA $\barA$ as given by~\autoref{cons:bar-A}. Therefore, it suffices to turn $\barA$ into an equivalent non-guessing NOFRA. To this end, we first modify $\barA$ in such a way that it keeps track of the set $S\seq\{1,\ldots,m\}$ of those registers whose content is determined by previous abstract transitions of $A$, and modifies the content of registers outside that set arbitrarily.

\begin{construction}\label{cons:tilde-A}
  Let $\barA = \maketuple{\barQ,\barDelta,\barI,\barF}$ be a NOFRA as in~\autoref{cons:bar-A}. Then the NOFRA $\widetilde{A} =
  (\widetilde{Q},\widetilde{\delta},\widetilde{I},\widetilde{F})$ is given by
\begin{itemize}
\item states $\widetilde{Q} =
  J \times \pow(\set{1,\dots,m}) \times \names^{m}$, where $\Pow$ denotes the powerset;
\item initial states $\widetilde{I} = J_I \times \set{\emptyset} \times \names^{m}$ and final states  $\widetilde{F} = J_F \times \pow(\set{1,\dots,m}) \times \names^{m}$;
\item transitions defined as follows:
for a set $E$ of equations and $S \seq \set{1,\dots,m}$ let $E_S$ denote the restriction of $E$ to those equations whose left-hand side refers to a register in $S$:
  \[E_S = \setw{k = \bullet \in E}{k \in S} \cup \setw{k = \bark \in E}{k\in S}.\]
For $S,S'\seq \{1,\ldots,m\}$ we write $S\rightsquigarrow_E S'$ if
\[ S' = \{ k\in \{1,\ldots,m\} : \bullet=k \in E\} \cup \{ \bark\in \{1,\ldots,m\} : k=\bark\in E \text{ for some $k\in S$}\}. \]
Given $((j,S,q),b,(j',S',q'))\in \widetilde{Q}\times\At\times \widetilde{Q}$ we have the transition $\maketuple{j,S,q} \xra{b} \maketuple{j',S',q'}$ in $\widetilde{A}$ iff there exists some abstract transition $(j,E,j')\in \abs(\delta)$ such that (i) the triple $((j,q),b,(j',q'))$ is consistent with $(j,E_S,j')$, and (ii) $S\rightsquigarrow_E S'$.
\end{itemize} 
\end{construction}

\begin{rem}\label{rem:tildeA-trans}
Since property (i) only requires consistency with $(j,E_S,j')$, transitions can be modified arbitrarily outside of $S$ and $S'$: for every transition $\maketuple{j,S,q} \xra{b} \maketuple{j',S',q'}$ of $\widetilde{A}$ one also has the transitions $\maketuple{j,S,\ol{q}} \xra{b} \maketuple{j',S',\ol{q}'}$ for all $\ol{q},\ol{q}'\in \At^m$ such that $q_k=\ol{q}_k$ for $k\in S$ and $\ol{q}'_k=q'_k$ for $k\in S'$.
\end{rem}

\begin{lemma}\label{lem:barA-bs-tildeA}
The NOFRA $\barA$ and $\widetilde{A}$ are equivalent.
\end{lemma}

\begin{proof}
\emph{$L(\barA)\seq L(\widetilde{A})$:} Suppose that $b_1\ldots b_n\in L(\barA)$ with accepting run
\[ (j_0,q_0)\xto{b_1} (j_1,q_1)\xto{b_2} \cdots \xto{b_n} (j_n,q_n) \]
in $\barA$.
Then for all $r<n$ the transition $(j_{r},q_{r})\xto{b_{r+1}} (j_{r+1},q_{r+1})$ is consistent with some abstract transition $j_{r}\xto{E_{r+1}} j_{r+1}$ (\autoref{prop:trans-bar-A}). Hence it is also consistent with $(j_{r},({E_{r+1}})_S, j_{r+1})$ for every $S\seq \{1,\cdots,m\}$. It follows that we have the accepting run 
\[ (j_0,S_0,q_0)\xto{b_1} (j_1,S_1,q_1)\xto{b_2} \cdots \xto{b_n} (j_n,S_n,q_n) \]
in $\widetilde{A}$ where $S_0=\emptyset$ and $S_{r}\rightsquigarrow_{E_r} S_{r+1}$ for all $r<n$, whence $b_1\ldots b_n\in L(\widetilde{A})$.

\medskip\noindent \emph{$L(\widetilde{A})\seq L(\barA)$:} Suppose that $b_1\ldots b_n\in L(\widetilde{A})$ with accepting run
\[ (j_0,S_0,q_0)\xto{b_1} (j_1,S_1,q_1)\xto{b_2} \cdots \xto{b_n} (j_n,S_n,q_n) \]
in $\widetilde{A}$. By definition of the transitions of $\widetilde{A}$, for all $r<n$ there exists an abstract transition $j_{r}\xto{E_{r+1}} j_{r+1}$ of $A$ such that  $((j_{r},q_{r}),b_{r+1},(j_{r+1},q_{r+1}))$ is consistent with $(j_{r},(E_{r+1})_{S_{r}},j_{r+1})$, and moreover $S_{r}\rightsquigarrow_{E_r} S_{r+1}$ (where $S_0=\emptyset$). To prove $b_1\ldots b_n\in L(\barA)$ we employ \autoref{prop:abstract-lang-char}: we verify that the abstract run
\[ j_0\xto{E_1} j_1\xto{E_2} \cdots \xto{E_n} j_n \]
with its associated predicates $\Eq_i^{(k)}$ satisfies property \eqref{eq:eq-cond-1}. Thus let $r<n$, $k=\bullet$ in $E_{r+1}$ and $\Eq^{(i)}_k(r)$ for some $k$. By definition of $\Eq^{(i)}_k(r)$ and $\rightsquigarrow$, we have $b_i=(q_{r})_k$ and $k\in S_r$.
Since $k=\bullet\in (E_{r+1})_{S_r}$ and the triple $((j_r,q_r),b_{r+1},(j_{r+1},b_{r+1}))$ is consistent with $(j_r,(E_{r+1})_{S_r},j_{r+1})$, it follows that $b_i=(q_r)_k=b_{r+1}$, as required.
\end{proof}
\autoref{thm:positive-non-guessing} now follows from the above lemma and the following one:
\begin{lemma} \label{lem:pos-nofa-non-guessing}
The NOFRA $\widetilde{A}$ is equivalent to a non-guessing NOFRA.
\end{lemma}

\begin{proof}
We turn $\widetilde{A}$ into an equivalent non-guessing NOFRA $\widetilde{A}_\ng$ by first removing all guessing transitions, and then dealing with initial states with non-empty support. In more detail:
\begin{enumerate}
\item Let $\widetilde{A}_\ngt$ be the sub-NOFRA of $\widetilde{A}$ obtained by restricting to non-guessing transitions, i.e.~transitions $(j,S,q)\xto{b} (j',S',q')$ where $\supp q'\seq \supp q \cup \{b\}$. We claim that $L(\widetilde{A}_\ngt)=L(\widetilde{A})$. The left-to-right inclusion is clear. For the right-to-left inclusion, suppose that $b_1\cdots b_n\in L(\widetilde{A})$ with accepting run
\[ (j_0,S_0,q_0)\xto{b_1} (j_1,S_1,q_1)\xto{b_2} \cdots \xto{b_n} (j_n,S_n,q_n) \]
in $\widetilde{A}$. By \autoref{rem:tildeA-trans} we obtain another accepting run
\[ (j_0,S_0,\ol{q}_0)\xto{b_1} (j_1,S_1,\ol{q}_1)\xto{b_2} \cdots \xto{b_n} (j_n,S_n,\ol{q}_n) \]
where $\ol{q}_0=b_1^m$ and for $r=1,\ldots,n$ we put $(\ol{q}_r)_k=(q_r)_k$ if $k\in S_r$ and $(\ol{q}_r)_k=b_{r-1}$ if $k\not\in S_{r}$. Since all these transitions are non-guessing, this an accepting run in $\widetilde{A}_\ngt$, so $b_1\cdots b_n\in L(\widetilde{A}_\ngt)$.
\item Now let $\widetilde{A}_\ng$ emerge from $\widetilde{A}_\ngt$ by adding a new initial state $q_0$ with $\supp q_0 = \emptyset$ (which is also final if $J_I\cap J_F\neq \emptyset$), making all states of $\widetilde{A}_\ngt$ non-initial, and adding a transition $q_0\xto{b} (j',S',q')$ for each transition $(j,\emptyset,b^m)\xto{b} (j',S',q')$ of $\widetilde{A}_\ngt$ where $j\in J_I$. The NOFRA $\widetilde{A}_\ng$ is non-guessing and satisfies $L(\widetilde{A}_\ng)=L(\widetilde{A}_\ngt)$ by \autoref{rem:tildeA-trans}.\qedhere
\end{enumerate}
\end{proof}

\subsection*{Proof of \autoref{thm:pos-nofa-vs-pos-reg}}
The ``if'' direction follows from the fact that every register automaton admits an equivalent NOFA~\cite{BojanczykEA14} and from 

\begin{proposition}\label{prop:pos-reg-acc-pos-lang}
Every positive register automaton accepts a positive language.
\end{proposition}

\begin{proof}
Let $A$ be a positive register automaton and $w=a_1\ldots a_n\in L(A)$ with accepting run
\[ (c_0,r_0)\xto{a_1} (c_1,r_1)\xto{a_2} \cdots \xto{a_{n}} (c_n,r_n). \]
For $i=0,\ldots,n-1$ we have that $(c_i,r_i)\xto{a_{i+1}} (c_{i+1},r_{i+1})$ is consistent with some transition $c_i\xto{\phi_{i+1}} c'_{i+1}$. Then $(c_i,\rho^\star r_i)\xto{\rho a_{i+1}} (c_{i+1},\rho^\star r_{i+1})$ for every renaming $\rho\colon \At\to\At$ since this is also consistent with $c_i\xto{\phi_{i+1}} c'_{i+1}$. Hence $A$ accepts $\rho^\star w=\rho a_1\ldots \rho a_n$ via the run
\[ (c_0,\rho^\star r_0)\xto{\rho a_1} (c_1,\rho^\star r_1)\xto{\rho a_2} \cdots \xto{\rho a_{n}} (c_n,\rho^\star r_n). \]
This proves $\rho^\star w \in L(A)$, showing that $L(A)$ is a positive language.
\end{proof}

For the ``only if'' direction, suppose that $L\seq\Ats$ is a positive NOFA-recognizable language. Then $L$ is accepted by a NOFA of the form $\barA=(\barQ,\barDelta,\barI,\barF)$, cf.\ \autoref{cons:bar-A}, in particular $\barQ=J\times \At^{\#m}$ for some $J$ and $m$. We regard an equation as per \autoref{def:nofa-abstract}.\ref{def:nofa-abstract:1} as an equation in $\Phi$ by identifying
\[ k=\bullet \;\leftrightarrow\; (k,\mathrm{before})=\bullet,\quad \bullet=k \;\leftrightarrow\; \bullet=(k,\mathrm{after}),\quad k=\bark \;\leftrightarrow\; (k,\mathrm{before})=(\bark,\mathrm{after}),  \]
and turn $\barA$ into a positive register automaton $\barA_\reg=(J_\reg,m,\delta_\reg, \{j_{0,\reg}\},F_\reg)$ as follows:
\begin{itemize}
\item The set of control states is $J_\reg = J \cup \{ j_{0,\reg} \}$ where $j_{0,\reg}\not\in J$;
\item $j_{0,\reg}$ is the only initial state;
\item every state in $J_F$ is final; additionally $j_{0,\reg}$ is final if $J_I\cap J_F\neq \emptyset$;
\item for every abstract transition $j\xto{E} j'$ of $A$, the automaton $\barA_\reg$ contains the transition $j\xto{\bigwedge E} j'$, where $\bigwedge E$ is the conjunction of all equations in $E$ (note that $\bigwedge \emptyset = \mathrm{true}$);
\item for every abstract transition $j_{0}\xto{E} j$ of $A$ where $j_0\in J_I$, the automaton $\barA_\reg$ contains the transition $j_{0,\reg}\xto{\phi_E} j$ where $\phi_E=(\bullet=k)$ if $\bullet=k$ in $E$, and $\phi_E=\mathrm{true}$ otherwise. 
\end{itemize}

We claim that $L=L(\barA_\reg)$. For the inclusion ($\seq$), let $w=b_1\ldots b_n\in L=L(\barA)$. Then in $\barA$ we have an accepting run
\[
(j_0,q_0)\xto{b_1} (j_1,q_1)\xto{b_2}\cdots \xto{b_n} (j_n,q_n).
\]
whose transitions are consistent with some accepting abstract run
\[ j_0\xto{E_1} j_1 \xto{E_2} \cdots  \xto{E_n}  j_n. \]
It follows that the register automaton $\barA_\reg$ admits the transitions 
\[ j_{0,\reg} \xto{\phi_{E_1}} j_1 \xto{\bigwedge E_2} \cdots  \xto{\bigwedge E_n}  j_n \]
and that
\[ (j_{0,\reg},\bot^m)\xto{b_1} (j_1,q_1)\xto{b_2}\cdots \xto{b_n} (j_n,q_n)\]
is an accepting run consistent with them. Therefore $w\in L(\barA_\reg)$.

\medskip\noindent
For the inclusion ($\supseteq$), let $w=b_1\cdots b_n\in L(\barA_\reg)$ with accepting run
\[ (j_{0,\reg},\bot^m)\xto{b_1} (j_1,r_1)\xto{b_2} \cdots \xto{n_{n}} (j_n,r_n), \]
in $\barA_\reg$, where $r_i\in (\At\cup \{\bot\})^m$. Note that $j_1,\ldots,j_n\in J$. By definition of the transitions of $\barA_\reg$, there exists an abstract transition $j_0\xto{E_1} j_1$ of $A$ such that $(j_{0,\reg},\bot^m)\xto{b_1} (j_1,r_1)$ is consistent with the transition $j_{0,\reg}\xto{\phi_E} j_1$ of $\barA_\reg$, and for $i=1,\ldots,n$ there exists an abstract transition $j_{i-1}\xto{E_i} j_i$ of $\barA$ such that $(j_{i-1},r_{i-1})\xto{b_i} (j_i,r_i)$
is consistent with the transition $j_{i-1}\xto{\bigwedge E_i} j_i$ of $\barA_\reg$.
Now choose $q_0,\ldots,q_n\in \At^m$ as follows:
\begin{itemize}
\item For $i=1,\ldots,n$ choose $q_i$ such that $q_{i,k}=r_k$ whenever $r_k\neq \bot$.
\item Choose $q_0$ such that $q_{0,k}=q_{1,\overbar{k}}$ if $k=\bark$ in $E_1$, and $q_{0,k}=b_1$ if $k=\bullet$ in $E_1$. 
\end{itemize}
Then for each $i=1,\ldots,n$ we have the transition $(j_{i-1},q_{i-1})\xto{b_i} (j_i,q_i)$ in $\barA$, as it consistent with the abstract transition $j_{i-1}\xto{E_i} j_i$. Therefore
\[
(j_0,q_0)\xto{b_1} (j_1,q_1)\xto{b_2}\cdots \xto{b_n} (j_n,q_n).
\]
is an accepting run in $\barA$, showing that $w\in L(\barA)=L$.

\subsection*{Details for~\autoref{rem:pos-reg-aut-vs-fsuba}}

We provide more details on the stated equivalence between positive register automata and a a version finite-state unification-based automata (FSUBA)~\cite{tal99,kt06}. We first recall
the definition of the latter.

\begin{notation}
  For any natural number $r$, we denote by $\ul{r}$ the set of all natural numbers between $1$ and $r$, inclusively. $\ul{0}$ denotes the empty set.
  We denote by $\Perm(\ul{r})$ the group of all permutations on the finite set $\ul{r}$ and note that there is an obvious
  group action of $\Perm(\ul{r})$ on $\names^{\# r}$ that is defined as follows:
  For any $\pi \in \Perm(\ul{r})$ and $w \in \names^{\# r}$, we define the word $\pi\star w\in \At^{\#r}$ by $(\pi \star w)_k=w_{\pi(k)}$ for $k\in \ul{r}$. Note that this action is compatible with the $\Perm(\At)$-action: $\pi\star (\rho\cdot w) = \rho\cdot (\pi\star w)$ for $\rho\in \Perm(\At)$.
\end{notation}

\begin{defn} \label{def:fsuba}
  A \emph{finite-state unification-based automaton (FSUBA)} is a quintuple $A = \maketuple{Q,m,\mu,q_0,F}$ where
  $Q$ is a finite set of control states, $m \in \Nat$ is the number of registers (numbered from $1$ to $m$), $q_0$
  is the initial state, $F \seq Q$ the set of final states, and $\mu \seq Q \times \ul{m} \times \pow(\ul{m}) \times Q$
  is the transition relation. Here $\pow$ denotes the powerset.
  A \emph{configuration} of $A$ is a pair $\maketuple{q,w}$ of a state $q \in Q$ and a word $w \in (\names \cup \set{\bot})^m$ corresponding to
  a partial assignment of data values to the registers. The initial configuration is $\maketuple{q_0,\bot^m}$, final configurations are all
  $\maketuple{q_f,w}$ with $q_f \in F$. We let $Q^c$ and $F^c$ denote the sets of configurations and final configurations, respectively. Given an input $a \in \names$
  and configurations $\maketuple{q,w}, \maketuple{q',w'}$ we write $\maketuple{q,w} \xra{a} \maketuple{q',w'}$ if this move
  is \emph{consistent} with some transition $\maketuple{q,k,T,q'}$, which means that the following conditions are satisfied:
  (i) $w_k \in \set{\bot, a}$; (ii) $k \notin T \implies w'_k = a$; (iii) $\forall j \in T.\ w_j = \bot$; and
  (iv) $\forall j \notin T \cup \set{k}.\ w'_j = w_j$.
  We denote the induced move relation on $Q^c \times \names \times Q^c$ by $\mu^c$.
  A word $a_1 \cdots a_n \in \Ats$ is \emph{accepted} by $A$ if there exists an accepting run for it, viz.~a sequence of configurations
  $\maketuple{q_0,\bot^m} \xra{a_1} \maketuple{q_1,{w_1}} \xra{a_2} \cdots \xra{a_n} \maketuple{q_n,{w_n}}$, where $q_n \in F$.
  We write $L(A) \seq \Ats$ for the language of accepted words.
\end{defn}

\begin{rem}
  In comparison to the original definition of Tal~\cite{tal99,kt06}) we do not allow an initial assignment
  of the registers, since otherwise the accepted languages are not equivariant but only finitely supported. Doing this also
  suppresses the \enquote{read-only} alphabet, a subset of the data values occurring in the initial assignment.
\end{rem}

\begin{rem} \label{rem:fsuba-is-nofa}
  Every FSUBA $A = \maketuple{Q,m,\mu,q_0,F}$ can be translated into an expressively equivalent NOFA
  $N = \maketuple{Q^c,\mu^c,\set{\maketuple{q_0,\bot^m}},F^c}$, i.e. the configurations of the FSUBA are simply regarded as states of a NOFA. The corresponding NOFA has the set of configurations $Q^c$ as states,
  the singleton set $\set{\maketuple{q_0,\bot^m}}$ as initial states, and $F^c$ as final states. An accepting run of $A$ is then precisely an accepting run of $N$. 
\end{rem}

The structural difference between NOFA and FSUBAs is the inherent \enquote{non-guessing}
of FSUBAs and the fact that FSUBAs cannot move data values from one register to another; e.g., if register $2$ contains the data value $a$, then it cannot be moved to register $3$ in the next step, which is possible with a NOFA. The first issue will be fixed by use of~\autoref{thm:positive-non-guessing},
while for the second we will turn a NOFA $A$ given by~\autoref{rem:strong} into a \emph{rigid} NOFA, where the data value contained in a register is never moved to another register:

\begin{definition} \label{def:rigid}
  A NOFA $A = \maketuple{Q,\delta,I,F}$ with states $Q=J\times \At^{\#m}$ is \emph{rigid} if for every transition
  $\maketuple{j,p} \xra{a} \maketuple{j',p'}$ and for every $b\in \supp p\cap \supp p'$, there exists $k\in \ul{m}$ such that $p_k=b=q_k$.
\end{definition}

\begin{rem}\label{rem:eqs-rigid-nofa}
Hence rigid NOFA are those whose abstract transitions are of the form $\bullet = k$,
  $k = \bullet$, or $k = k$. The construction below turns any NOFA into a rigid one. The idea is to keep track, via the control state, which data values have changed their register.
\end{rem}

\begin{construction}\label{cons:non-perm}
  Let $A = \maketuple{Q, \delta, I, F}$ be a NOFA with states $Q=J\times \At^{\#m}$. We construct the rigid NOFA $A_\rg = \maketuple{Q_\rg, \delta_\rg, I_\rg, F_\rg}$
  given by
  \begin{itemize}
    \item states $Q_\rg = J \times \Perm(\ul{m}) \times \names^{\# m}$;
    \item initial states $I_\rg = J_I \times \set{\id_{\ul{m}}} \times \names^{\# m}$ and final states $F_\rg = J_F \times \Perm(\ul{m}) \times \names^{\# m}$;
    \item transitions defined as follows: Given $\maketuple{\maketuple{j,\pi,p},a,\maketuple{j',\pi',p'}} \in Q_\rg \times \names \times Q_\rg$ we have the transition
      $\maketuple{j,\pi,p} \xra{a} \maketuple{j',\pi',p'}$ in $A_\rg$ iff (i) $\maketuple{j,\pi\star p} \xra{a} \maketuple{j',\pi'\star p'}$ is a transition in $A$
      and (ii) for every $b \in \supp(p) \cap \supp(p')$ there exists $k \in \ul{m}$, such that $p_k = b = p'_k$.
  \end{itemize}

Note that by property (ii) of transitions, the NOFA $A_\rg$ is rigid.
\end{construction}

\begin{lemma}\label{lem:rigid-nofa}
  The NOFA $A$ and $A_\rg$ are equivalent.
\end{lemma}
\begin{proof}
  \emph{$L(A_\rg)\seq L(A)$:} Every accepting run
\[ 
(j_0,\pi_0,p_0)\xto{a_1} (j_1,\pi_1,p_1)\xto{a_2} \cdots \xto{a_n} (j_n,\pi_n,p_n)
\]
of $A_\rg$ yields the following accepting run of $A$:
\[ 
(j_0,\pi_0\star p_0)\xto{a_1} (j_1,\pi_1\star p_1) \xto{a_2} \cdots \xto{a_n} (j_n,\pi_n\star p_n).
\]
\medskip\noindent \emph{$L(A)\seq L(A_\rg)$:} Given $a_1\ldots a_n \in L(A)$ with an accepting run 
\[\maketuple{j_0,q_0} \xra{a_1} \cdots \xra{a_{n}} \maketuple{j_n,q_n}\]
  in $A$, we inductively construct an accepting run 
\[\maketuple{j_0,\pi_0,q_0'} \xra{a_1} \cdots \xra{a_n} \maketuple{j_n,\pi_n,q_n'}\]
in $A_\rg$ such that $\pi_r \star q_r' = q_r$ for $r=0,\ldots,n$. We put $\maketuple{j_0,\pi_0,q_0'}=\maketuple{j_0,\id_{\ul{m}},q_0}$, which clearly fulfills $\pi_0\star q_0'=q_0$. Now suppose that that $0\leq r<n$ and that the first $r$ transitions with the required properties have been constructed. Then we construct the next transition $\maketuple{j_{r},\pi_{r},q_{r}'} \xra{a_{r+1}} \maketuple{j_{r+1},\pi_{r+1},q_{r+1}'}$ as follows. 
    Let $\mathcal{I} := \setw{k\in \ul{m}}{q_{r,k}' \in \supp(q_{r+1})}$. Choose $q_{r+1}'\in \At^{\#m}$ such that $\supp q_{r+1}'=\supp q_{r+1}$ and $q_{r+1,k}'=q_{r,k}'$ for $k\in \mathcal{I}$. Choose moreover a permutation $\pi_{r-1}\in \Perm(\ul{m})$ such that $\pi\star q_{r+1}'=q_{r+1}$. Note that since also $\supp q_r' = \supp q_r$, every data value in $\supp q_{r}'\cap \supp q_{r+1}'$ is equal to $q_{r,k}'$ for some $k\in \mathcal{I}$. Therefore $\maketuple{j_{r},\pi_{r},q_{r}'} \xra{a_{r+1}} \maketuple{j_{r+1},\pi_{r+1},q_{r+1}'}$ is a transition of $A_\rg$, as required. 
\end{proof}

\begin{rem}
  If we apply the construction $A\mapsto \widetilde{A}_\ng$ of~\autoref{lem:barA-bs-tildeA} and~\autoref{lem:pos-nofa-non-guessing} to a rigid NOFA, we see that for every transition
  $\maketuple{j,S,p} \xra{a} \maketuple{j',S',q}$ the set $S'$ is of the form $S \setminus T$ or $(S\setminus T)\cup \{k\}$ where $T\seq S$ and $k\in \ul{m}$. This follows directly from~\autoref{rem:eqs-rigid-nofa}.
\end{rem}
We will now show how to translate $\widetilde{A}_\ng$ into an equivalent FSUBA. The idea is to maintain in the control state, in addition to the set $S$ of ``relevant'' registers, a subset $R\seq S$ containing all registers that will eventually be compared with a future input value. Registers outside of $R$ then may be deleted.

\begin{construction} \label{cons:nofra-fsuba}
  Let $A = \maketuple{Q,\delta,I,F}$ be a rigid NOFA and $\widetilde{A}_{\ng} = \maketuple{\widetilde{Q},\widetilde{\delta}_{\ng},\set{q_0},\widetilde{F}_{\ng}}$
  be the corresponding non-guessing NOFRA of~\autoref{lem:pos-nofa-non-guessing} with states $J \times \pow(\ul{m}) \times \names^m \cup \set{q_0}$.
  Then the FSUBA $A_{\fsuba} = \maketuple{C, m + 1, \mu, q_0, F_{\fsuba}}$ with $m+1$ registers is given as follows:
  \begin{itemize}
    \item control states $C = \set{q_0}\cup \{ (j,R,S) : j\in J\text{ and } R\seq S\seq \ul{m} \}$;
    \item all states $(j,R,S)$ where $j\in J_F$ are final, and additionally $q_0$ is final if $J_I\cap J_F\neq \emptyset$;
    \item for each abstract transition $j\xto{E} j'$ in $\abs(\delta)$ and all pairs of sets $R,S\seq \ul{m}$ such that $R\seq S$ and $k=\bullet\in E_S$ implies $k\in R$, we have the following transitions:
\begin{enumerate}
\item if $k=\bullet$ in $E_S$ (hence $E_R$) or $\bullet=k$ in $E$ for some (unique) $k\in\ul{m}$, then  $\mu$ contains the transitions
\[ ((j,R,S),k,\ol{R'},(j',R',S')) \qquad\text{and}\qquad ((j,R,S),k,\ol{R'\setminus \{k\}},(j',R'\setminus \{k\},S')),  \]
where $R\rightsquigarrow_E R'$, $S\rightsquigarrow_E S'$, and $\ol{T}=\ul{m+1}\setminus T$ for any subset $T\seq \ul{m+1}$.
\item otherwise, $\mu$ contains the transition
\[ ((j,R,S),m+1,\ol{R'},(j',R',S')).  \]
\end{enumerate}
\item Additionally, for every transition $((j_0,\emptyset,\emptyset),k,T,(j,R,S))$ where $j_0\in J_I$, $k\in \ul{m+1}$ and $R,S,T\seq \ul{m+1}$, we have the transition $(q_0,k,T,(j,R,S))$.
  \end{itemize}
\end{construction}

\begin{lemma} \label{lem:pos-non-guessing-nofra-fsuba}
 The automata $\widetilde{A}_{\ng}$ and $A_{\fsuba}$ are equivalent.
\end{lemma}
\begin{proof}
  \emph{$L(\widetilde{A}_{\ng})\seq L(A_{\fsuba})$:} Let $a_1\cdots a_n\in L(\widetilde{A}_\ng)=L(\widetilde{A})$ with accepting run
\[ (j_0,S_0,v_0)\xto{a_1} (j_1,S_1,v_1)\xto{a_2} \cdots \xto{a_n} (j_n,S_n,v_n) \]
in $\widetilde{A}$. Each transition $(j_{r-1},S_{r-1},v_{r-1})\xto{a_r} (j_r,S_r,v_r)$, $r=1,\ldots,n$, is induced by some abstract transition $j_{r-1}\xto{E_r} j_r$ in $\abs(\delta)$, that is, the triple $((j_{r-1},v_{r-1}),a_r,(j_r,v_r))$ is consistent with $(j_{r-1},(E_r)_{S_{r-1}},j_r)$ and $S_{r-1}\rightsquigarrow S_r$. We now construct an  accepting run 
\[  
(q_0,w_0) \xto{a_1} ((j_1,R_1,S_1),w_1)\xto{a_2} \cdots \xto{a_n} ((j_n,R_n,S_n),w_n)
\]
in $A_\fsuba$ whose data is defined as follows for $r=1,\ldots,n$:
\begin{itemize}
\item $R_r$ is the set of registers in $S_r$ that will be compared with some later input, that is,
\[ R_r = \{\,k\in S_r : \exists s\in \{r+1,\ldots,n\} : k=k\in E_{r+1},\ldots E_{s-1} \text{ and } k=\bullet\in E_s\,\}. \]
\item $w_0=\bot^{m+1}$, and $w_{r,k}=v_{r,k}$ for $k\in R_r$ and $w_{r,k}=\bot$ for $k\in \ol{R_r}$.
\end{itemize}

Note that the first move is equivalent to having a move $((j_0,R_0,S_0),w_0) \xto{a_1} ((j_1,R_1,S_1),w_1)$ where $j_0\in J_I$ and $S_0=R_0=\emptyset$. Let us now verify that $((j_{r-1},R_{r-1},S_{r-1}),w_{r-1}) \xto{a_r} ((j_r,R_r,S_r),w_r)$ is indeed a valid move for $r=1,\ldots,n$, i.e.\ consistent with some transition of $A_\fsuba$. We distinguish two cases:
\begin{itemize}
\item If $k=\bullet$ in $(E_r)_{S_{r-1}}$ or $\bullet=k$ in $E_r$ for some $k\in \ul{m}$, take the transition \begin{equation}\label{eq:fsuba-trans} ((j_{r-1},R_{r-1},S_{r-1}),k,\ol{R_r},(j_r,R_r,S_r)).\end{equation}
Note that this transition is induced by the abstract transition $j_{r-1}\xto{E_r} j_r$ and the pair $R_{r-1},S_{r-1}$ as per \autoref{cons:nofra-fsuba}: if $R_{r-1}\rightsquigarrow_{E_r} R'$ then either $R_r=R'$ or $R_r=R'\setminus \{k\}$
by definition of the sets $R_{r-1}$ and $R_r$, and moreover if $k=\bullet$ in $(E_r)_{S_{r-1}}$ then $k\in R_{r-1}$ by definition of $R_{r-1}$. The move $((j_{r-1},R_{r-1},S_{r-1}),w_{r-1}) \xto{a_r} ((j_r,R_r,S_r),w_r)$ clearly satisfies the consistency conditions (i)--(iv) of \autoref{def:fsuba} w.r.t.\ \eqref{eq:fsuba-trans}.
\item Otherwise, take the transition
\begin{equation}\label{eq:fsuba-trans2} ((j_{r-1},R_{r-1},S_{r-1}),m+1,\ol{R_r},(j_r,R_r,S_r)).
\end{equation}
Note again that this transition is induced by the abstract transition $j_{r-1}\xto{E_r} j_r$ and the pair $R_{r-1},S_{r-1}$ as per \autoref{cons:nofra-fsuba}: one has $R_{r-1}\rightsquigarrow R_r$ by definition of the sets $R_{r-1}$ and $R_r$. The move $((j_{r-1},R_{r-1},S_{r-1}),w_{r-1}) \xto{a_r} ((j_r,R_r,S_r),w_r)$ clearly satisfies the consistency conditions (i)--(iv) of \autoref{def:fsuba} w.r.t.\ \eqref{eq:fsuba-trans2}.
\end{itemize}

\noindent $L(A_\fsuba)\seq L(\widetilde{A}_\ng)$: Let $a_1\cdots a_n\in L(A_\fsuba)$ with accepting run 
\[  
(q_0,w_0) \xto{a_1} ((j_1,R_1,S_1),w_1)\xto{a_2} \cdots \xto{a_n} ((j_n,R_n,S_n),w_n).
\]
The first move is equivalent to having a move $((j_0,R_0,S_0),w_0) \xto{a_1} ((j_1,R_1,S_1),w_1)$ where $j_0\in J_I$ and $S_0=R_0=\emptyset$. For $r=1,\ldots,n$ the move $((j_{r-1},R_{r-1},S_{r-1}),w_{r-1}) \xto{a_r} ((j_r,R_r,S_r),w_r)$ is consistent with a transition of $A_\fsuba$ induced by some abstract transition $j_{r-1}\xto{E_r}j_r$ in $\abs(\delta)$ and the pair $R_{r-1}, S_{r-1}$. This yields the accepting run
\[ (j_0,S_0,v_0)\xto{a_1} (j_1,S_1,v_1)\xto{a_2} \cdots \xto{a_n} (j_n,S_n,v_n)\]
in $\widetilde{A}$ where  $v_0=a_1^m$, and $v_r\in \At^m$ for $r=1,\ldots,n$ is defined as follows:
\begin{itemize}
\item for $k\in R_r$ put $v_{r,k}=w_{r,k}$;
\item for $k\not\in S_r$ put $v_{r,k}=a_r$;
\item for $k\in S_r\setminus R_r$, if $k=k$ in $(E_{r})_{S_{r-1}}$ then take $v_{r,k}=v_{r-1,k}$, and if $\bullet=k$ in $E_{r}$ then take $v_{r,k}=a_k$. Note that at least one of these cases must occur because $S_{r-1}\rightsquigarrow_{E_r} S_r$. Moreover, if both $k=k$ and $\bullet=k$ in $E_{r}$ then also $k=\bullet$ in $E_{r}$ and moreover $k\in S_{r-1}$, whence $k\in R_{r-1}$. Therefore $v_{r-1,k}=a_r$ as the move $((j_{r-1},R_{r-1},S_{r-1}),w_{r-1}) \xto{a_1} ((j_r,R_r,S_r),w_r)$ is consistent with a transition of $A_\fsuba$ induced by $j_{r-1}\xto{E_r}j_r$ and $R_{r-1}, S_{r-1}$, which is of the form $((j_{r-1},R_{r-1},S_{r-1}),k,\cdots)$ since $\bullet=k$ in $E_r$. Hence $v_{r,k}$ is properly defined.
\end{itemize}

It remains to show that $((j_{r-1},S_{r-1}),v_{r-1}) \xto{a_1} ((j_r,S_r),v_r)$ is a valid transition of $\widetilde{A}$. By definition we have $S_{r-1}\rightsquigarrow_{E_r} S_r$, so we only need to show that the transition is consistent with $(E_r)_{S_{r-1}}$. Indeed:
\begin{itemize}
\item If $k=\bullet$ in $(E_r)_{S_{r-1}}$ (hence $k\in R_{r-1}$) then a transition of $A_\fsuba$ induced by $j_{r-1}\xto{E_r}j_r$ and $R_{r-1}, S_{r-1}$ is of the form $((j_{r-1},R_{r-1},S_{r-1}),k,\cdots)$, and since by assumption the move $((j_{r-1},R_{r-1},S_{r-1}),w_{r-1}) \xto{a_r} ((j_r,R_r,S_r),w_r)$ is consistent with it, we have $v_{r-1,k}=w_{r-1,k}=a_r$.
\item Now suppose that $k=k$ in $(E_r)_{S_{r-1}}$ but not $k=\bullet$ in $(E_r)_{S_{r-1}}$ (hence not $\bullet=k$ in $E_r$). We distinguish two subcases. If $k\in S_{r-1}\setminus R_{r-1}$, then $k\in S_r\setminus R_r$, so by definition of $v_{r,k}$ we get $v_{r-1,k}=v_{r,k}$. If $k\in R_{r-1}$, then the transition of $A_\fsuba$ induced by $j_{r-1}\xto{E_r}j_r$ and $R_{r-1}, S_{r-1}$ is given by $((j_{r-1},R_{r-1},S_{r-1}),m+1,\ol{R_r},(j_r,R_r,S_r))$ where $R_{r-1}\rightsquigarrow_{E_r} R_r$ and $S_{r-1}\rightsquigarrow_{E_r} S_r$. Since the move $((j_{r-1},R_{r-1},S_{r-1}),w_{r-1}) \xto{a_r} ((j_r,R_r,S_r),w_r)$ is consistent with that transition, we have $v_{r-1,k}=w_{r-1,k}=w_{r,k}=v_{r,k}$. 
\end{itemize}

This concludes the proof.
\end{proof}

We conclude:

\begin{theorem}
Positive register automata and FSUBA are equivalent.
\end{theorem}
\begin{proof}
  Every language accepted by some FSUBA as in~\autoref{def:fsuba} is clearly positive, and by~\autoref{rem:fsuba-is-nofa}
  it is NOFA-recognizable, hence by~\autoref{thm:pos-nofa-vs-pos-reg} it is accepted by some positive register automaton. Conversely, every language accepted by positive register automaton is positive and NOFA-recognizable (using \autoref{thm:pos-nofa-vs-pos-reg} again), and so \autoref{thm:positive-non-guessing} and \autoref{lem:pos-non-guessing-nofra-fsuba} show that is accepted by an FSUBA.
\end{proof}

\subsection*{Details for \autoref{rem:mso-vs-nofa}}
We prove that the language $L\seq \Ats$ of all words where no data value occurs exactly once is not NOFA-recognizable. Suppose that $L$ is recognized by a NOFA $A$, and let $m\in \Nat$ such that every state has a support of size $m$. Choose $m+1$ distinct names $a_1,\ldots,a_{m+1}$. Then the word $w=a_1\ldots a_{m+1}a_1\ldots a_{m+1}$ lies in $L$, hence it admits an accepting run
\[ q_0\xto{a_1} \cdots \xto{a_{m+1}} q_{m+1}\xto{a_1} q'_1 \xto{a_2} \cdots \xto{a_{m+1}} q'_{m+1}. \]
Since $\supp q_{m+1}$ has at most $m$ elements, some $a_j$ is fresh for $q_{m+1}$. Choose a name $a_j'$ fresh for $a_1,\ldots,a_{m+1}$ and $q_{m+1}$, and let $\pi = (a_j\, a_j')$. Then by equivariance of transitions we have the following accepting run: 
\[ \pi\cdot q_0\xto{a_1} \cdots \xto{a_j'} \pi\cdot q_{j} \xto{a_{j+1}} \cdots \xto{a_{m+1}} \pi\cdot q_{m+1} = q_{m+1} \xto{a_1} q'_1 \xto{a_2} \cdots \xto{a_{m+1}} q'_{m+1}. \]
This means that the word $a_1\cdots a_j'\cdots a_{m+1} a_1\cdots a_j\cdots a_{m+1}$ is accepted by $A$ although it does not lie in $L$, a contradiction.

\subsection*{Proof of \autoref{thm:mso-vs-nofa}}
The ``if'' statement follows from \autoref{prop:mso-to-pos}. For the ``only if'' statement, suppose that $L\seq \Ats$ is a positive NOFA-recognizable data language; then $L$ is accepted by a NOFRA $\barA=(\barQ,\barDelta,\barI,\barF)$ as given by \autoref{cons:bar-A}. Our task is to construct an $\MSOep$-sentence $\phi$ such that, for all $w\in \Ats$,
\[ \text{$\barA$ accepts $w$}\qquad\iff\qquad \text{$w$ satisfies $\phi$}. \]
We shall make use of the characterization of accepted words from \autoref{prop:abstract-lang-char}: when interpreted over $w=b_1\ldots b_n$, the sentence $\phi$ will state existence of an abstract run satisfying condition \eqref{eq:eq-cond-1}. For this purpose, we introduce for every abstract transition $(j,E,j')\in \abs(\delta)$ a second-order variable $R_{(j,E,j')}$ with the intended interpretation
\[ R_{(j,E,j')}(x) \quad\hat{=}\quad \text{$(j,E,j')$ is the $x$-th transition of an accepting abstract run of $\barA$}.\]
The sentence $\phi$ is then given as follows:
\[ \phi\;\;=\;\; \vec{\exists}_{(j,E,j)\in \abs(\delta)} R_{(j,E,j')}.\, \phi_{\mathrm{run}} \wedge \forall x.\, \exists \Eq^{(x)}_1\,\ldots\,\exists \Eq^{(x)}_m.\, \phi_{\mathrm{aux}} \wedge \phi_{\mathrm{eq}}. \]
Here $\vec{\exists}_{(j,E,j)\in \abs(\delta)} R_{(j,E,j')}$ denotes the concatenation of all $\exists R_{(j,E,j')}$ where $(j,E,j')\in \abs(\delta)$, and we make the convention that quantifiers have maximal scope. The subformula $\phi_{\mathrm{run}}$ ensures that the variables $R_{(j,E,j')}$ define an accepting abstract run of $\barA$; the subformula $\phi_{\mathrm{aux}}$ ensures that the second-order variables $\Eq^{(x)}_k$ ($k=1,\ldots,m$) are interpreted as the auxiliary predicates of \autoref{not:xik-pred}; the subformula $\phi_{\mathrm{eq}}$ states the equality condition \eqref{eq:eq-cond-1}. In more detail, the three subformulas are given as follows:

\medskip\noindent\underline{Definition of $\phi_{\mathrm{run}}$:}
\begin{align*}
 \phi_{\mathrm{run}}\;\; = \;\;& \forall x. \bigvee_{(j,E,j')\in \abs(\delta)} R_{(j,E,j')}(x)\\
&\wedge \;\; \forall x.\bigwedge_{(j,E,j')\neq (\ol{j},\overbar{E},\ol{j}')\in \abs(\delta)} \neg  \big(\,R_{(j,E,j')}(x) \wedge R_{(\ol{j},\overbar{E},\ol{j}')}(x)\,\big) \\
&\wedge \;\; \forall x.\forall y.\, \mathrm{succ}(x,y) \implies \bigvee_{(j,E,j'), (j',E',j'')\in \abs(\delta)}
R_{(j,E,j')}(x) \wedge R_{(j',E',j'')}(y) \\
&\wedge \;\; \forall x.\, \mathrm{first}(x) \implies \bigvee_{\substack{(j,E,j')\in \abs(\delta)\\ j\in J_I}} R_{(j,E,j')}(x) \\
&\wedge \;\; \forall x.\, \mathrm{last}(x) \implies \bigvee_{\substack{(j,E,j')\in \abs(\delta)\\ j'\in J_F}} R_{(j,E,j')}(x)
\end{align*}
As usual $\psi\Rightarrow \xi$ means $\neg\psi \vee \xi$, and the formulas $\mathrm{succ}(x,y)$, $\mathrm{first}(x)$ and $\mathrm{last}(x)$ define the successor relation and the first and last position, respectively:
\[
 \mathrm{succ}(x,y) = x<y \,\wedge\, \neg\exists z.\, x<z \, \wedge\, z<y,\qquad
 \mathrm{first}(x) = \neg\exists y.\, y<x,\qquad
 \mathrm{last}(x) = \neg\exists y.x<y.
\]
The first two lines assert that every position $x$ is associated with a unique abstract transition of $A$, and the remaining part that these transitions form an accepting abstract run.

\medskip\noindent\underline{Definition of $\phi_{\mathrm{aux}}$:}
\begin{align*}
\phi_{\mathrm{aux}} \;\;=\;\; & \bigwedge_{k=1}^m\, \bigwedge_{\substack{(j,E,j')\in \abs(\delta)\\ \bullet=k\in E}} \big(\,R_{(j,E,j')}(x)\implies \Eq^{(x)}_k(x)\,\big) \\
& \wedge\;\;  \forall y.\,\forall z.\, \mathrm{succ}(y,z) \implies \bigwedge_{k,\overbar{k}=1}^m\, \bigwedge_{\substack{(j,E,j')\in \abs(\delta)\\ k=\overbar{k}\in E}} \big(\,R_{(j,E,j')}(z)\,\wedge\, \Eq^{(x)}_k(y) \implies \Eq^{(x)}_{\overbar{k}}(z)\,\big) \\
& \wedge\;\; \forall \ol{\Eq}^{(x)}_1\,\cdots\, \forall \ol{\Eq}^{(x)}_m.\, \big[\, \bigwedge_{k=1}^m\, \bigwedge_{\substack{(j,E,j')\in \abs(\delta)\\ \bullet=k\in E}} \big(\,R_{(j,E,j')}(x)\implies \ol{\Eq}^{(x)}_k(x)\,\big) \\
& \wedge\;\;  \forall y.\,\forall z.\, \mathrm{succ}(y,z) \implies \bigwedge_{k,\overbar{k}=1}^m \, \bigwedge_{\substack{(j,E,j')\in \abs(\delta)\\ k=\overbar{k}\in E}} \big(\,R_{(j,E,j')}(z)\,\wedge\, \ol{\Eq}^{(x)}_k(y) \implies \ol{\Eq}^{(x)}_{\overbar{k}}(z)\,\big)\,\big] \\
& \implies \bigwedge_{k=1}^m \forall y.\, \Eq^{(x)}_k(y) \implies \ol{\Eq}^{(x)}_k(y).
\end{align*} 
The first two lines state that the predicates $\Eq^{(x)}_k$ satisfy the two clauses of \autoref{not:xik-pred}, and the remaining part of the formula asserts that they are minimal with that property. This entails that $\Eq^{(x)}_k$ is precisely the inductively defined predicate of \autoref{not:xik-pred}.

\medskip\noindent\underline{Definition of $\phi_{\mathrm{eq}}$:} 
\begin{align*}
 \phi_{\mathrm{eq}} \;\;=\;\; & \forall y.\forall z.\,\mathrm{succ}(y,z) \implies \bigwedge_{k=1}^m\, \bigwedge_{\substack{(j,E,j')\in \abs(\delta)\\ k=\bullet\in E}} \big(\,R_{(j,E,j')}(z) \wedge \Eq^{(x)}_k(y) \implies x\sim z\,\big).
%\;\;&\wedge \;\; \big[\,\forall y.\bigwedge_{k=1}^m\, \bigwedge_{\substack{(j,E,j')\in \abs(\delta)\\ \bullet=k\in E}} \big(\,R_{(j,E,j')}(y) \wedge \Eq^{(x)}_k(y) \implies x\sim y\,\big)\,\big]
\qedhere
\end{align*} 

\subsection*{Proof of \autoref{prop:orbit-finite-vs-fp}}
We start with two preliminary remarks on directed colimits,
\begin{rem}\label{rem:directed-colimits}
Given a directed diagram $D\colon I\to \Set$, its colimit cocone $D_i\xto{c_i} \colim D$ ($i\in I$) is characterized by two properties: (i) the morphisms $c_i$ are jointly surjective (every element of $\colim D$ lies in the image of some $c_i$), and (ii) for every $i\in I$ and $x,y\in D_i$ such that $c_i(x)=c_i(y)$, there exists $j\geq i$ in $I$ such that $D_{i,j}(x)=D_{i,j}(y)$, where $D_{i,j}=D(i\to j)$ is the connecting morphism induced by the unique arrow $i\to j$ in $I$. The same characterization applies to directed colimits in $\Nom$ and $\RnNom$, since colimits in these two categories are formed at the level of underlying sets. 
\end{rem}

\begin{rem}
Recall that an object $X$ of a category $\C$ is finitely presentable if the functor $\C(X,-)\colon \C\to\Set$ preserves directed colimits. By \autoref{rem:directed-colimits}, this means precisely that for every directed diagram $D\colon I\to \C$ with colimit cocone  $c_i\colon D_i\to \colim D$ ($i\in I$),
\begin{enumerate}
\item every morphism \( f \colon X \rightarrow \colim D \) factorizes as $f=c_i\circ g$ for some $i\in I$ and $g\colon D_i\to C$;
\item the factorization is essentially unique: given another factorization $f=c_i\cdot h$, one has \(D_{i,j}\circ g = D_{i,j}\circ h \) for some $j\geq i$ in $I$.
\end{enumerate}  
\end{rem}

Moreover, we need
\begin{lemma}\label{lem:rnnom-of}
Let $X$ be a nominal renaming set generated by a single element $x\in X$, that is, $X=\{\, \rho\cdot x : \rho\in \Fin(\At)\,\}$. Then $X$ is orbit-finite. 
\end{lemma}

\begin{proof}
Take a $\Perm(\At)$-equivariant map $e\colon \At^{\#m}\to X$, where $m\in\Nat$, whose image contains $x$. Since the forgetful functor $U\colon \RnNom\to \Nom$ has a left adjojnt $F\colon \Nom\to\RnNom$ sending $\At^{\#m}$ to $\At^m$~\cite[Thm.~3.7]{mr20}, the map $e$ uniquely extends to a $\Fin(\At)$-equivariant map $\widehat{e}\colon \At^m\to X$. Choose $w\in \At^{m}$ such that $x=\widehat{e}(y)$. The map $\widehat{e}$ is surjective because  $X$ is generated by $x$. Since $\At^{m}$ is orbit-finite, it follows that $X$ is orbit-finite (using that $e$ is $\Perm(\At)$-equivariant, hence it sends orbits to orbits).
\end{proof}

Now we prove \autoref{prop:orbit-finite-vs-fp}. By \cite[Prop.~2.3.7]{Petrisan2012} orbit-finite nominal sets are precisely the finitely presentable objects of $\Nom$. Hence we only need to prove the corresponding statement for nominal renaming sets, which can be reduced to the one for nominal sets.

Thus suppose that $X\in \RnNom$ is finitely presentable. Express $X$ as the directed union of its orbit-finite $\Fin(\At)$-equivariant subsets. To see that this is indeed a directed colimit, note that the cocone of inclusions is jointly surjective: every $x\in X$ is contained in the $\Fin(\At)$-equivariant subset $\{\sigma\cdot x : \sigma\in \Fin(\At)\}$, which is orbit-finite by \autoref{lem:rnnom-of}. Thus the identity map $\id_X\colon X\to X$ factorizes through some inclusion $X'\hookto X$ of an orbit-finite $\Fin(\At)$-equivariant subset, which implies that $X\cong X'$ and thus $X$ is orbit-finite.

Conversely, suppose that $X\in \RnNom$ is orbit-finite. Let $c_i\colon C_i\to C$ ($i\in I$) be a directed colimit in $\RnNom$ (with connecting morphisms $c_{i,j}\colon C_i\to C_j$ for $i\leq j$), and let $f\colon X\to C_i$ be a $\Fin(\At)$-equivariant map. Since $X$ is finitely presentable as a nominal set, $f$ factorizes in $\Nom$ as $f=c_i\cdot g$ for some $i\in I$. The map $g$ is $\Perm(\At)$-equivariant, but may not be $\Fin(\At)$-equivariant. To fix this, choose
elements $x_1,\ldots,x_n\in X$ representing the orbits of $X$, and names $a_1,\ldots,a_d\in \At$ such that $\supp x_r\seq \{a_1,\ldots,a_d\}$ for $r=1,\ldots,n$. Moreover, let $\rho_1,\ldots,\rho_k\in \Fin(\At)$ be all renamings that restrict to a map from $\{a_1,\ldots,a_d\}$ to $\{a_1,\ldots,a_d\}$ and fix all names in $\At\setminus \{a_1,\ldots,a_d\}$. For each $x_r$, $\rho_s$ we have
\[c_i(g(\rho_s\cdot x_r))=f(\rho_s\cdot x_r)=\rho_s\cdot f(x_r) = \rho_s\cdot c_i(g(x_r))=c_i(\rho_s\cdot g(x_r))\]
using that $c_i$ and $f$ are $\Fin(\At)$-equivariant.
Thus $c_{i,j}(g(\rho_s\cdot x_r)) = c_{i,j}(\rho_s\cdot g(x_r))$ for some $j\geq i$, and so $c_{i,j}(g(\rho_s\cdot x_r))=\rho_s\cdot c_{i,j}(g(x_r))$. Since $I$ is directed, we may choose $j$ independently of $r$ and $s$. Thus, after replacing $g$ with $c_{i,j}\circ g$ and $i$ with $j$, we may assume that $g(\rho_s\cdot x_r)=\rho_s\cdot g(x_r)$ for all $r,s$.

We now show that this implies $g(\rho\cdot x)=\rho\cdot g(x)$
for all $x\in X$ and $\rho\in \Fin(\At)$, hence $g$ is $\Fin(\At)$-equivariant. First, since the elements $x_1,\ldots,x_n$ represent the orbits of $X$, we have $x=\pi\cdot x_r$ for some $r$ and $\pi\in\Perm(\At)$. Choose $\tau\in \Perm(\At)$ that restricts to an injective map from $\rho\circ \pi[\{a_1,\ldots,a_d\}]$ to $\{a_1,\ldots,a_d\}$.
Then $\tau\cdot \rho\cdot \pi$ restricts to a map from $\{a_1,\ldots,a_d\}$ to $\{a_1,\ldots,a_d\}$, hence is equal on $\{a_1,\ldots,a_d\}$ to some $\rho_s$. It follows that 
\begin{align*}
 g(\rho\cdot x) &= g(\tau^{-1}\cdot \tau \cdot\rho\cdot \pi \cdot x_r)\\
&= \tau^{-1}\cdot g(\tau\cdot \rho\cdot \pi\cdot x_r)\\
&= \tau^{-1}\cdot g(\rho_s\cdot x_r) \\
&= \tau^{-1}\cdot \rho_s\cdot g(x_r)\\
&= \tau^{-1}\cdot \tau\cdot \rho\cdot \pi\cdot g(x_r)\\
&= \rho\cdot g(\pi\cdot x_r) \\
&= \rho\cdot g(x). 
\end{align*}
The third and fifth step follows from $\supp x_r,\, \supp g(x_r)\seq \{a_1,\ldots, a_d\}$, the fourth step by the choice of $g$, and the second and sixth step use $\Perm(\At)$-equivariance of $g$.

This concludes the proof that $f$ factorizes through $c_i$ in $\RnNom$. That the factorization is essentially unique is immediate from the corresponding property in $\Nom$.

\subsection*{Proof of \autoref{prop:super-finitary}}
An object $X$ of a category $\C$ is \emph{finitely
generated} if the hom-functor $\C(X,-)$ preserves {directed unions}, i.e.~colimits of directed diagrams $D\colon I \to \C$ for which each connecting morphism $D_{i,j}$ ($i\leq j$), is monic. (In locally finitely presentable categories~\cite{adamek_rosicky_1994}, including all presheaf categories, the colimit injections of a directed union are also monic.) Clearly every finitely presentable object is finitely generated. We will first prove that super-finitary presheaves in $\Set^\II$ and $\Set^\FF$ coincide with finitely generated presheaves, and subsequently prove that the latter coincide with finitely presentable presheaves.

\begin{proposition}\label{prop:super-finitary-refined}
For a presheaf $P\in \Set^\C$, $\C\in \{\II,\FF\}$, the following are equivalent:
  \begin{alphaenumerate}
    \item\label{sfequiv:A} $P$ is a finitely generated object of $\Set^\C$;
		\item $P$ is super-finitary;\label{sfequiv:B}
		\item $P$ is a quotient of a finite coproduct of representables; that is, there exists a componentwise surjective natural transformation  $\phi\colon \coprod_{i\in I} \C(S_i,-)\epito P$ with $I$ finite and $S_i\seq_\f \At$. \label{sfequiv:C}
  \end{alphaenumerate}
\end{proposition}
\begin{proof}
 (\ref{sfequiv:A})$\implies$(\ref{sfequiv:B}) Let $P\colon \C \to \Set$ be finitely generated. Since every presheaf in $\Set^\C$ is the directed union of its componentwise finite sub-presheaves, $P$ itself is componentwise finite. For every $S\seq_\f \At$ we consider the sub-presheaf $P_S\seq P$ defined by
\[ P_S T = {\bigcup_{S' \seq S}}\,{\bigcup_{\rho\in \C(S',T)}} P\rho[PS'] \qquad\text{for}\qquad T\seq_\f \At.\]
Note that $P_S$ is super-finitary with generating set $S$. Since $P_R,P_S\seq P_{R\cup S}$ for all $R,S\seq_\f \At$, the map $D\colon (\Pow_\f \At,\seq) \to \Set^\C$, $S\mapsto P_S$, yields a directed diagram of monomorphisms with colimit cocone $P_S\xto{\seq} P$ ($S\seq_\f \At$). By hypothesis the presheaf $P$ is finitely generated, so the identity map $\id\colon P\to P$ factorizes through some $P_S\xto{\seq} P$. This implies that the inclusion is surjective, whence $P=P_S$ is super-finitary.

\medskip\noindent (\ref{sfequiv:B})$\implies$(\ref{sfequiv:C}) For every presheaf $P\in \Set^\C$ and $S\seq_\f \At$ we have the natural transformation $\phi$ whose component at $T\seq_\f \At$ is given by
	\[ \varphi_T\colon \coprod_{S'\seq S} PS'\times \C(S',T)\to PT,\qquad (x,\rho)\mapsto P\rho(x). \]
   If $P$ is super-finitary and $S$ a generating set, then $\phi_T$ is surjective for each $T$.

\medskip\noindent (\ref{sfequiv:C})$\implies$(\ref{sfequiv:A}) This follows from two general facts: First, in every presheaf category the representable functors are finitely presentable~\cite[Example~1.2(7)]{adamek_rosicky_1994}, thus finitely generated. Second, in every locally finitely presentable category (hence in every presheaf category), finitely generated objects are closed under finite coproducts and strong quotients~\cite[Prop.~1.69]{adamek_rosicky_1994}.
\end{proof}

\begin{lemma}\label{lem:fg-closure-props}
  Finitely generated objects in $\Set^\C$, $\C\in \{\II,\FF\}$, are closed under finite products and sub-presheaves (hence under finite limits).
\end{lemma}
\begin{proof}
\emph{Closure under finite products.} We first prove that a product $\C(S_1,-)\times \C(S_2,-)$ of representables is finitely generated, i.e.~super-finitary. Assuming w.l.o.g.\ that $S_1,S_2\seq_\f \At$ are disjoint, we prove that $S_1\cup S_2$ is a generating set. Thus let $T\seq_\f \At$ and $(\sigma_1,\sigma_2)\in \C(S_1,T)\times \C(S_2,T)$. Then there exists $S\seq S_1\cup S_2$ and $\rho\in \C(C,T)$ that corestricts to a bijection $\pi\colon S\to \sigma_1[S_1]\cup \sigma_2[S_2]$. Let $\sigma_k'$ denote the corestriction of $\sigma_k$ to $\sigma_1[S_1]\cup \sigma_2[S_2]$. Then $\sigma_k=\rho\circ \pi^{-1}\circ \sigma_k'$, i.e. $(\sigma_1,\sigma_2)=\left(\C(S_1,\rho)\times \C(S_2,\rho)\right)(\pi^{-1}\circ \sigma_1',\pi^{-1}\circ \sigma_2')$. This proves $\C(S_1,-)\times \C(S_2,-)$ to be super-finitary.

Now we turn to the general case. Suppose that $P_1,P_2\colon \C \to \Set$ are finitely generated presheaves. Then, by \autoref{prop:super-finitary-refined}, there exist componentwise surjective natural transformations
	\[ \varepsilon_k\colon \coprod_{i\in I_k} \C(S_i^{(k)},-) \twoheadrightarrow P_k\]
for $k=1,2$, where $I_k$ is finite and $S_i^{(k)}\seq_\f\At$. They induce the componentwise surjective natural transformation
	\begin{center}
		\begin{tikzcd}
			\coprod_{i\in I_1, j\in I_2} \C(S_i^{(1)},-)\times \C(S_j^{(2)},-) \ar[phantom]{r}[midway]{\cong} &
			(\coprod_{i\in I_1}\C(S_i^{(1)},-)) \times (\coprod_{j\in I_2}\C(S_j^{(2)},-)) \ar[two heads]{d}{\varepsilon_1 \times \varepsilon_2} \\
			& P_1 \times P_2,
		\end{tikzcd}
	\end{center}
Since $\C(S_i^{(1)},-)\times \C(S_j^{(2)},-)$ is finitely generated as shown above, and finitely generated objects in any locally finitely presentable category are closed under finite coproducts and strong quotients, we conclude that $P_1\times P_2$ is finitely generated.

\medskip\noindent\emph{Closure under sub-presheaves.} We first prove that every sub-presheaf $Q\seq \C(C,-)$ of a representable presheaf is finitely generated, i.e.~super-finitary. We may assume that $C\neq \emptyset$ and that $Q$ is not the constant functor on $\emptyset$, for otherwise the claim is obvious. We consider the cases $\C=\II$ and $\C=\FF$ separately:
\begin{itemize}
\item $\C=\II$: Choose $S\seq_\f \At$ of least cardinality such that $QS\neq \emptyset$; since $QS\seq \II(C,S)$ one has $|C|\leq |S|$. Note that $QS=\II(C,S)$: for any two maps $\sigma,\tau\in \II(C,S)$ one has $\tau=\pi\circ \sigma$ for some bijection $\pi\colon S\to S$, hence $\sigma\in QS$ implies $\tau\in QS$. Given $T\seq_\f \At$ and $\sigma\in QT$, one has $|C|\leq |S|\leq |T|$ and hence $\sigma$ factorizes as $C\xto{\tau} S \xto{\rho} T$ for some $\rho,\tau\in \II$. Then $\tau\in QS=\II(C,S)$ and $\sigma=\C(C,\rho)(\tau)=Q\rho(\tau)$. Hence $Q$ is super-finitary, as claimed.
\item $\C=\FF$: We prove that the set $C$ generates $Q$. Given $T\seq_\f \At$ and $\sigma\in QT$, choose a map $\tau\in \FF(T,C)$ that sends every element of $\sigma[C]\seq T$ to a preimage under $\sigma$, and is arbitrary otherwise. (Here we use that $C\neq\emptyset$.) Then $\tau\circ \sigma=Q\tau(\sigma)\in QC$ and $\sigma=\sigma\circ \tau\circ \sigma=\FF(C,\sigma)(\tau\circ \sigma)=Q\sigma(\tau\circ \sigma)$, proving that $Q$ is super-finitary.  
\end{itemize}

Now we turn to the general case. Suppose that $Q\seq P$ is a sub-presheaf of a finitely generated presheaf  $P\in \Set^\C$. By \autoref{prop:super-finitary-refined} we have a componentwise surjective natural transformation $\varepsilon\colon \coprod_{i\in I} \C(S_i,-)\epito P$ with $I$ finite. Form the following pullback:
	\begin{center}
		\begin{tikzcd}
			H \ar[hook]{r}[midway, above]{\seq} \ar[two heads]{d}[midway, left]{\overline{\varepsilon}}
			\pullbackangle{-45} & \coprod_{i\in I} \C(S_i,-) \ar[two heads]{d}[midway, right]{\varepsilon} \\
			Q \ar[hook]{r}[midway, below]{\seq} & P
		\end{tikzcd} 
	\end{center}
By extensivity of the presheaf topos $\Set^\C$, the presheaf $H$ is of the form $H=\coprod_{i\in I} H_i$ for sub-presheaves $H_i\seq \C(S_i,-)$. Each $H_i$ is finitely generated as shown above, hence so is $H$ because finitely generated objects are closed under finite coproducts.
\end{proof}
Finally, we have the following simple criterion for coincidence of finitely presentable and finitely generated objects, see \cite[Lemma~3.32]{amsw19_1}:

\begin{lemma}\label{lem:fp-vs-fg}
Let $\C$ be a locally finitely presentable category where strong and regular epimorphisms coincide and
  finitely generated objects are closed under kernel pairs. Then the finitely presentable and finitely generated objects of $\C$ coincide.
\end{lemma}

With these preparations, \autoref{prop:super-finitary} now easily follows. By \autoref{prop:super-finitary-refined} we only need to show coincidence of finitely presentable and finitely generated objects in $\Set^\C$, where $\C\in \{\II,\FF\}$. To this end apply \autoref{lem:fp-vs-fg}: every presheaf category is locally finitely presentable, regular and strong epis coincide in every topos (and are just the epimorphisms), and closure of finitely generated objects under kernel pairs follows from \autoref{lem:fg-closure-props}.

\subsection*{Proof of \autoref{prop:fp-presheaves}}
For the proof the following result (see e.g.~\cite[Lem.~2.4]{amsw19_2}) will be helpful:

\begin{lemma}\label{lem:left-adj-pres-fp}
For every adjunction $F\dashv U\colon \C\to \D$ between categories with directed colimits, if $U$ preserves directed colimits then $F$ preserves finitely presentable objects.
\end{lemma}

To prove that the left adjoints in \eqref{eq:adjunctions} preserve finitely presentable objects, it suffices to show that their right adjoints preserve directed colimits (\autoref{lem:left-adj-pres-fp}). 
\begin{itemize}
\item The forgetful functors $E^\star$ and $U$ preserve all colimits because colimits in the four categories are formed at the level of underlying sets.
\item To show that $I_\star$ preserves directed colimits, let $c_k\colon C_k\to C$ ($k\in K$) be a directed colimit cocone in $\Nom$. Then the morphisms $(I_\star c_k)_S\colon (I_\star C_k)S \to (I_\star C)S$ are jointly surjective for every $S\seq_\f \At$ by \cite[Lem.~5.14]{Pitts2013}, and any two elements of $(I_\star C_k)S$ merged by $(I_\star c_k)_S$ are merged by $(I_\star c_{k,l})_S$ for some $k\leq l$ because this holds in the directed colimit cocone $(c_k)$ in $\Nom$. Therefore the morphisms $I_\star c_k$ form a colimit cocone in $\Set^\II$.
\item The proof that $J_\star$ preserves directed colimits is analogous.
\end{itemize} 
It remains to show that the four right adjoints preserve finitely presentable objects.
\begin{itemize}
\item $U$ clearly preserves finitely presentable objects. i.e.~orbit-finite sets (\autoref{prop:orbit-finite-vs-fp}).
\item To show that $E^\star$ preserves finitely presentable objects, by \autoref{prop:fp-presheaves} we need to show that for every super-finitary presheaf $P\colon \FF\to \Set$ the presheaf $E^\star P = P\circ E\colon \II\to \Set$ is super-finitary. Clearly $P\circ E(T)=P(T)$ is finite for every $T\seq_\f\At$. Moreover,
we claim that any set $S\seq_\f \At$ generating $P$ also generates $P\circ E$. Indeed, suppose that $T\seq_\f \At$ and $x\in (P\circ E)T=PT$. Then $x=P\rho(x')$ for some $S'\seq S$ and $\rho\colon S'\to T$. The map $\rho$ factorizes as $\rho=\rho_1\circ \rho_0$ where $\rho_1\colon S''\to T$ is injective and $S''\seq S'$. Thus, putting $x''=P\rho_0(x')$, we have
\[ x=P\rho(x') = P\rho_1(P\rho_0(x')) = P\rho_1(x'') = (P\circ E)\rho_1(x'')\in (P\circ E)\rho_1[(P\circ E)S''], \]
proving that $P\circ E$ is super-finitary. 
\item To show that the functor $I_\star$ preserves finitely presentable objects, by \autoref{prop:orbit-finite-vs-fp} and \autoref{prop:fp-presheaves} we need to show that its sends orbit-finite nominal sets to super-finitary presheaves. Thus let $X\in \Nom$ be orbit-finite. Then $(I_\star X)T$ is finite for every $T\seq_\f\At$ because an orbit-finite set contains only finitely elements of any given finite support. We claim that $I_\star X$ is generated by any set $S\seq_\f \At$  of $n=\max_{x\in X} |{\supp x}|$ names. To see this, let $x\in (I_\star X)T$, that is, $x\in X$ with support $T$. Then $x$ is supported by some subset $T'\seq T$ with at most $n$ elements, that is, $x\in (I_\star X)T'$. Choose a bijection $\pi\colon S'\to T'$ where $S'\seq S$, which extends an injection $\rho\colon S'\to T$. It then follows that $x=(I_\star X)\rho(\pi^{-1}\cdot x) = x$, as required. 
\item An analogous argument shows that $J_\star$ preserves finitely presentable objects.
\end{itemize}

\subsection*{Proof of \autoref{prop:epi-pres-lang}}
\begin{rem}
Recall that by definition strong epimorphisms satisfy the \emph{diagonal fill-in} property: For every commutative square as shown below where $e$ is a strong epimorphism and $m$ is a monomorphism, there exists a unique $d\colon B\to C$ making both triangles commute.
\[
\begin{tikzcd}
A \ar[two heads]{r}{e} \ar{d}[swap]{f} & B \ar{d}{g} \ar[dashed]{dl}[swap]{d} \\
C \ar[rightarrowtail]{r}{m} & D
\end{tikzcd}
\] 
\end{rem}

We will make use of the \emph{pullback lemma}, see e.g.~\cite[Prop.~2.5.9]{borceux-94-1}:
\begin{lemma}\label{lem:pullback-lemma}
Given a commutative diagram as shown below in a category with pullbacks,
\begin{enumerate}
\item\label{lem:pullback-lemma-1} if (I) and (II) are pullbacks, then the outer rectangle is a pullback;
\item\label{lem:pullback-lemma-2} if (II) and the outer rectangle are pullbacks, then (I) is a pullback.
\end{enumerate}
\[
\begin{tikzcd}
\bullet \ar{r} \ar[phantom]{dr}[description]{\text{(I)}} \ar{d} & \bullet \ar[phantom]{dr}[description]{\text{(II)}} \ar{r} \ar{d} & \bullet \ar{d} \\
\bullet \ar{r} & \bullet \ar{r} & \bullet
\end{tikzcd}
\]
\end{lemma}
In the following let $A=(Q,\Sigma,\delta,I,F)$ and $A'=(Q',\Sigma,\delta',I',F')$ be nondeterministic $\C$-automata over the same alphabet $\Sigma$ and suppose that $h\colon A'\to A$ is a morphism such that $h_\al=\id$. We prove the two parts of the proposition.
\begin{enumerate}
\item For every $n\geq 0$ we show that $L^{(n)}(A')\leq L^{(n)}(A)$ as subobjects of $\Sigma^n$. We only consider the case $n>0$, as the argument for $n=0$ is very similar. The universal property of the pullback $\Run_A$ yields a unique morphism $e$ such that the upper part and the left-hand part of the diagram below commute; note that the outside and the other parts commute by definition.
\begin{equation}\label{eq:run-to-run}
\begin{tikzcd}[row sep = 3em]
\Run_{A'} \mar{r}{\ol{m}_{\delta'}^{(n)}} \ar[shiftarr = {xshift=-30}]{ddd}[swap]{\ol{d}_{n,A'}} \ar[dashed]{d}[swap]{e} & I'\times (\Sigma\times Q')^{n-1}\times \Sigma\times F' \ar{d}[description]{h_\init\times (\id\times h_\states)^{n-1}\times \id \times h_\final} \ar[shiftarr = {xshift=75}]{ddd}{d_{n,A'}} \\
\Run_A \mar{r}{\ol{m}_\delta^{(n)}} \ar{d}[swap]{\ol{d}_{n,A}} & I\times (\Sigma\times Q)^{n-1}\times \Sigma\times F \ar{d}{d_{n,A}} \\
\delta^n \ar{r}{m_\delta^n} & (Q\times\Sigma\times Q)^n \\
(\delta')^n \ar{u}{h_\delta^n} \ar{r}{m_{\delta'}^n} & (Q'\times \Sigma\times Q')^n \ar{u}[swap]{(h_\states\times\id\times h_\states)^n}
\end{tikzcd}
\end{equation}
Diagonal fill-in yields a unique morphism $i\colon L^{(n)}(A')\to L^{(n)}(A)$ making the upper part and the left-hand part of the diagram below commute; the outside and the other parts commute by definition. Hence the morphism $i$ witnesses that $L^{n}(A')\leq L^{(n)}(A)$.
\begin{equation}\label{eq:lang-inclusion}
\begin{tikzcd}
L^{(n)}(A') \ar[dashed]{dr}{i} \mar{ddd}[swap]{m_{L(A)}^{(n)}} & & &\Run_{A'} \ear{lll}[swap]{e_{n,A'}} \ar{dl}[swap]{e} \ar{ddd}{\ol{m}_{\delta'}^{(n)}} \\
& L^{(n)}(A) \mar{d}[swap]{m_{L(A)}^{(n)}} & \Run_A \ear{l}[swap]{e_{n,A}} \ar{d}{\ol{m}_{\delta}^{(n)}} & \\
& \Sigma^n \ar{dl}[swap]{\id} &  I\times (\Sigma\times Q)^{n-1}\times \Sigma \times F \ar{l}[swap]{p_{n,A}} & \\
\Sigma^n & & & I'\times (\Sigma\times Q')^{n-1}\times \Sigma \times F' \ar{lll}[swap]{p_{n,A'}} \ar{ul}[description]{h_\init\times (\id\times h_\states)^{n-1}\times \id\times h_\final}
\end{tikzcd}
\end{equation}
\item Now suppose that $h_\states$ is a strong epimorphism in $\C$ and that the three squares \eqref{eq:naut-morphism} are pullbacks. Our task is to show that $L^{(n)}(A')=L^{(n)}(A)$ for all $n\geq 0$; as above we only consider the case $n>0$. We first observe that the upper rectangle of \eqref{eq:run-to-run} is a pullback. To see this, note that the outside is a pullback by definition, and that the lower rectangle is a pullback by our assumption on $h$ and the fact that in every category pullbacks commute with products. 
Thus \autoref{lem:pullback-lemma}.\ref{lem:pullback-lemma-1} shows that the composite of the upper and the central rectangle forms a pullback, as it is is equal to the composite of the outside and the lower rectangle. Moreover the central rectangle is a pullback by definition, and so by \autoref{lem:pullback-lemma}.\ref{lem:pullback-lemma-2}, the upper rectangle is a pullback as well.

Since strong epimorphisms in $\C$ are stable under pullbacks and products (\autoref{ass:cats}), the morphism $h_\init\times (\id\times h_\states)^{n-1}\times \id\times h_\final$ appearing in the upper rectangle is a strong epimorphism. Using stability under pullbacks again, we see that $e$ is a strong epimorphism.
Therefore, by the uniqueness of image factorizations, the diagonal fill-in $i$ in \eqref{eq:lang-inclusion} is an isomorphism, proving that $L^{(n)}(A')=L^{(n)}(A)$ as subobjects of $\Sigma^n$.\qedhere
\end{enumerate}

\subsection*{Proof of \autoref{prop:liftadjaut}}
  \begin{enumerate}
    \item Given a functor $G\colon \C \to \D$ we define the lifted functor as follows:
      \[ \functor{\ol{G}}{\NAut(\C)}{\NAut(\D)}{\maketuple{Q,\Sigma,\delta,I,F}}{\maketuple{GQ,G\Sigma,\ol{G\delta},\ol{GI},\ol{GF}}}{\maketuple{h_\states,h_\al}}{\maketuple{Gh_\states,Gh_\al}}  \]
      Herein, the objects $\ol{G\delta}$, $\ol{GI}$, and $\ol{GF}$ are given by the image factorizations shown below, with $\canmap$ denoting the canonical morphism induced by the product projections:
      \[
        \begin{tikzcd}[column sep=4.5em]
          G\delta \ar[two heads]{r}{e_{\ol{G\delta}}} \ar{d}[left]{Gm_{\delta}} & \ol{G\delta} \ar[tail]{d}{m_{\ol{G\delta}}} \\
          G(Q \times \Sigma \times Q) \ar{r}[below]{\canmap} & GQ \times G\Sigma \times GQ
        \end{tikzcd}
          \quad
        \begin{tikzcd}[column sep=3em]
          GI \ar{d}[left]{Gm_I} \ar[two heads]{r}{e_{\ol{GI}}} & \ol{GI} \ar[tail]{ld}{m_{\ol{GI}}} \\
          GQ
        \end{tikzcd}
          \quad
        \begin{tikzcd}[column sep=3em]
          GF \ar{d}[left]{Gm_F} \ar[two heads]{r}{e_{\ol{GF}}} & \ol{GF} \ar[tail]{ld}{m_{\ol{GF}}} \\
          GQ
        \end{tikzcd}
      \]
      We only need to prove that $\maketuple{Gh_\states,Gh_\al}$ is an $\NAut(\D)$-morphism for every $\NAut(\C)$-morphism $h=\maketuple{h_\states,h_\al}\colon A' \to A$ between
      automata $A' = \maketuple{Q',\Sigma',\delta',I',F'}$ and $A = \maketuple{Q,\Sigma,\delta,I,F}$. Indeed, via diagonal fill-in we obtain the dashed morphisms making the diagrams below commute; note that in all three diagrams the outside commutes because $h$ is an $\NAut(\C)$-morphism, and the parts not involving the dashed morphisms commute either by definition or by naturality of $\canmap$. The central part of the first diagram and the lower parts of the other two diagrams show that $(Gh_\states, Gh_\al)$ is an $\NAut(\C)$-morphism.
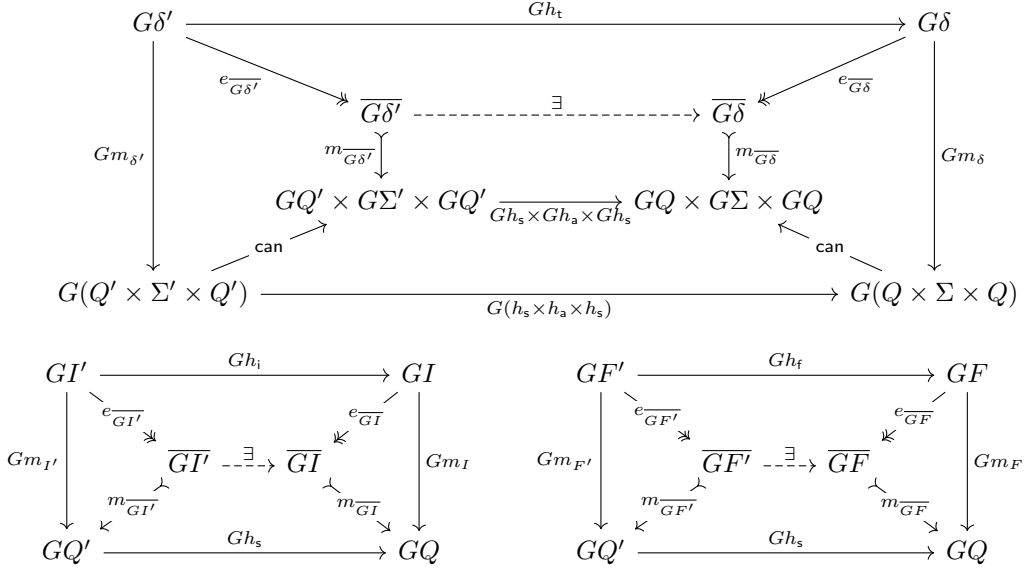
      \[
        \begin{tikzcd}[column sep=0.125em,/tikz/column 2/.style={column sep=4.5em}]
          G\delta' \arrow[ddd, "Gm_{\delta'}"'] \arrow[rd, two heads, "e_{\ol{G\delta'}}"'] \arrow[rrrr, "Gh_{\trans}"] & & & & G\delta \arrow[ddd, "Gm_{\delta}"] \arrow[ld, two heads, "e_{\ol{G\delta}}"] \\
          & \ol{G\delta'} \arrow[d, "m_{\ol{G\delta'}}"', tail] \arrow[rr, "\exists", dashed] &  & \ol{G\delta} \arrow[d, "m_{\ol{G\delta}}", tail] &                                                          \\
          & GQ' \times G\Sigma' \times GQ' \arrow[rr, "Gh_\states \times Gh_\al \times Gh_\states"']                                       &  & GQ \times G\Sigma \times GQ                          &                                                          \\
          G(Q'\times\Sigma'\times Q') \ar{ru}[description]{\canmap} \arrow[rrrr, "G(h_\states \times h_\al \times h_\states)"'] & & & & G(Q \times \Sigma \times Q) \ar{lu}[description]{\canmap}                  
        \end{tikzcd}
      \]
      \[
        \begin{tikzcd}[column sep=2em]
          GI' \ar{dd}[left]{Gm_{I'}} \ar[two heads]{dr}[description]{e_{\ol{GI'}}} \ar{rrr}{Gh_\init} &&& GI \ar{dd}{Gm_I} \ar[two heads]{dl}[description]{e_{\ol{GI}}} \\
          & \ol{GI'} \ar[dashed]{r}[midway, above]{\exists} \ar[tail]{dl}[description]{m_{\ol{GI'}}} & \ol{GI} \ar[tail]{dr}[description]{m_{\ol{GI}}} & \\
          GQ' \ar{rrr}{Gh_\states} &&& GQ
        \end{tikzcd}
          \qquad
        \begin{tikzcd}[column sep=2em]
          GF' \ar{dd}[left]{Gm_{F'}} \ar[two heads]{dr}[description]{e_{\ol{GF'}}} \ar{rrr}{Gh_\final} &&& GF \ar{dd}{Gm_F} \ar[two heads]{dl}[description]{e_{\ol{GF}}} \\
          & \ol{GF'} \ar[dashed]{r}[midway, above]{\exists} \ar[tail]{dl}[description]{m_{\ol{GF'}}} & \ol{GF} \ar[tail]{dr}[description]{m_{\ol{GF}}} & \\
          GQ' \ar{rrr}{Gh_\states} &&& GQ
        \end{tikzcd}
      \]
    \item Let  $L\dashv R\colon \C\to \D$ be an adjunction with unit $\eta\colon \id_\D \to RL$ and counit $\varepsilon\colon LR \to \id_\C$. We only need to establish the following two statements:
\begin{enumerate}
\item\label{adj-a} for every $\D$-automaton $A=(Q,\Sigma,\delta,I,F)$ the pair $\ol{\eta}_A=(\eta_Q,\eta_\Sigma)\colon A\to \ol{R}\,\barL A$ is an $\NAut(\D)$-morphism;
\item\label{adj-b} for every $\C$-automaton $A=(Q,\Sigma,\delta,I,F)$ the pair $\ol{\epsilon}_A=(\epsilon_Q,\epsilon_\Sigma)\colon \barL\,\ol{R}A\to A$ is an $\NAut(\C)$-morphism.
\end{enumerate}

Then $\barL\dashv \ol{R}$ is an adjunction with unit $\ol{\eta}$ and counit $\ol{\epsilon}$. Indeed, naturality of $\ol{\eta}$ and $\ol{\epsilon}$ and the triangle laws are immediate from the corresponding properties of $\eta$ and $\epsilon$.

\medskip\noindent The proof of the first statement is given by the commutative diagrams below, where we write $L_I$, $L_F$, $L_\delta$ for $\ol{LI}$, $\ol{LF}$, $\ol{L\delta}$ and the dashed morphisms are just given by composition. In all three diagrams the outside commutes by naturality of $\eta$, and the parts not involving the dashed morphisms commute by definition.
            
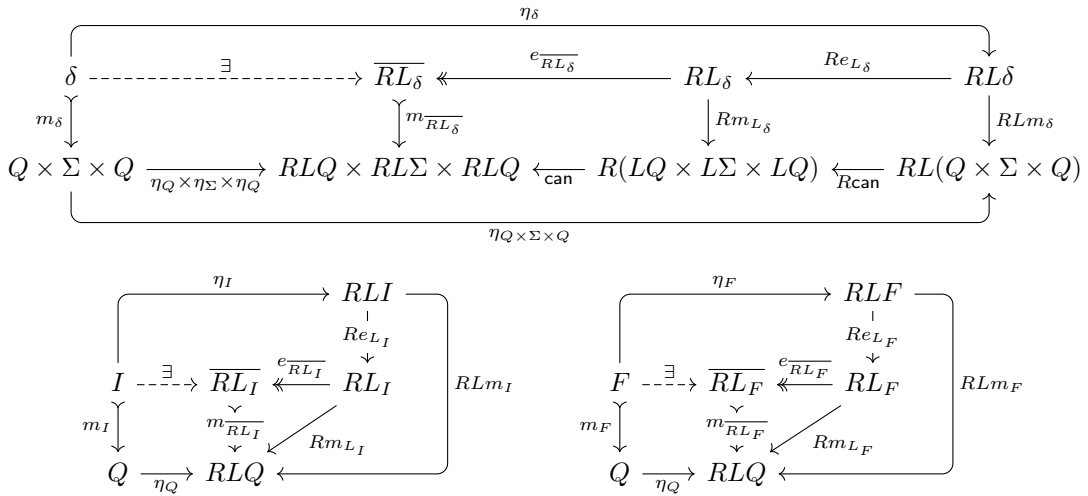
\[
              \begin{tikzcd}[column sep=2em,/tikz/column 1/.style={column sep=4.5em}]
                \delta \ar[tail]{d}[swap]{m_\delta} \ar[dashed]{r}{\exists} \ar[shiftarr={yshift=20}]{rrr}[above]{\eta_\delta} &
                \ol{RL_\delta} \ar[tail]{d}{m_{\ol{RL_\delta}}} & RL_\delta \ar[two heads]{l}[above]{e_{\ol{RL_\delta}}} \ar{d}{Rm_{L_\delta}} & RL\delta \ar{d}{RLm_\delta} \ar{l}[above]{Re_{L_\delta}} \\
                Q \times \Sigma \times Q \ar{r}[below]{\eta_Q \times \eta_\Sigma \times \eta_Q} \ar[shiftarr={yshift=-20}]{rrr}[below]{\eta_{Q \times \Sigma \times Q}} &
                RLQ \times RL\Sigma \times RLQ & R(LQ \times L\Sigma \times LQ) \ar{l}{\canmap} & RL(Q \times \Sigma \times Q) \ar{l}{R\canmap}
              \end{tikzcd}
            \]
          \[
            \begin{tikzcd}[column sep=2em]
              && RLI \ar{d}[description]{Re_{L_I}} \ar[shiftarrr = {xshift=30}]{ddl}[right, pos=0.25]{RLm_I} \\
              I \ar[roundcornerarrrr]{urr}[pos=0.75]{\eta_I} \ar[tail]{d}[left]{m_I} \ar[dashed]{r}{\exists} & \ol{RL_I} \ar[tail]{d}[description]{m_{\ol{RL_I}}} &
              RL_I \ar[two heads]{l}[above]{e_{\ol{RL_I}}} \ar{dl}{Rm_{L_I}} \\
              Q \ar{r}[below]{\eta_Q} & RLQ
            \end{tikzcd}
            \qquad
            \begin{tikzcd}[column sep=2em]
              && RLF \ar{d}[description]{Re_{L_F}} \ar[shiftarrr = {xshift=30}]{ddl}[right, pos=0.25]{RLm_F} \\
              F \ar[roundcornerarrrr]{urr}[pos=0.75]{\eta_F} \ar[tail]{d}[left]{m_F} \ar[dashed]{r}{\exists} & \ol{RL_F} \ar[tail]{d}[description]{m_{\ol{RL_F}}} &
              RL_F \ar[two heads]{l}[above]{e_{\ol{RL_F}}} \ar{dl}{Rm_{L_F}} \\
              Q \ar{r}[below]{\eta_Q} & RLQ
            \end{tikzcd}
          \]
Similarly, the second statement is proven by the three diagrams below, where we write $R_I$, $R_F$, $R_\delta$ for $\ol{RI},\ol{RF},\ol{R\delta}$ and the dashed morphisms are given by diagonal fill-in. Here we use the fact that $L$ preserves strong epimorphisms, being a left adjoint, and that strong epimorphisms are closed under composition.
          \[
            \begin{tikzcd}[column sep=2em,/tikz/column 3/.style={column sep=4.5em}]
              LR\delta \ar[shiftarr={yshift=20}]{rrr}[above]{\varepsilon_\delta} \ar{d}[left]{LRm_\delta} \ar[two heads]{r}[above]{Le_{R_\delta}} & LR_\delta 
              \ar[two heads]{r}[above]{e_{\ol{LR_\delta}}} \ar{d}{Lm_{R_\delta}} & \ol{LR_\delta} \ar[tail]{d}{m_{\ol{LR_\delta}}} \ar[dashed]{r}{\exists} &
              \delta \ar[tail]{d}{m_\delta} \\
              LR(Q \times \Sigma \times Q) \ar[shiftarr={yshift=-20}]{rrr}[below]{\varepsilon_{Q \times \Sigma \times Q}} \ar{r}[below]{L\canmap} & L(RQ \times R\Sigma \times RQ)
              \ar{r}[below]{\canmap} & LRQ \times LR\Sigma \times LRQ \ar{r}[below]{\varepsilon_Q \times \varepsilon_\Sigma \times \varepsilon_Q} &
              Q \times \Sigma \times Q
            \end{tikzcd}
\]
          \[
            \begin{tikzcd}[column sep=2em]
              LRI \ar[roundcornerarrr]{drr}[above,pos=0.25]{\varepsilon_I}  \ar[two heads]{d}[description]{Le_{R_I}} \ar[shiftarrr = {xshift=-30}]{ddr}[left, pos=0.25]{LRm_I} && \\
              LR_I \ar[two heads]{r}[above]{e_{\ol{LR_I}}} \ar{dr}[description]{Lm_{R_I}} & \ol{LR_I} \ar[tail]{d}[description]{m_{\ol{LR_I}}} \ar[dashed]{r}[above]{\exists} &
              I \ar[tail]{d}{m_I}  \\
              & LRQ \ar{r}[below]{\varepsilon_Q} & Q 
            \end{tikzcd}
            \qquad
      \begin{tikzcd}[column sep=2em]
              LRF \ar[roundcornerarrr]{drr}[above,pos=0.25]{\varepsilon_F}  \ar[two heads]{d}[description]{Le_{R_F}} \ar[shiftarrr = {xshift=-30}]{ddr}[left, pos=0.25]{LRm_F} && \\
              LR_F \ar[two heads]{r}[above]{e_{\ol{LR_F}}} \ar{dr}[description]{Lm_{R_F}} & \ol{LR_F} \ar[tail]{d}[description]{m_{\ol{LR_F}}} \ar[dashed]{r}[above]{\exists} &
              F \ar[tail]{d}{m_F}  \\
              & LRQ \ar{r}[below]{\varepsilon_Q} & Q 
            \end{tikzcd}
          \]
This concludes the proof. \qedhere
  \end{enumerate}

\subsection*{Proof of \autoref{prop:pos-closure-presheaf-lang}}
Put $\Lan=\Lan_E$ and consider the morphism
 \[ \varphi\colon \Lan(L) \xra{\Lan(\iota)} \Lan(V_{\II}^*) \cong \coprod_{k}\Lan(V_{\II}^k) \xto{\coprod_k \canmap_k} \coprod_{k}\Lan(V_{\II})^k = \coprod_{k} V_\FF^k = V_{\FF}^* \]
in $\Set^\FF$, where $\iota\colon L\hookto V_\II^\star$ is the inclusion and $\canmap_k$ is the canonical morphism induced by the product projections, and form its image factorization
\[
\begin{tikzcd}[column sep=3em]
\varphi\;=\; \big(\,\Lan(L) \ar[two heads]{r}{\coim \phi} & \barL \ar[hook]{r}{\im \phi} & V_\FF^\star\big).
\end{tikzcd}  
\]
We prove that $\barL$ is a positive closure of $L$. First, the diagram below demonstrates that $L\seq \barL E$, witnessed by the morphism $(\coim \phi) E\circ \eta_L$, where $\eta$ is the unit of the adjunction $\Lan \dashv E^\star\colon \Set^\FF\to \Set^\II$. Indeed, all parts commute either by definition or by naturality.
\[
\begin{tikzcd}[column sep=3em]
L \ar{r}{\eta_L} \ar{dddd}[swap]{\iota} & \Lan(L)E \ar{rrrr}{(\coim\phi)E} \ar{d}{\Lan(\iota)E} \ar[bend left=2em]{ddddrrrr}{\phi E} & &&& \barL E \ar{dddd}{(\im\phi)E} \\
& \Lan(V_\II^\star)E \ar{d}{\cong}  & &&& \\
& (\coprod_k \Lan(V_\II^k))E \ar{dr}{(\coprod_k \canmap_k)E}  &&&& \\
& \coprod_k V_\II^k \ar{u}{\coprod_k \eta_{V_\II^k}} \ar{r}{\coprod_k \eta_{V_\II}^k } & (\coprod_k \Lan(V_\II)^k)E \ar{d}{\cong} &&& \\
V_\II^\star \ar[equals]{rr} \ar{ur}{\cong} \ar[bend left=2em]{uuur}{\eta_{V_\II^\star}}  & &  V_\FF^\star E \ar[equals]{rrr} &&& V_\FF^\star E 
\end{tikzcd}
\]
To show that $\barL$ is minimal with that property, let $K\seq V_\FF^\star$ such that $L\seq KE$; denote the inclusions by $\psi\colon K\hookto V_\FF^\star$ and $\xi\colon L\hookto KE$. By the universal property of $\Lan(L)$, there exists a unique $\ol{\xi}\colon \Lan(L)\to K$ such that $\xi=\ol{\xi} E \circ \eta$. Now consider the first diagram below; its left-hand part commutes by definition, and the outside commutes because it does so when restricted to $E$ and precomposed with the universal map $\eta$, see the second diagram (note that $\iota = \phi E\circ \eta$ by the diagram above). Hence we obtain the dashed morphism via diagonal fill-in, witnessing that $\barL\seq K$.
\[
        \begin{tikzcd}[sep=small]
          {\Lan(L)} &&&& K \ar[tail, roundcornerarr={yshift=-4ex}]{ddddddll}[pos=.25, right]{\psi} \\
          \\
          {\Lan(V_{\II}^*)} \\
          {\coprod_{k} \Lan(V_{\II}^k)} && {\overline{L}}\\
          \\
          {\coprod_{k} \Lan(V_{\II})^k} \\
          {V_{\FF}^*} && {V_{\FF}^*} \\
          \arrow["\exists", dashed, from=4-3, to=1-5]
          \arrow["{\Lan(\iota)}"', from=1-1, to=3-1]
          \arrow["\cong"{marking}, draw=none, from=3-1, to=4-1]
          \arrow["{\coprod_k \canmap_k}"', from=4-1, to=6-1]
          \arrow["\cong"{marking}, draw=none, from=6-1, to=7-1]
          \arrow[equals, from=7-1, to=7-3]
          \arrow["{\im\varphi}", tail, from=4-3, to=7-3]
          \arrow["{\coim\varphi}", two heads, from=1-1, to=4-3]
          \arrow["\ol{\xi}", from=1-1, to=1-5]
        \end{tikzcd}
\qquad\qquad     
\begin{tikzcd}[sep=small]
        && {\Lan(L)E} \ar[shiftarr={xshift=4em}]{dddd}{\varphi E} \\
        \\
        KE && L && \\
        \\
        && {V_{\FF}^*E}
        \arrow["{\ol{\xi}E}"', from=1-3, to=3-1]
        \arrow["{\psi E}"', tail, from=3-1, to=5-3]
        \arrow["\ini", tail, from=3-3, to=5-3]
        \arrow["{\xi}"', from=3-3, to=3-1]
        \arrow["{\eta}"', from=3-3, to=1-3]
\end{tikzcd}
\]

\subsection*{Proof of \autoref{prop:presheaf-aut-closure}}

\begin{rem}\label{rem:left-kan-extension}
The left Kan extension $\Lan_E P\colon \FF\to \Set$ of a presheaf $P\colon \II\to \Set$ along $E\colon \II\hookto \FF$ is computed as follows, see e.g.~\cite[Thm.~X.3.1]{mac-71}:
\begin{itemize}
\item For $S\seq_\f \At$ the set $\Lan_E P(S)$ is the colimit of the diagram
\[ D_S\colon E{\downarrow}S\to \Set, \quad (\rho\colon ET\to S)\mapsto PT.\]
Here $E{\downarrow}S$ is the comma category whose objects are maps $\rho\colon ET\to S$ in $\FF$ where $T\seq_\f \At$ and whose morphisms from $\rho$ to $\rho'\colon ET'\to S$ are maps $\tau\colon T\to T'$ in $\II$ such that $\rho=\rho'\circ E\tau$. For an even more explicit description of $\Lan_E P(S)$ consider the set of all pairs $(x,\rho)$ where $\rho\colon ET\to S$ for some $T\seq_\f \At$ and $x\in PT$, and for any two such pairs put $(x,\rho)\sim (x',\rho')$ iff there exists a ziz-zag
\[ T=T_0\xto{\tau_1} T_1 \xleftarrow{\tau_2} T_2 \to \cdots \xleftarrow{\tau_{2n}} T_{2n}=T'  \]
in $\II$ and elements $x_i\in PT_i$ ($i=0,\ldots,n$) such that $x_0=x$, $x_{2n}=x'$, $P\tau_i(x_{i-1})=x_{i}$ for $i$ odd, and $P\tau_i(x_i)=x_{i-1}$ for $i>0$ even. 
Then $\sim$ is an equivalence relation, and $\Lan_E P(S)$ is the set of equivalence classes $[x,\rho]$ of $\sim$. The colimit injection $c_\rho\colon PT\to \Lan_E P(S)$  associated to $\rho\in E{\downarrow}S$ maps $x\in PT$ to $[x,\rho]$.
\item For $\sigma\colon S\to S'$ in $\FF$, the map $\Lan_E P(\sigma)\colon \Lan_E P(S)\to \Lan_E P(S')$ sends $[x,\rho]$ to $[x,\sigma\circ \rho]$.   
\item For a morphism $f\colon P\to P'$ in $\Set^\II$, the component of $\Lan_E f\colon \Lan_E P\to \Lan_E P'$ at $S\seq_\f \At$ is given by $[x,\rho]\mapsto [f_T(x),\rho]$ where $\rho\colon ET\to S$ and $x\in PT$.
\end{itemize}
\end{rem}

Let $A=(Q,V_\II,\delta,I,F)$ be a nondeterministic $\Set^\II$-automaton with a strong presheaf $Q$ of states. We put $\Lan=\Lan_E$ and $\barA=\ol{\Lan}\, A$, that is,
\[\barA = (\Lan\, Q, V_\FF, \barDelta,\barI,\barF),\] where $\barDelta=\ol{\Lan\,\delta}$, $\barI=\ol{\Lan\, I}$ and $\barF=\ol{\Lan\, F}$ 
 are obtained via the image factorizations of \autoref{prop:liftadjaut}. Our task is to show that $\barA$ accepts the language $\ol{L(A)}$, that is, $L^{(n)}(\barA)=\ol{L(A)}^{(n)}$ for all $n\geq 0$. We shall only treat the case $n>0$; the argument for $n=0$ is similar.

\subparagraph*{Step 1.} The universal property of the pullback $\Run_{\overbar{A}}^{(n)}$ yields a unique morphism $\epsilon$ making the diagram below commute:
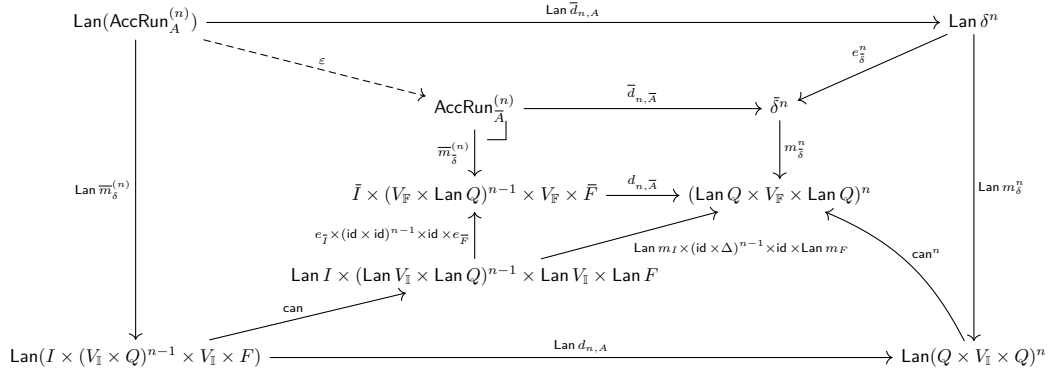
\[
\begin{tikzcd}[scale cd=.7, column sep=.5em]
\Lan(\Run_A^{(n)}) \ar{rrr}{\Lan\, \ol{d}_{n,A}} \ar{dddd}[swap]{\Lan\, \ol{m}_\delta^{(n)}} \ar[dashed]{dr}{\epsilon} & & & \Lan\, \delta^{n} \ar{dddd}{\Lan\, m_\delta^n} \ar{dl}[swap]{e_{\overbar{\delta}}^n} \\
& \Run_{\overbar{A}}^{(n)} \ar{r}{\ol{d}_{n,\overbar{A}}} \ar{d}[swap]{\ol{m}_{\overbar{\delta}}^{(n)}} \pullbackangle{-45} & \barDelta^{n} \ar{d}{m_{\overbar{\delta}}^{n}}  & \\
& \barI\times (V_\FF\times \Lan\, Q)^{n-1}\times V_\FF\times \barF \ar{r}{d_{n,\overbar{A}}} & (\Lan\, Q \times V_\FF\times \Lan\, Q)^n & \\
& \Lan\, I\times (\Lan\, V_\II\times \Lan\, Q)^{n-1}\times \Lan\, V_\II \times \Lan\, F \ar{u}{e_{\overbar{I}} \times (\id\times \id)^{n-1} \times \id \times e_{\overbar{F}} } \ar{ur}[swap]{\Lan\, m_I\times (\id\times \Delta)^{n-1}\times \id\times \Lan\, m_F} & & \\
\Lan (I\times (V_\II \times Q)^{n-1}\times V_\II \times F) \ar{ur}{\canmap} \ar{rrr}{\Lan\, d_{n,A}} & & & \Lan (Q\times V_\II\times Q)^n \ar[bend right=2em]{uul}[swap]{\canmap^n}
\end{tikzcd}
\]

\subparagraph*{Step 2.}
We will show below that $\epsilon$ is a (strong) epimorphism. With this we can conclude the proof as follows. Consider the diagram below, where $p$ is the projection. The part marked $(\star)$ commutes because the outside and all other parts commute either by definition or by naturality. By definition, $\ol{L(A)}^{(n)}$ is the image of the morphism $\canmap\circ \Lan\, m_{L(A)}^{(n)}$ appearing on the left-hand side of the diagram. Since $\Lan\, e_{n,A}$ is an epimorphism (using that the left adjoint $\Lan$ preserves epimorphisms), the morphism  $\canmap\circ \Lan\, m_{L(A)}^{(n)}\circ \Lan\, e_{n,A}$ has the same image, and by commutativity of $(\star)$ and because $\epsilon$ is an epimorphism, this image is precisely $m_{L(\overbar{A})}^{(n)}\colon L^{(n)}(\barA)\monoto V_\FF^n$. Hence $L^{(n)}(\barA)\cong \ol{L(A)}^{(n)}$ as subobjects of $V_\FF^n$, as required.
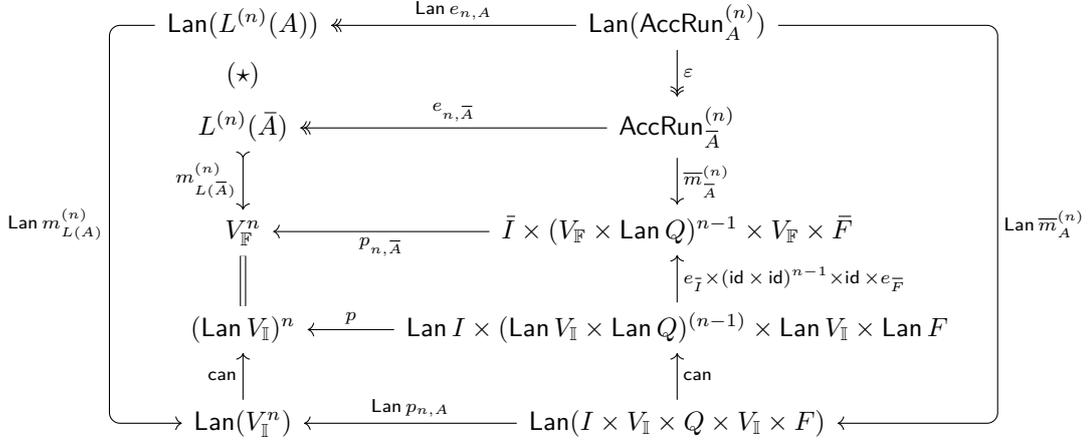
\[
\begin{tikzcd}
\ar[shiftarr={xshift=-5em}]{dddd}[swap]{\Lan\, m_{L(A)}^{(n)}} \ar[phantom]{d}[description]{(\star)} \Lan(L^{(n)}(A)) & \Lan(\Run_A^{(n)}) \ar[two heads]{l}[swap]{\Lan\, e_{n,A}} \ar[two heads]{d}{\epsilon} \ar[shiftarr = {xshift=12em}]{dddd}{\Lan\, \ol{m}_A^{(n)}} \\
L^{(n)}(\barA) \ar[tail]{d}[swap]{m_{L(\overbar{A})}^{(n)}} & \Run_{\overbar{A}}^{(n)} \ar[two heads]{l}[swap]{e_{n,\overbar{A}}} \ar{d}{\ol{m}_{\overbar{A}}^{(n)}} \\
V_\FF^n & \barI\times (V_\FF\times \Lan\, Q)^{n-1}\times V_\FF\times \barF \ar{l}{p_{n,\overbar{A}}} \\
(\Lan\, V_\II)^n \ar[equals]{u} & \Lan\, I\times (\Lan\, V_\II \times \Lan\, Q)^{(n-1)}\times \Lan\, V_\II \times \Lan\, F \ar{l}[swap]{p} \ar{u}[swap]{e_{\overbar{I}}\times (\id\times \id)^{n-1}\times \id \times e_{\overbar{F}}} \\
\Lan(V_\II^n) \ar{u}{\canmap} & \Lan(I\times V_\II\times Q\times V_\II\times F) \ar{u}[swap]{\canmap} \ar{l}[swap]{\Lan\, p_{n,A}}  
\end{tikzcd}
\]
\subparagraph*{Step 3.} It remains to prove the above claim that $\epsilon$ is an epimorphism, i.e.~each component
\[\epsilon_S\colon \Lan(\Run_A^{(n)})S\to \Run_{\barA}^{(n)}S \qquad (S\seq_\f \At)\] is surjective. We first give an explicit description of $\epsilon_S$ using \autoref{rem:left-kan-extension}. For each $T\seq_\f \At$ the set $\Run_A^{(n)}T$ consists of all $T$-supported runs of $A$, that is, all triples $(p_0,a_1,p_1,\ldots, a_n,p_n)$ where $p_0\in IT$, $p_n\in FT$, and $(p_{r-1},a_r,p_r)\in \delta T$ for $r=1,\ldots,n$.
 Then the map $\epsilon_S$ is given by 
\[ [(q_0,a_1,q_1,\ldots,a_n,q_n),\rho] \quad\mapsto\quad ([q_0,\rho],\rho(a_1),[q_1,\rho],\ldots,\rho(a_n),[q_n,\rho]),\]
for all $\rho\colon ET\to S$ in $\FF$ and $(q_0,a_1,q_1,\ldots,a_n,q_n)\in \Run_A^{(n)}T$. 

To prove $\epsilon_S$ surjective, regard the pullback $\Run_{\overbar{A}}^{(n)}S$ as a subset of $((\Lan\, Q)S\times V_\FF S\times (\Lan\, Q)S)^n$, see \autoref{def:naut-acc-lang}. Then every element of $\Run_{\overbar{A}}^{(n)}S$ is a tuple of the form
\begin{equation}\label{eq:tuple} ([q_0,\rho_1], \rho_1(a_1), [q_1,\rho_1],[q_1',\rho_2], \rho_2(a_2), [q_2,\rho_2],  \ldots [q_{n-1}',\rho_n], \rho_n(a_n), [q_{n},\rho_{n}])\end{equation}
where  $[q_0,\rho_1]\in (\Lan\, I)S$, $[q_n,\rho_n]\in (\Lan\, F)S$, $[(q_{r-1}',a_r,q_r),\rho_r]\in (\Lan\, \delta)S$ for $r=1,\ldots, n$ (putting $q_0':=q_0$), and $[q_{r},\rho_{r}]=[q_r',\rho_{r+1}]$ for $r=1,\ldots,n-1$. We will show that we can choose the representatives such that $q_r=q_r'$ for $r=1,\ldots,n-1$ and $\rho_1=\cdots=\rho_n=:\rho$. Then $(q_0,a_1,q_1,\ldots,a_n,q_n)\in \Run_A^{(n)}$ and thus $\epsilon_S$ maps $[(q_0,a_1,q_1,\ldots,a_n,q_n),\rho]$ to \eqref{eq:tuple}.

Suppose that we already have $q_r=q_r'$ for $r=0,\ldots,m-1$ and $\rho_1=\cdots=\rho_m=:\rho$ for some $m<n$. We show that we can modify $q_m'$, $q_{m+1}$, $\rho$ and $\rho_{m+1}$ in such a way that this property also holds after $m+1$ steps. This is achieved by suitable choice of permutations and fresh names, much like in the proof of \autoref{prop:barA-acc-barL}. Recall that we assume the presheaf $Q$ to be strong, that is, $Q=\coprod_{j\in J} \II(S_j,-)$ for some finite set $J$ and $S_j\seq_\f \At$. 
\begin{enumerate}
\item Let $\rho\colon T\to S$ and $\rho_{m+1}\colon T_{m+1}\to S$. 
 Since $[q_m,\rho]=[q_{m}',\rho_{m+1}]$ in $\Lan\, Q$, the states $q_m,q_m'$ must belong to the same summand of $Q$, that is, $q_m\in \II(S_j,T)$ and $q_m'\in \II(S_j,T_{m+1})$ for some $j\in J$. Let $T+T_{m+1}$ denote the disjoint union of $T$ and $T_{m+1}$ with injections $\inl,\inr$. Then we have the following zig-zag in $E{\downarrow}S$:
\[
\begin{tikzcd}[column sep=6em]
 & S & \\
ET \ar{ur}{\rho} \ar{r}{E\inl} & E(T+T_{m+1}) \ar{u}[swap]{[\rho,\rho_{m+1}]} & ET_{m+1} \ar{l}[swap]{E\inr} \ar{ul}[swap]{\rho_{m+1}} 
\end{tikzcd}
\]
Hence, by replacing $T$ and $T_{m+1}$ with $T+T_{m+1}$, we may assume that $T=T_{m+1}$.
\item Since $q_m,q_m'\colon S_j\to T$ are injective maps, there exists a bijection $\pi\colon T\to T$ such that $q_m=\pi\circ q_m'$, witnessing that
\[[(q'_m, a_{m+1},q_{m+1}),\rho_{m+1}]=[(q_m,\pi(a_{m+1}),\pi\circ q_{m+1}),\rho_{m+1}\circ \pi^{-1}]\quad \text{in} \quad \Lan\,\delta,\]
in particular $[q'_m,\rho_{m+1}]=[q_m,\rho_{m+1}\circ \pi^{-1}]$ in $\Lan\, Q$.
Therefore, after replacing $\rho_{m+1}$ by $\rho_{m+1}\circ \pi^{-1}$ and $q_m'$ by $q_m$, we may assume that $q_m=q_m'$.
\item Finally, we consider  the following zig-zag in $E{\downarrow}S$:
\[
\begin{tikzcd}[column sep=6em]
 & S & \\
ET \ar{ur}{\rho} \ar{r}{E\inl} & E(T+T) \ar{u}[swap]{[\rho,\rho_{m+1}]} & ET \ar{l}[swap]{E\iota} \ar{ul}[swap]{\rho_{m+1}} 
\end{tikzcd}
\]
where $\iota\colon T\to T+T$ is the injective map sending every element $a\in q_m[S_j]\seq T$ to $\inl(a)$, and every other element of $T$ to $\inr(a)$. Note that $\inl\circ q_m=\ini\circ q_m$ and that the right-hand triangle commutes: since $[q_m,\rho]=[q_m,\rho_{m+1}]$ we have $\rho\circ q_m = \rho_{m+1}\circ q_m$, hence the maps $\rho,\rho_{m+1}\colon ET\to S$ agree on $q_m[S_j]\seq T$. Therefore, after replacing $\rho$ and $\rho_{m+1}$ with $[\rho,\rho_{m+1}]$, $q_0,q_1,\ldots,q_m$ by $\inl\circ q_0,\ldots,\inl\circ q_m = \iota\circ q_m$, $q_{m+1}$ by $\iota(q_{m+1})$ and $a$ by $\iota(a)$, we can assume that $\rho=\rho_{m+1}$. 
\end{enumerate} 

This concludes the proof of \autoref{prop:presheaf-aut-closure}.

\subsection*{Proof of \autoref{thm:presheaf-automata-vs-nofa}.1}

\begin{rem}
For a presheaf automaton $A=(Q,V_\C,\delta,I,F)$ in $\Set^\C$, $\C\in \{\II,\FF\}$, we write $q\xrightarrow[S]{a} q'$ if $(q,a,q')\in \delta S$ for $S\seq_\f \At$. The accepted word language $\wordlang{L(A)}$ is the set of all $a_1\ldots a_n\in \Ats$ for which there exists an accepting run, i.e.~a sequence of transitions $q_0\xrightarrow[S]{a_1} q_1\xrightarrow[S]{a_2} \cdots \xrightarrow[S]{a_n} q_n$ where $S\seq_\f\At$, $q_0\in IS$ and $q_n\in FS$.
\end{rem}

\begin{rem}\label{rem:i-upper-star}
We recall the left adjoint $I^\star\colon \Set^\II\to \Nom$ of $I_\star$, a.k.a.\ the \emph{sheafification functor}~\cite[Lem.~6.7]{Pitts2013}. For each $P\in \Set^\II$, the nominal set $I^\star P$ is defined as follows:
\begin{itemize}
\item The underlying set of $I^\star P$ is the colimit of the directed diagram
\[ D_P\colon \II_\seq \hookto \II \xto{P} \Set,  \]
where $\II_\seq$ is the poset of finite subsets of $\At$, i.e.~the restriction of $\II$ to  inclusion maps $i_{S,T}\colon S\xto{\seq} T$ for $S\seq T\seq_\f \At$. More explicitly, elements of $I^\star P$ are equivalence classes $[S,x]$ for the equivalence relation on $\coprod_{S\seq_\f \At} PS = \{\,(S,x): S\seq_\f \At,\, x\in PS\,\}$ given by
\[ (S,x)\sim (S',x') \qquad\text{iff}\qquad \exists T\supseteq S,S'.\, Pi_{S,T}(x)=Pi_{S,T'}(x'). \]
\item The group action on $I^\star P$ is given by
\[ \pi\cdot [S,x] = [\pi[S],P\pi|_S(x)] \qquad \text{for $\pi\in \Perm(\At)$ and $[S,x]\in I^\star P$},\]
with $\pi|_S\colon S\to \pi[S]$ denoting the domain-codomain restriction of $\pi\colon \At\to \At$.
\end{itemize}

Since directed colimits in $\Set$ commute with finite limits, the left adjoint $I^\star$ preserves finite limits, in particular products and monomorphisms. In fact, this property  holds in general for sheafification functors~\cite[Thm.~III.5.1]{mm92}.
\end{rem}

\noindent One direction of \autoref{thm:presheaf-automata-vs-nofa}.1 is established by the following lemma. Recall the embedding $I_\star\colon \Nom\to \Set^\II$ (\autoref{sec:toposes}) and its lifting $\barI_\star\colon \NAutfp(\Nom)\to \NAutfp(\Set^\II)$ from \eqref{eq:adjunctions-naut}.

\begin{lemma}\label{lem:nofa-to-setI}
Every NOFA $A$ is word-language equivalent to the $\Set^\II$-automaton $\barI_\star A$.
\end{lemma}

\begin{proof}
By definition of $I_\star$, every accepting run 
\begin{equation}\label{eq:nofa-accrun} 
q_0\xto{a_1} q_1\xto{a_2}\cdots \xto{a_n} q_n 
\end{equation}
of the NOFA $A$ yields the accepting run
\begin{equation}\label{eq:setI-accrun} 
q_0\xrightarrow[S]{a_1} q_1\xrightarrow[S]{a_2}\cdots \xrightarrow[S]{a_n} q_n 
\end{equation}
of the presheaf automaton $\barI_\star A$, where $S\seq_\f \At$ is any set of names containing $a_1,\ldots,a_n$ and supporting $q_0,\ldots,q_n$. Conversely, every accepting run \eqref{eq:setI-accrun} of $\barI_\star A$ yields the accepting run \eqref{eq:nofa-accrun} of $A$.
\end{proof}

Similarly, for the reverse direction we use the lifting ${\barI}^\star\colon \NAutfp(\Set^\II)\to \NAutfp(\Nom)$. 

\begin{lemma}\label{lem:setI-aut-to-nofa}
Every super-finitary nondeterministic $\Set^\II$-automaton $A$ is word-language equivalent to the NOFA ${\barI}^\star{A}$.
\end{lemma}

\begin{proof}
The inclusion $W(L(A))\seq L({\barI}^\star A)$ holds because every accepting run 
\[ q_0\xrightarrow[S]{a_1} q_1\xrightarrow[S]{a_2}\cdots \xrightarrow[S]{a_n} q_n \]
of $A$ yields the accepting run 
\[  
[S,q_0]\xto{a_1} [S,q_1] \xto{a_2}\cdots \xto{a_n} [S,q_n]
\]
of ${\barI}^\star A$. For the proof of $L({\barI}^\star A)\seq W(L(A))$, suppose that $a_1\cdots a_n\in L({\barI}^\star A)$. By definition of ${\barI}^\star A$, an accepting run of $a_1\cdots a_n$ then has the form
\begin{align*}
[S_1,q_0]\xto{a_1} [S_1,q_1]=[S_2,q_1'] \xto{a_2} [S_2,q_2]=[S_3,q_2'] \xto{a_3}\cdots \\ \cdots \xto{a_{n-1}} [S_{n-1},q_{n-1}]=[S_n,q_{n-1}'] \xto{a_n} [S_n,q_n]
\end{align*}
where $q_{r-1}'\xrightarrow[S_r]{a_r} q_r$ in $A$ (putting $q_0':=q_0$) for $r=1,\ldots,n$, and $[S_1,q_0]=[\ol{S}_1,\ol{q}_0]$ for some $\ol{S}_1\seq_\f \At$ and $\ol{q}_0\in I\ol{S}_1$, and $[S_n,q_n]=[\ol{S}_n,\ol{q}_n]$ for some  $\ol{S}_n\seq_\f \At$ and $\ol{q}_n\in F\ol{S}_n$. Replacing the sets $S_1,\ldots,S_n,\ol{S}_1,\ol{S}_n$ by their union we can assume that $S_1=\cdots S_n=\ol{S}_1=\ol{S}_n=:S$. Since $[S,q_{r}]=[S,q_r']$ we know that there exists $T_r\supseteq S$ such that $Pi_{S,T}(q_{r})=Pi_{S,T}(q_{r}')$. Similarly, we have sets $\ol{T}_0,\ol{T}_n\supseteq S$ witnessing that $[S,q_0]=[S,\ol{q}_0]$ and $[S,q_n]=[S,\ol{q}_n]$. Taking the union again, we can assume that $T_1=\cdots T_n=\ol{T}_0=\ol{T}_n=:T$. Hence, after replacing $q_1,\ldots,q_n,\ol{q}_0,\ol{q}_n$ by $Pi_{S,T}(q_1),\ldots Pi_{S,T}(q_n), Pi_{S,T}(\ol{q}_0), Pi_{S,T}(\ol{q}_n)$ we can assume that $q_r=q_r'$ for $r=1,\ldots,n$
and $q_0=\ol{q_0}\in IT$ and $q_n=\ol{q}_n\in FT$. We conclude that
\[ q_0\xrightarrow[T]{a_1} q_1=q_1'\xrightarrow[T]{a_2} q_2=q_2'\xrightarrow[T]{a_3}\cdots q_{n-1}=q_{n-1}'\xrightarrow[T]{a_n} q_n \]
is an accepting run of $A$, proving $L({\barI}^\star A)\seq W(L(A))$.
\end{proof}

\subsection*{Proof of \autoref{thm:presheaf-automata-vs-nofa}.2}
One may argue analogously to \autoref{thm:presheaf-automata-vs-nofa}.1, replacing $\Set^\II$ by $\Set^\FF$ and NOFA by NOFRA (which are equivalent to NOFA for positive word languages by \autoref{thm:nofa-vs-nofra}). We give an alternative argument that relates $\Set^\II$- and $\Set^\FF$-automata in a more direct manner. By \autoref{thm:presheaf-automata-vs-nofa}.1 it suffices to prove the following two lemmas.

\begin{lemma}
Every super-finitary nondeterministic $\Set^\FF$-automaton $A$ accepts a positive word language and is word-language equivalent to the super-finitary nondeterministic $\Set^\II$-automaton ${\barE}^\star A$.
\end{lemma}

\begin{proof}
Let $A$ be a super-finitary nondeterministic $\Set^\FF$-automaton. We first prove that $\wordlang{L(A)}$ is a positive word language. Given $a_1\cdots a_n\in \wordlang{L(A)}$ and a renaming $\rho\colon \At\to\At$, choose $S\seq_\f \At$ such that $a_1\cdots a_n\in L(A)(S)$. Then $\rho(a_1)\cdots \rho(a_n)\in L(A)(\rho[S])$ because $L(A)\seq V_\FF^\star$ is a sub-presheaf. Hence $\rho(a_1)\cdots \rho(a_n)\in \wordlang{L(A)}$, so $\wordlang{L(A)}$ is positive.

Since $E^\star\colon \Set^\FF\to \Set^\II$ is just a forgetful functor, clearly $A$ is word-language equivalent to the $\Set^\II$-automaton ${\barE}^\star A$: both automata have the same accepting runs.
\end{proof}

\begin{rem}\label{rem:pos-closure-presheaf-lang}
For every presheaf language $L\seq V_\II^\star$, the positive closure $\ol{L}\seq V_\FF^\star$ is given at $S\seq_\f \At$ by \[\ol{L}(S)=\{\,\rho^\star(w) : \text{$w\in L(T)$ and $\rho\in \FF(T,S)$ for some $T\seq_\f \At$}\,\}. \] Indeed, this language clearly satisfies the universal property of \autoref{def:pos-closure-presheaf-lang}. In particular, if $W(L)$ is a positive word language, then $W(\ol{L})=W(L)$.
\end{rem}

\begin{lemma}
Every super-finitary nondeterministic $\Set^\II$-automaton $A$ accepting a positive word language is word-language equivalent to the super-finitary nondeterministic $\Set^\FF$-automaton $\ol{\Lan}_E A$. 
\end{lemma}

\begin{proof}
Let $A$ be a super-finitary nondeterministic $\Set^\II$-automaton such that $\wordlang{L(A)}$ is a positive word language. Assuming w.l.o.g.\ that $A$ has a strong presheaf of states, by \autoref{prop:presheaf-aut-closure} the automaton $\ol{\Lan}_E A$ accepts the positive closure $\ol{L(A)}$ of $L(A)$, and \autoref{rem:pos-closure-presheaf-lang} shows that  $\wordlang{L(A)}=\wordlang{\ol{L(A)}}$. Hence $A$ and $\ol{\Lan}_E A$ are word-language equivalent. 
\end{proof}

\end{document}